%% file: nightlightsVsGDP_arxiv.tex
\documentclass[a4paper,12pt]{article}

\PassOptionsToPackage{disable}{todonotes} 
\usepackage{latexstyle}
\usepackage{natbib}
\usepackage{hyperref}
\hypersetup{hidelinks}
\usepackage{subcaption}
\usepackage{lscape}
\usepackage[margin=2cm]{geometry}
\usepackage{eurosym}
\usepackage{bbding}
\usepackage{multirow}
\usepackage{booktabs}
\usepackage{longtable}
\usepackage{caption}
\usepackage{xcolor}
\usepackage{array}
\usepackage{pdflscape}
\usepackage{comment}
\usepackage{amsthm}

\DeclareMathOperator*{\plim}{p-lim}

\usepackage{setspace}
\makeatletter
\g@addto@macro\normalsize{%
  \setlength{\abovedisplayskip}{6pt plus 2pt minus 2pt}%
  \setlength{\belowdisplayskip}{6pt plus 2pt minus 2pt}%
  \setlength{\abovedisplayshortskip}{0pt plus 2pt}%
  \setlength{\belowdisplayshortskip}{4pt plus 2pt minus 2pt}%
}
\makeatother

\graphicspath{{figures/}}

\title{Aggregation Bias in Proxy Measurement:\\ Nighttime Lights and Local Economic Activity}

\author{Davide Fiaschi\thanks{Davide Fiaschi (corresponding author), University of Pisa, Via Ridolfi, 10, 56124 Pisa, Italy (e-mail: davide.fiaschi@unipi.it).}
\and Angela Parenti\thanks{Angela Parenti, University of Pisa, Via Ridolfi, 10, 56124 Pisa, Italy (e-mail: angela.parenti@unipi.it).}
\and Cristiano Ricci\thanks{Cristiano Ricci, University of Pisa, Via Ridolfi, 10, 56124 Pisa, Italy (e-mail: cristiano.ricci@unipi.it).}
}

\date{\today}

\begin{document}

\maketitle

\begin{abstract}
This paper studies when high-resolution signals aggregated to administrative units can recover unobserved local economic activity. We develop a reverse-regression framework for signals generated by activity but used to predict it at coarser spatial supports. The main theorem decomposes predictive elasticity into elementary elasticity, reverse-regression attenuation, and a spatial aggregation term driven by unit size and within-unit dispersion, showing aggregation pulls elasticities toward one. Monte Carlo evidence confirms the decomposition and clarifies transferability conditions. Applications to VIIRS nighttime lights and local GDP or income in Brazil, Italy, the United States, Indonesia, and Kenya support local calibration mainly in richer contexts.
\end{abstract}

\noindent \textit{JEL Classification Numbers}: C23, C52, R12, R15

\noindent \textit{Keywords}: aggregation bias, attenuation bias, spatial aggregation, proxy measurement, out-of-sample validation, Black Marble VIIRS

\bigskip

\noindent {\scriptsize \textit{Acknowledgements}: We thank Piotr Wójcik and the participants at presentations given at ERSA 2023, SEA 2024, SIE 2024, AISRE 2025, IBEO 2026, SEW 2026 for useful comments and suggestions. The usual disclaimer applies. The authors have been supported by the Italian Ministry of University and Research (MIUR), in the framework of PRIN project 2017FKHBA8 001, by the project MAPPE (GRINS, PNRR), and by the University of Pisa, in the framework of the PRA Project PRA\_2022\_86.}


\clearpage

\section{Introduction \label{sec:introduction}}

Many empirical questions require local measures of economic activity at a spatial resolution at which official statistics are unavailable, incomplete, or not comparable across countries. A common response is to use a high-resolution signal, satellite imagery, remotely sensed data, or other spatially detailed measurements, as a proxy for the missing economic variable. This strategy creates a distinctive econometric problem. The signal is typically generated by the underlying economic outcome, but the empirical objective is to use the signal to predict that outcome. Moreover, the signal is observed on a fine spatial grid and then aggregated to administrative units whose size, shape, and internal heterogeneity differ across places. Proxy measurement is therefore not only a question of predictive fit; it is also a question of reverse regression, measurement error, and spatial aggregation.

Nighttime lights (NTL) provide a leading example of this problem. Since the pioneering contribution of \citet{nordhaus2006geography}, an expanding literature has used satellite-recorded luminosity to augment official income measures, study growth at sub- and supranational scales \citep{henderson2012measuring}, audit GDP statistics where institutions are weak \citep{martinez2022much}, measure poverty and wealth where household data are scarce \citep{jean2016combining,abbes2024deepwealth}, and construct high-resolution economic maps. The appeal is clear: NTL are global, spatially detailed, comparable across borders, and repeatedly observed. They appear to provide exactly the local and frequent information that official statistics often lack.

However, ``using lights as a proxy'' is not a single empirical operation. Lights have been used to predict levels of GDP or income, growth rates, rankings, poverty, wealth, and local well-being \citep{henderson2012measuring,galimberti2020forecasting,jean2016combining,abbes2024deepwealth,huber2024economic}, and the light--activity relationship varies with the spatial support (country to grid cell) and with development, sectoral composition, electrification, informality, geography, and time \citep{galimberti2020forecasting,bluhm2022can,gibson2024luminosity,lehnert2023proxying,huber2024economic}. Empirically, estimated GDP--NTL elasticities vary widely across studies, differ across countries, become unstable at large geographies, and change across aggregation levels and NTL sources \citep{galimberti2020forecasting,bluhm2022can,gibson2024luminosity}; in developing settings, where large shares of territory may be unlit, NTL may also carry non-classical measurement error \citep{huber2024economic}. The relevant question is therefore not whether lights are a universally valid proxy, but what they proxy well, at which scale, in which development context, and with how much uncertainty.

This paper argues that part of this instability is the outcome of two biases that are usually discussed separately. First, although economic activity generates lights, empirical proxy regressions typically place lights on the right-hand side. The resulting coefficient is therefore a \textit{predictive reverse-regression elasticity}, not the structural elasticity of light with respect to activity.\footnote{We use ``reverse regression'' descriptively, to signal that the direction of the fitted regression (activity on lights) is the reverse of the data-generating direction (lights emitted by activity). This is distinct from the classical reverse-regression device of \citet{wald1940} and \citet{durbin1954}, in which the regression of the regressor on the outcome is used to \emph{bound} an errors-in-variables coefficient. Here the reverse direction is the object of interest, and the attenuation it induces is a feature to be characterised, not a bound to be exploited.} When lights are an imperfect signal of activity, this predictive elasticity is attenuated by the noise component of the luminosity signal. Second, observed units are spatial aggregates of more elementary locations. If the elementary light--activity relationship is nonlinear, the elasticity estimated on aggregates differs systematically from the elementary elasticity. We characterise these two forces in a single probability-limit decomposition. The decomposition shows that reverse-regression attenuation is strongest at fine scales, whereas spatial aggregation becomes more important at coarser scales and contracts elasticities toward one. Unit elasticity is therefore the only aggregation-invariant benchmark.

This result links the NTL literature to classic problems in spatial statistics and econometrics: the change-of-support and modifiable areal unit problems, under which estimates depend on the spatial support at which variables are observed or aggregated \citep{cressie1996change}; and the long-recognised distinction between measurement error, grouping, and aggregation: suitably constructed grouping estimators can address errors-in-variables bias, while grouping need not itself introduce additional bias \citep{wald1940,durbin1954}; by contrast, aggregation over heterogeneous units can cause aggregate coefficients to differ from the underlying micro-level parameters \citep{theil1954,stoker1993}. In the NTL setting these mechanisms coexist, so the theoretical framework is not specific to lights: any finely observed signal used to measure an economic variable on heterogeneous administrative units inherits the same tension between signal noise, regression direction, and spatial support.

The empirical analysis combines high-resolution Black Marble VIIRS data with official local GDP or income data for Brazil, Italy, the United States, Indonesia, and Kenya over 2012--2019, the common period over which consistently processed nightlight data and comparable local GDP or income series are available. Throughout the paper, ``local economic activity'' denotes the real outcome available at the relevant spatial support. This is local GDP for Brazil, the United States, Kenya, and Indonesia, and taxable personal income for Italian municipalities. We therefore interpret the Italian municipal estimates as income-density estimates, while using NUTS2 and NUTS3 GDP data to assess whether the income-based elasticity is informative about GDP-based activity at coarser spatial supports. The sample includes high-income settings where local accounts are relatively strong, an upper-middle-income economy with large territorial disparities, a lower-middle-income archipelagic economy observed at two nested administrative scales, and a lower-income agricultural economy where lights are most likely to miss part of production.

The results support a conditional, rather than universal, use of NTL. In Italy and the United States, and approximately in Brazil, the density elasticity of economic activity with respect to NTL is close to one at the finest scales, so proportional changes in lights track proportional changes in local GDP or income density and level prediction is feasible once the light-to-activity conversion is locally anchored. Indonesia and Kenya differ: their elasticities are substantially below one and their level conversions much higher. Indonesia, observed at two nested scales (kabupaten/kota and provinces), is especially informative---the estimated elasticity rises toward one under aggregation, as the mechanism predicts when the finest-scale elasticity is below one. A near-unit elasticity is therefore an empirical property of some contexts, not a mechanical feature of the data.

The out-of-sample exercises show both the promise and the limits of the proxy. Cross-sectional, cross-scale, and temporal exercises show that NTL improve prediction over models with only area and time effects, but errors widen in the lowest luminosity ranges and transferring coefficients across scales requires re-anchoring the intercept. We translate the benchmark-country evidence into a pooled finest-level proxying rule for Brazil, Italy, and the United States: a transparent first approximation for developed and upper-middle-income economies with near-unit elasticity, not to be applied mechanically to lower-income, agricultural, informal, or weakly illuminated settings without further calibration.

In summary, the paper makes three contributions. First, it provides an analytical framework for proxy regressions based on aggregated high-resolution signals. The main theorem decomposes the probability limit of the predictive elasticity into the elementary elasticity, a reverse-regression attenuation component, and a spatial aggregation component driven by unit size and within-unit dispersion. Second, it uses Monte Carlo experiments, in which the elementary relationship is known and aggregation is controlled, to assess the finite-sample implications of the theorem. The simulations show when attenuation or aggregation dominates, why density specifications with area controls reduce but do not eliminate aggregation bias, and why slopes and intercepts have different transferability properties. Third, it evaluates these implications empirically using a common NTL source, a common econometric specification, and official local economic accounts for five heterogeneous countries.

The paper is organised as follows. Section~\ref{sec:montecarlo} develops the methodological framework and reports the Monte Carlo evidence. Section~\ref{sec:empirical} implements the framework empirically: it describes the data, estimates the predictive elasticity and level conversion across countries and scales, evaluates out-of-sample transferability, and reports the pooled finest-level calibration. Section~\ref{sec:concludingRemarks} concludes. The Appendix contains the proof, while the Online Appendix contains Monte Carlo tables, robustness checks, additional validation results, and supplementary descriptive material.

\section{Methodological framework: reverse regression and spatial aggregation}
\label{sec:montecarlo}

This section develops the econometric framework that underpins the empirical analysis. The problem is not only that nighttime lights are an imperfect signal of economic activity. It is that the signal is generated at a fine spatial support, observed with noise, aggregated to heterogeneous administrative units, and then used in the reverse direction to predict the economic variable that generated it. The object of interest is therefore a predictive elasticity: the slope of the best linear projection of local economic activity on observed luminosity at a given spatial support. This coefficient is useful for proxy construction and out-of-sample validation, but it should not be interpreted as a structural light-generation parameter or as a causal effect of lights on economic activity.

The framework proceeds in three steps. First, we define an elementary-level light--activity relationship. Second, we derive the aggregate relationship that arises when elementary locations are grouped into larger spatial units. This yields a theorem on the probability-limit decomposition of the predictive elasticity into reverse-regression attenuation and spatial aggregation bias. Third, we use Monte Carlo experiments to assess the finite-sample relevance of the theorem and to motivate the density-with-area specification estimated on the real data in Section~\ref{sec:personalIncome}.

\subsection{Analytical framework}
\label{subsec:mcAnalyticalBenchmark}

The framework separates two possible sources of bias that are otherwise confounded in regressions of economic activity on nighttime lights. The first is a reverse-regression bias: lights are generated by activity with an idiosyncratic component, but the empirical specification uses lights to predict activity. The second is an aggregation bias: sums of nonlinear elementary relationships need not preserve the elementary elasticity once locations are combined into heterogeneous administrative units.

\paragraph{Elementary relationship and aggregation}

Consider $K$ elementary locations grouped into $P$ spatial units $m_p$, with $\abs{m_p}$ denoting the number of elementary locations contained in unit $p$. Nighttime light emissions follow a log-linear function of output with an independent shock:
\begin{equation}
	ntl_i = y_i^{1/\mu} \cdot \exp(-\phi/\mu) \cdot \exp(-\eps_i/\mu),
	\label{eq:generatorNTL}
\end{equation}
where $\mu$ is the true elasticity, $\phi$ is the scale parameter, and $\eps_i$ an idiosyncratic shock.
At the elementary level, write the inverse constant-elasticity relationship as
\begin{equation}
	y_i = e^{\phi} ntl_i^{\mu} e^{\eps_i},
	\label{eq:mcMicroAnalytical}
\end{equation}
where $y_i$ and $ntl_i$ denote output and nighttime lights in elementary location $i$, $\mu>0$ is the elementary elasticity of output with respect to lights, $\phi$ is a common scale parameter, and $\eps_i$ is an idiosyncratic output-to-light shock. We assume that the shocks $\eps_i$ are i.i.d. with $\EE{\eps_i}=0$ and $\mathrm{Var}(\eps_i)=\sigma_\eps^2$, and that $\eps_i$ and $y_i$ are independent. Equivalently, with $\nu_i\equiv\log y_i$, the physical direction is
\begin{equation}
	\log ntl_i=\frac{1}{\mu}(\nu_i-\phi-\eps_i).
	\label{eq:mcLightGeneration}
\end{equation}
This ordering matters. The maintained exogeneity condition is imposed in the light-generation equation, not in the reverse predictive regression. Hence, even when lights are generated from activity with an exogenous idiosyncratic component, regressing activity on observed lights produces a reverse-regression attenuation term whenever lights contain such a component.

Let $Y_p=\sum_{i\in m_p}y_i$, $NTL_p=\sum_{i\in m_p}ntl_i$, and $s_{ip}=ntl_i/NTL_p$ be the within-unit light share. Since $ntl_i=s_{ip}NTL_p$,
\begin{align}
Y_p
&= e^\phi\sum_{i\in m_p}ntl_i^\mu e^{\eps_i}
= e^\phi NTL_p^\mu\sum_{i\in m_p}s_{ip}^\mu e^{\eps_i}, \nonumber\\
\log Y_p
&= \phi + \mu\log NTL_p + \Lambda_p,
\qquad
\Lambda_p\equiv
\log\!\left(\sum_{i\in m_p}s_{ip}^{\mu}e^{\eps_i}\right).
\label{eq:mcExactAggregation}
\end{align}
Aggregation therefore affects the slope of a linear regression only through $\Lambda_p$, a term that depends on within-unit light shares and on the output-to-light shocks. For a cross section with $P$ aggregate units, let $\hat\beta_P$ denote the OLS slope from regressing $\log Y_p$ on $\log NTL_p$ across units $p=1,\ldots,P$ at the chosen aggregation scale. Under the usual large-cross-section regularity conditions, as $P\to\infty$,
\begin{equation}
\plim_{P\to\infty}\,\hat\beta_P = \mu + \frac{\mathrm{Cov}(\Lambda_p,\log NTL_p)}{\mathrm{Var}(\log NTL_p)}.
\label{eq:mcPlim}
\end{equation}
Throughout this subsection, $\mathrm{Cov}(\cdot,\cdot)$ and $\mathrm{Var}(\cdot)$ refer to the corresponding cross-unit population moments at the chosen aggregation scale.
Eq.~\eqref{eq:mcPlim} is the central object: the overall bias is the projection of the aggregate residual $\Lambda_p$ on aggregate lights.

\paragraph{Decomposing attenuation and aggregation}
To compute and interpret this projection, we decompose both objects entering the covariance: the residual $\Lambda_p$ and the regressor $\log NTL_p$. This separates the two sources of bias: the aggregation channel, which comes from summing nonlinear elementary relationships, and the reverse-regression channel, which comes from the fact that lights are generated by activity but used as a predictor of activity.

First, decompose the aggregate residual. Define
\begin{equation}
\Lambda_p^{\mathrm{agg}} \equiv \log\!\left(\sum_{i\in m_p}s_{ip}^{\mu}\right),\qquad
w_{ip} \equiv \frac{s_{ip}^{\mu}}{\sum_{j\in m_p}s_{jp}^{\mu}},\qquad
\Lambda_p^{\mathrm{noise}} \equiv \log\!\left(\sum_{i\in m_p}w_{ip}e^{\eps_i}\right).
\label{eq:mcResidualComponents}
\end{equation}
Since
\begin{equation*}
\sum_{i\in m_p}s_{ip}^{\mu}e^{\eps_i} = \left(\sum_{i\in m_p}s_{ip}^{\mu}\right)\left(\sum_{i\in m_p}w_{ip}e^{\eps_i}\right),
\end{equation*}
we have the exact algebraic split
\begin{equation}
\Lambda_p = \Lambda_p^{\mathrm{agg}} + \Lambda_p^{\mathrm{noise}}.
\label{eq:mcResidualSplit}
\end{equation}
The aggregation component depends only on within-unit light shares and captures the mechanical effect of aggregation. The noise component captures the contribution of output-to-light shocks to the aggregate residual.

Second, the regressor splits, to first order, into a noise-free activity index and an aggregate shock,
\begin{equation}
\log NTL_p = A_p-\frac{1}{\mu}\tilde\eps_p+O_p(\sigma_\eps^2),
\label{eq:mcLogNTLExpansion}
\end{equation}
where $A_p\equiv\log\!\left(\sum_{i\in m_p}\exp\{(\nu_i-\phi)/\mu\}\right)$ is the aggregate light that would be observed in the absence of shocks, $\tilde\eps_p\equiv\sum_{i\in m_p}s_{ip}^{(0)}\eps_i$ is the first-order aggregate shock entering the regressor, and $s_{ip}^{(0)}=\exp\{(\nu_i-\phi)/\mu\}\big/\sum_{j\in m_p}\exp\{(\nu_j-\phi)/\mu\}$ is the noise-free within-unit light share; the expansion is derived in the Appendix. Under the maintained independence between activity and shocks, $A_p$ is orthogonal to $\tilde\eps_p$ to first order.

Combining the exact residual split with the local regressor split, the numerator in~\eqref{eq:mcPlim} can be read, to first order, as four terms:
\begin{align}
\mathrm{Cov}(\Lambda_p,\log NTL_p)
&\approx
\mathrm{Cov}\!\left(
\Lambda_p^{\mathrm{agg}}+
\Lambda_p^{\mathrm{noise}},
A_p-\frac{1}{\mu}\tilde\eps_p
\right) \nonumber \\
&=
\underbrace{\mathrm{Cov}(\Lambda_p^{\mathrm{agg}},A_p)}_{\text{aggregation--activity}}
-
\frac{1}{\mu}
\underbrace{\mathrm{Cov}(\Lambda_p^{\mathrm{agg}},\tilde\eps_p)}_{\text{aggregation--noise}}
\nonumber \\
&\quad+
\underbrace{\mathrm{Cov}(\Lambda_p^{\mathrm{noise}},A_p)}_{\text{noise--activity}}
-
\frac{1}{\mu}
\underbrace{\mathrm{Cov}(\Lambda_p^{\mathrm{noise}},\tilde\eps_p)}_{\text{noise--noise}}.
\label{eq:mcDoubleDecomposition}
\end{align}
The first and fourth terms are the two leading components. To see their approximate magnitudes, define the squared coefficient of variation of elementary-location lights within unit $p$ as
\begin{equation}
\mathrm{CV}_p^2 \equiv \frac{1}{\abs{m_p}}\sum_{i\in m_p}\frac{\left(ntl_i-\overline{ntl}_p\right)^2}{\overline{ntl}_p^{\,2}},\qquad \overline{ntl}_p=\frac{NTL_p}{\abs{m_p}}.
\label{eq:mcCVDef}
\end{equation}
A second-order expansion of $\Lambda_p^{\mathrm{agg}}$ around equal within-unit shares gives
\begin{equation}
\Lambda_p^{\mathrm{agg}} \approx (1-\mu)\log\abs{m_p} + \frac{\mu(\mu-1)}{2}\mathrm{CV}_p^2,
\label{eq:mcAggExpansion}
\end{equation}
so that, normalised by $\mathrm{Var}(\log NTL_p)$, the aggregation--activity term is approximately $(1-\mu)\left[\delta_P-\tfrac{\mu}{2}\rho_P\right]$, with the finite-$P$ projection coefficients
\begin{equation}
\delta_P \equiv \frac{\mathrm{Cov}(\log\abs{m_p},\log NTL_p)}{\mathrm{Var}(\log NTL_p)},\qquad
\rho_P \equiv \frac{\mathrm{Cov}(\mathrm{CV}_p^2,\log NTL_p)}{\mathrm{Var}(\log NTL_p)}.
\label{eq:mcDeltaRhoDef}
\end{equation}
Here $\delta_P$ measures the unit-size channel: larger aggregates mechanically contain more elementary locations and therefore more total lights. The coefficient $\rho_P$ measures how within-unit light dispersion covaries with aggregate lights.

The fourth term is the reverse-regression channel. With the weighted aggregate shock entering the residual, $\bar\eps_{p,w}\equiv\sum_{i\in m_p}w_{ip}\eps_i$, define the cross-sectional attenuation coefficient at the chosen aggregation scale as
\begin{equation}
\kappa_P \equiv \frac{\mathrm{Cov}(\bar\eps_{p,w},\tilde\eps_p)}{\mathrm{Var}(\mu A_p)+\mathrm{Var}(\tilde\eps_p)}.
\label{eq:mcReliability}
\end{equation}
The numerator pairs the shock aggregator entering the residual, $\bar\eps_{p,w}$, with that entering aggregate lights, $\tilde\eps_p$; the denominator is, to first order, $\mathrm{Var}(\log NTL_p)\approx\mu^{-2}\left[\mathrm{Var}(\mu A_p)+\mathrm{Var}(\tilde\eps_p)\right]$. At the elementary level $\bar\eps_{p,w}=\tilde\eps_p=\eps_i$ and $\kappa_P=\sigma_\eps^2/(\sigma_\nu^2+\sigma_\eps^2)$ (with $\sigma_\nu^2\equiv\mathrm{Var}(\nu_i)$), the classical noise-to-total-variance ratio. Since shocks enter the residual through $\bar\eps_{p,w}$ but the regressor through $-\tilde\eps_p/\mu$, the normalised noise--noise term in~\eqref{eq:mcDoubleDecomposition} equals $-\mu\kappa_P$ to first order (see the Appendix): the two powers of $\mu$ from normalising $\mathrm{Var}(\log NTL_p)$ cancel the $1/\mu$ of the regressor expansion, so attenuation is proportional to $\mu$, exactly the errors-in-variables slope $\mu(1-\kappa_P)$.

The remaining two terms in~\eqref{eq:mcDoubleDecomposition}, aggregation--noise and noise--activity, are cross terms. Under the maintained independence between activity and output-to-light shocks, they vanish to first order, up to higher-order interactions between within-unit dispersion and shocks. Hence:
\begin{equation}
	\plim_{P \to \infty}\,\hat\beta_P
	\approx
	\mu \;\underbrace{-\mu\kappa_P}_{\text{attenuation bias}}
	+ \underbrace{(1-\mu)\left[\delta_P-\tfrac{\mu}{2}\rho_P\right]}_{\text{aggregation bias}}.
	\label{eq:intuitionBias}
\end{equation}
The reasoning is based on a local approximation: it combines an exact decomposition of $\Lambda_p$ with a first-order expansion of $\log NTL_p$ and a second-order expansion of $\Lambda_p^{\mathrm{agg}}$.

The following result formalises this approximation.
\begin{teo}[Aggregation and attenuation bias]
	\label{teo:aggregation}
Assume the following conditions.
\begin{enumerate}
\item The light-generation Eq.~\eqref{eq:mcLightGeneration} holds with $\mu>0$, $\EE{\eps_i}=0$, finite second moments, and $\nu_i\perp\eps_j$ for all $i,j$.
\item Along the sequence of cross sections indexed by $P$, each observed at the chosen aggregation scale, the vector $\left(\log NTL_p,\ A_p,\ \tilde\eps_p,\ \bar\eps_{p,w},\ \log\abs{m_p},\ \mathrm{CV}_p^2\right)$ has finite second moments, with $\mathrm{Var}(\log NTL_p)>0$ and $\mathrm{Var}(\mu A_p)+\mathrm{Var}(\tilde\eps_p)>0$. The finite-$P$ projection coefficients satisfy
\begin{equation*}
\kappa_P\stackrel{P \to \infty}{\longrightarrow}\kappa_\infty,\qquad \delta_P\stackrel{P\to\infty}{\longrightarrow}\delta_\infty,\qquad \rho_P\stackrel{P\to\infty}{\longrightarrow}\rho_\infty
\end{equation*}
for some $\kappa_\infty,\delta_\infty,\rho_\infty\in\RR$.
\item Expressing the within-unit light shares as a local perturbation of the equal-share benchmark:
\begin{equation}
s_{ip}(\chi)=\abs{m_p}^{-1}\left(1+\chi d_{ip}\right),\qquad \sum_{i\in m_p}d_{ip}=0,
\label{eq:mcLocalShares}
\end{equation}
where $\chi\geq0$ is the scalar non-uniformity parameter, the directions $d_{ip}$ are uniformly bounded along the cross-section sequence:
\begin{equation*}
\sup_P\max_{1\le p\le P}\max_{i\in m_p}\abs{d_{ip}}<\infty.
\end{equation*}
\end{enumerate}
Then
\begin{equation}
\plim_{P \to \infty}\,\hat\beta_P = \mu  \underbrace{-\mu\kappa_\infty}_{\text{attenuation bias}} + \underbrace{(1-\mu)\left[\delta_\infty-\tfrac{\mu}{2}\rho_\infty\right]}_{\text{aggregation bias}} + R(\chi,\sigma_\eps),
\label{eq:mcBiasApprox}
\end{equation}
where
\begin{equation*}
R(\chi,\sigma_\eps) = O(\chi^3) + O(\sigma_\eps^2) + O(\sigma_\eps\chi).
\end{equation*}
\end{teo}
\begin{proof}
See the Appendix.
\end{proof}

Before proceeding, we make some remarks on Theorem~\ref{teo:aggregation}.
\begin{oss}
The finite-second-moment condition ensures that the projection in~\eqref{eq:mcPlim} and the finite-$P$ coefficients $\kappa_P$, $\delta_P$, and $\rho_P$ are well defined. The limits $\kappa_\infty$, $\delta_\infty$, and $\rho_\infty$ in which the theorem is stated are not elementary-level quantities but limits of projection coefficients at the chosen aggregation scale, and therefore depend on how elementary locations are grouped along the cross-section sequence. In this respect, what convergence in Assumption~2 requires is minimal: because $\kappa_P$, $\delta_P$, and $\rho_P$ are continuous functions of the cross-unit second moments of $\left(\log NTL_p,\ A_p,\ \tilde\eps_p,\ \bar\eps_{p,w},\ \log\abs{m_p},\ \mathrm{CV}_p^2\right)$, their limits exist as soon as the empirical second-moment matrix of this vector converges to a finite limit that is non-degenerate on the denominators $\mathrm{Var}(\log NTL_p)$ and $\mathrm{Var}(\mu A_p)+\mathrm{Var}(\tilde\eps_p)$. This condition is met, for instance, by an independent and exchangeable partition scheme, in which the unit sizes $\abs{m_p}$, the within-unit elementary log-activities $\nu_i$, and the shocks $\eps_i$ are drawn i.i.d.\ from fixed distributions with $\eps_i$ independent of activity: each summary vector is then a function of its own unit's draws alone, hence i.i.d.\ across units, and $\kappa_P$, $\delta_P$, $\rho_P$ converge to deterministic constants by the law of large numbers. This is exactly the data-generating scheme of the Monte Carlo experiments of Section~\ref{sec:MonteCarloExperiments}, which therefore lie within the scope of the theorem. However, these are sufficient, not necessary, conditions: convergence may also hold under weaker, spatially dependent designs.
\end{oss}
\begin{oss}
The parameter $\chi$ is the local non-uniformity parameter: all $O(\chi^k)$ terms are understood along the path $s_{ip}(\chi)=\abs{m_p}^{-1}(1+\chi d_{ip})$ as $\chi\to0$.
The equal-share benchmark is $\chi=0$; larger values of $\chi$ correspond to more unequal within-unit light shares. For a fixed aggregation scale and a fixed collection of directions $\{d_{ip}\}$, $\chi\to0$ is the local Taylor-expansion parameter, whereas the probability limit in~\eqref{eq:mcPlim} is taken over the cross-section of aggregate units at that scale. For $\chi$ small enough, $s_{ip}(\chi)>0$ for all $i,p$. To see how this parameter maps into the within-unit coefficient of variation, note that $ntl_i=s_{ip}(\chi)NTL_p$ and $\overline{ntl}_p=NTL_p/\abs{m_p}$. Hence
\begin{equation*}
\frac{ntl_i-\overline{ntl}_p}{\overline{ntl}_p} = \frac{\abs{m_p}ntl_i}{NTL_p}-1 = \abs{m_p}s_{ip}(\chi)-1 = \chi d_{ip},
\end{equation*}
and therefore
\begin{equation*}
\mathrm{CV}_p^2(\chi) = \frac{1}{\abs{m_p}}\sum_{i\in m_p}\left(\frac{ntl_i-\overline{ntl}_p}{\overline{ntl}_p}\right)^2 = \chi^2\frac{1}{\abs{m_p}}\sum_{i\in m_p}d_{ip}^2.
\end{equation*}

The boundedness of $d_{ip}$ is a restriction on relative, not absolute, light intensity. Since $\abs{m_p}s_{ip}-1=ntl_i/\overline{ntl}_p-1$, it rules out elementary locations whose lights become arbitrarily large relative to the mean of their aggregate unit. In particular, the condition is automatically satisfied if elementary lights are uniformly bounded away from zero and infinity along the aggregation sequence.
\end{oss}
\begin{oss} The decomposition rests on Assumption~1, which treats the output-to-light shock $\eps_i$ as classical: mean zero and independent of activity $\nu_i$. This is what makes the two cross terms in~\eqref{eq:mcDoubleDecomposition} (aggregation--noise and noise--activity) vanish to first order and delivers the clean separation between attenuation and aggregation. If instead $\eps_i$ is correlated with activity the noise--activity covariance $\mathrm{Cov}(\Lambda_p^{\mathrm{noise}},A_p)$ no longer vanishes, and an additional term of order $\mathrm{Cov}(\eps_i,\nu_i)$ enters~\eqref{eq:mcBiasApprox}.\footnote{This is the non-classical case emphasised by \citet{huber2024economic}, and the natural description of settings where lights systematically miss non-electrified, agricultural, or informal output, or are top-coded and bloomed in bright cores.} The attenuation and aggregation channels then cease to be additively separable, and the elementary elasticity is no longer point-identified from the predictive slope and the observable projection coefficients alone. We view this as the main limitation of the framework and return to it empirically: it is precisely in the low-luminosity, low-income settings (Kenya and, to a lesser extent, Indonesia) where the classical benchmark is least tenable, and where the diagnostic inversion of Section~\ref{sec:betaElasticity} breaks down in the direction predicted by a systematic, activity-correlated undercounting of output by lights. Modelling this non-classical component structurally is left to future work.
\end{oss}

\paragraph{Implications for the empirical analysis}

Equation~\eqref{eq:mcBiasApprox} has three implications for the empirical analysis. First, in the noiseless environment (i.e.\ no attenuation bias), $\mu=1$ is the unique elasticity that is invariant to arbitrary aggregation patterns: when $\mu=1$, both the unit-size term and the within-unit-dispersion term disappear. For $\mu\neq1$, aggregation generally changes the estimated elasticity, with the sign governed by the bracket $\delta-\mu\rho/2$. In the leading case in which the unit-size channel is positive and dominates within-unit dispersion, aggregation pushes estimates upward when $\mu<1$ and downward when $\mu>1$, contracting them toward unity.

Second, idiosyncratic output-to-light shocks attenuate the reverse predictive slope, but this attenuation weakens as aggregation averages out the shock component. Expressing variables as densities, or including area controls in levels, targets the first-order unit-size channel $\delta$; what remains is the reverse-regression attenuation and the smaller within-unit-dispersion component. The Monte Carlo exercises below verify these mechanisms in controlled environments.

Third, under the maintained assumption that reverse-regression attenuation is negligible at the relevant scale ($\kappa_\infty\approx0$), Eq.~\eqref{eq:mcBiasApprox} can be inverted to recover an aggregation-corrected estimate of the elementary elasticity $\mu$ from the observed predictive slope $\hat\beta_P$ and empirical estimates of $\delta_P$ and $\rho_P$ computed from the observed NTL distribution. We use this inversion only as a diagnostic, because it relies on the strong assumption that the unobserved output-to-light shock is small.

\subsection{Monte Carlo experiments}
\label{sec:MonteCarloExperiments}

In each simulation run, we generate data for $K = 100{,}000$ elementary locations, indexed by $i$ as in the analytical benchmark. Elementary output is drawn from a log-normal distribution, $y_i=\exp(\nu_i)$ with $\nu_i\sim\mathcal{N}(1,1)$. Nighttime light emissions follow Eq.~\eqref{eq:generatorNTL}, with an independent shock $\eps_i\sim\mathcal{N}(0,\sigma_\eps^2)$ and scale parameter $\phi=11.5$. Elementary locations are then aggregated into $P\in\{100,1{,}000,10{,}000,20{,}000,50{,}000,100{,}000\}$ spatial units with heterogeneous elementary-location counts. For each aggregation level we estimate the regression covered by Theorem~\ref{teo:aggregation}:
\begin{equation}
	\log Y _{p} = \alpha + \beta \log  NTL_{p} + \epsilon_{p}
	\label{eq:lightOnIncomeMC}
\end{equation}
where $Y_{p}$ and $NTL_{p}$ are total output and total lights of spatial unit $p$. Sections~\ref{sec:mcSmallSample}--\ref{sec:mcAttenuation} study this specification, whose bias is fully characterised by Equation~\eqref{eq:mcBiasApprox}; Section~\ref{sec:mcToEmpirics} then explains why the correction implied by the theorem is not directly operational, introduces unit areas, and evaluates the empirical specifications used in Section~\ref{sec:empirical}.

\subsubsection{Finite-sample verification of the decomposition}
\label{sec:mcSmallSample}

The first Monte Carlo exercise verifies the finite-sample accuracy of Eq.~\eqref{eq:mcBiasApprox}. For each combination of $\mu\in\{0.7,1,1.3\}$ and $\sigma_\eps\in\{0.1,0.3,0.5,1\}$, and for each aggregation ratio $K/P$ with $P\in\{100,1{,}000,10{,}000,20{,}000,50{,}000,100{,}000\}$, we regress $\log Y_p$ on $\log NTL_p$ and compare the simulated bias $\hat\beta_P-\mu$ with the approximation $-\mu\kappa_P+(1-\mu)\left[\delta_P-\tfrac{\mu}{2}\rho_P\right]$. The projection coefficients $\delta_P$ and $\rho_P$ are computed from the simulated lights using Eq.~\eqref{eq:mcDeltaRhoDef}, while $\kappa_P$ is computed from Eq.~\eqref{eq:mcReliability} using the true shocks $\eps_i$ and the true $\mu$, so that the only discrepancy left between the two sides is the remainder $R(\chi,\sigma_\eps)$ of Theorem~\ref{teo:aggregation}.

Figure~\ref{fig:mcBiasApproxVerification} summarises the comparison; the full set of estimates is reported in Section~\ref{sec:tabelleMonteCarlo} of the Online Appendix (Table~\ref{tab:mcBiasApproxVerification}). Three findings stand out. First, at the elementary scale ($K/P=1$) the approximation is essentially exact for every $(\mu,\sigma_\eps)$: with no aggregation, $\delta_P=\rho_P=0$ and the bias reduces to pure attenuation, $-\mu\kappa_P$. Second, for moderate noise ($\sigma_\eps\le0.5$) and moderate aggregation ($K/P\le10$), the remainder is uniformly small (below $0.04$ in absolute value, and typically below $0.02$), so Eq.~\eqref{eq:mcBiasApprox} tracks both the attenuation and the aggregation components of the bias remarkably well. Third, the approximation deteriorates exactly where the theorem says it should: for $\sigma_\eps=1$ the remainder grows with the $O(\sigma_\eps^2)$ term, and for $\mu=0.7$ at coarse scales ($K/P\ge100$) the within-unit light shares become highly unequal --- the variance of elementary log lights is $\sigma_\nu^2/\mu^2$, so low $\mu$ inflates within-unit dispersion --- and the local expansion in $\chi$ is no longer accurate, with $\rho_P$ taking large values. In this region the approximation still gets the sign and the qualitative contraction toward one right, but no longer the magnitude. The working range of Eq.~\eqref{eq:mcBiasApprox} therefore covers the empirically relevant configurations in which the calibrated light signal is not dominated by measurement noise and administrative aggregation is not at the extreme coarse limit of the simulation design.
\begin{figure}[htbp]
\centering
\includegraphics[width=0.7\textwidth]{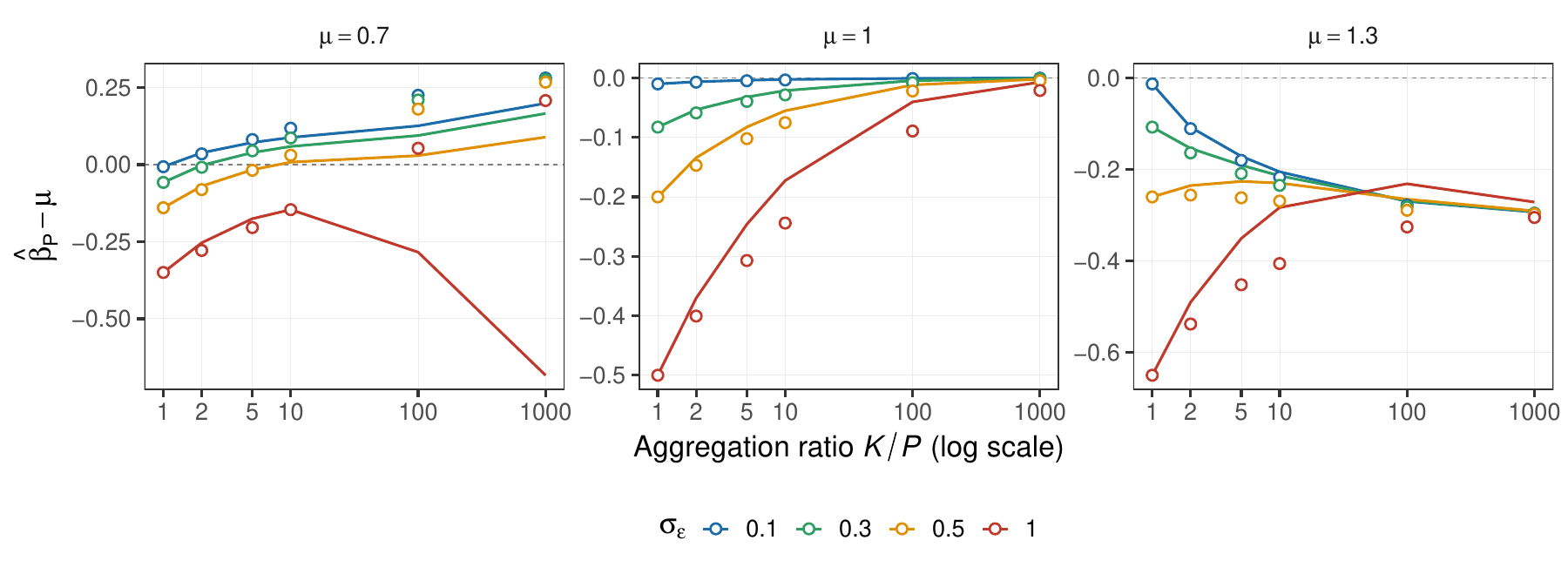}
\caption{Monte Carlo verification of Eq.~\eqref{eq:mcBiasApprox}. Points: simulated bias $\hat\beta_P-\mu$ from regressing $\log Y_p$ on $\log NTL_p$ (averages over replications, $K=100{,}000$); lines: the approximation $-\mu\kappa_P+(1-\mu)\left[\delta_P-\tfrac{\mu}{2}\rho_P\right]$ (with $\kappa_P$ from Eq.~\eqref{eq:mcReliability} using true shocks, $\delta_P,\rho_P$ from Eq.~\eqref{eq:mcDeltaRhoDef}). 
}
\label{fig:mcBiasApproxVerification}
\end{figure}

Figures~\ref{fig:mcBiasAttenuation} and~\ref{fig:mcBiasAggregation} decompose the approximation into the two components of Theorem~\ref{teo:aggregation}, computed separately in the same experiments. The decomposition shows that the two channels operate at opposite ends of the aggregation range. The attenuation component $-\mu\kappa_P$ is maximal at the elementary scale and vanishes monotonically as $K/P$ grows: within-unit averaging of the output-to-light shocks raises the signal content of aggregate lights, with a profile that is common across $\mu$ up to the scale factor. The aggregation component $(1-\mu)\left[\delta_P-\tfrac{\mu}{2}\rho_P\right]$ behaves in the opposite way: it is zero at the elementary scale, grows in magnitude with $K/P$, and is identically zero at every scale when $\mu=1$ --- the knife-edge property that makes the unit elasticity invariant to aggregation. Its sign is positive for $\mu<1$ and negative for $\mu>1$, generating the contraction toward one.
\begin{figure}[htbp]
	\centering
	\includegraphics[width=0.8\textwidth]{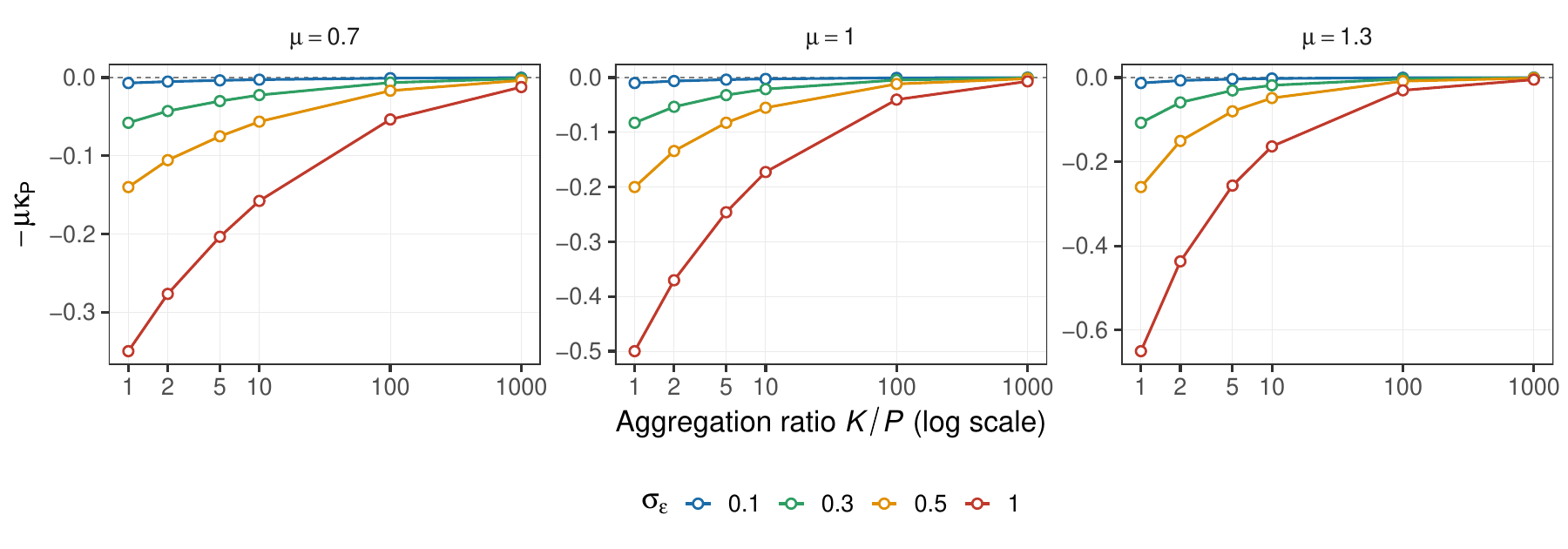}
	\caption{The attenuation component of Eq.~\eqref{eq:mcBiasApprox}, $-\mu\kappa_P$, with $\kappa_P$ from Eq.~\eqref{eq:mcReliability} (true shocks $\eps_i$; averages over Monte Carlo replications, $K=100{,}000$). Panels: true elasticity $\mu$; colours: shock s.d.\ $\sigma_\eps$; horizontal axis: aggregation ratio $K/P$ (log scale). 
	}
	\label{fig:mcBiasAttenuation}
\end{figure}
The crossing of the two profiles also explains the non-monotone total bias visible in Figure~\ref{fig:mcBiasApproxVerification} for high noise levels: for $\mu=0.7$ and $\sigma_\eps=1$ the bias moves from $-0.35$ at the elementary scale, where it is pure attenuation, to $+0.21$ at the coarsest scale, where the aggregation channel dominates.

\begin{figure}[htbp]
\centering
\includegraphics[width=0.8\textwidth]{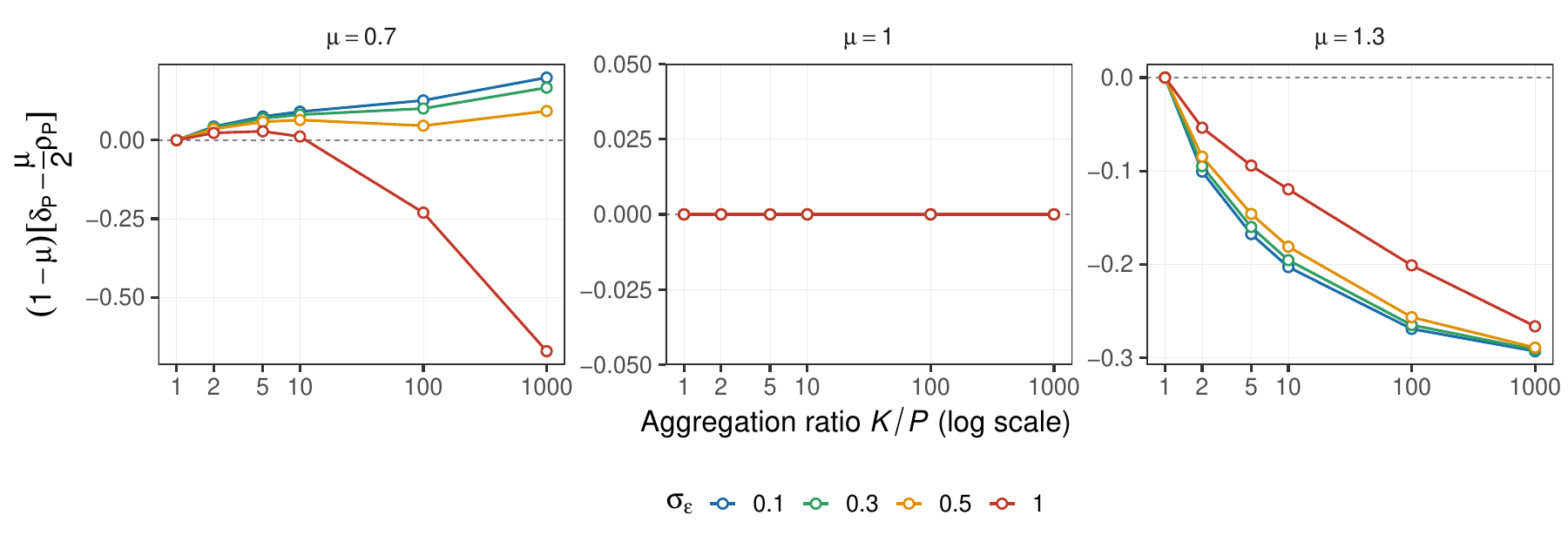}
\caption{The aggregation component of Eq.~\eqref{eq:mcBiasApprox}, $(1-\mu)\left[\delta_P-\tfrac{\mu}{2}\rho_P\right]$, with $\delta_P,\rho_P$ from Eq.~\eqref{eq:mcDeltaRhoDef} (averages over Monte Carlo replications, $K=100{,}000$). Panels: true elasticity $\mu$; colours: shock s.d.\ $\sigma_\eps$; horizontal axis: $K/P$ (log scale). 
}
\label{fig:mcBiasAggregation}
\end{figure}

\subsubsection{Aggregation bias in totals regressions}
\label{sec:mcAggregation}

We set $\sigma_\eps=0.1$ (i.e.\ $\sigma^2_\eps=0.01$), so that attenuation is small relative to aggregation. At the elementary scale, $\kappa_P=\sigma^2_\eps/(\sigma^2_\nu+\sigma^2_\eps)\approx0.01$, implying an attenuation bias of about one percent of $\mu$ in Eq.~\eqref{eq:mcBiasApprox}.

Figure~\ref{fig:mcBetaAlphaTotal} in the Online Appendix reports the Monte Carlo estimates of Model~\eqref{eq:lightOnIncomeMC} across aggregation levels. It confirms the analytical benchmark: aggregation leaves the elasticity unchanged at $\mu=1$, pushes estimates upward when $\mu<1$, and pushes them downward when $\mu>1$. Thus aggregation contracts the estimated elasticity toward one rather than merely adding noise. The bottom row shows the implication for level recovery: the scale parameter $\alpha$ is well recovered at the unit-elastic benchmark at every scale, but when $\mu\neq1$ the size-driven distortion of $\hat\beta$ is absorbed by the intercept, which drifts away from $\phi$ as aggregation coarsens. Level prediction from lights therefore requires an anchored light-to-output conversion whenever the elasticity is away from one. Full estimates and confidence bands are reported in Section~\ref{sec:tabelleMonteCarlo} of the Online Appendix (Table~\ref{tab:mcBetaAlphaTotalEstimates}).


\subsubsection{Reverse-regression attenuation}
\label{sec:mcAttenuation}

Figure~\ref{fig:mcNoiseNaive} in the Online Appendix varies the idiosyncratic-shock standard deviation $\sigma_\eps$ to put attenuation and aggregation at work jointly in Model~\eqref{eq:lightOnIncomeMC}. At the finest scale, the slope is $\hat\beta \approx \lambda\mu$, with $\lambda \equiv \sigma_\nu^2/(\sigma_\nu^2+\sigma_\eps^2)$; noisier light therefore implies stronger reverse-regression attenuation. As aggregation increases, the slope is pulled toward one from whatever attenuated starting point: the attenuation component dies out while the aggregation component takes over, exactly the crossing of the two profiles documented in Figures~\ref{fig:mcBiasAttenuation} and~\ref{fig:mcBiasAggregation}. Section~\ref{sec:tabelleMonteCarlo} in the Online Appendix reports the underlying numbers (Table~\ref{tab:mcNoise}).


\subsubsection{From the theorem to the empirical specifications}
\label{sec:mcToEmpirics}

Theorem~\ref{teo:aggregation} characterises the bias of Model~\eqref{eq:lightOnIncomeMC}, but the correction it implies is not directly operational. Undoing the bias requires the attenuation coefficient $\kappa_P$, which depends on the unobserved output-to-light shocks $\eps_i$, and the projection coefficients $\delta_P$ and $\rho_P$. Without auxiliary information, the theorem describes the bias but cannot remove it. Section~\ref{sec:betaElasticity} shows what can be achieved when partial auxiliary information is available (grid cells as elementary units and $\kappa_\infty\approx0$ as a maintained assumption), together with the failures of that route at coarse scales; here we take the complementary, fully operational route: specifications that mitigate the observable aggregation channels by construction. In particular, the empirical specification is written in densities and includes an area control. Under the log-log specification a totals regression with an area control and a density regression with an area control are equivalent; we adopt the density form for comparability across aggregation scales.

We therefore introduce unit land areas $S_p$ into the simulations. Area is the observable counterpart of the unit-size channel: if elementary locations had equal area, $S_p$ would be proportional to $\abs{m_p}$, and an area control would target the $\delta_P$ term of Theorem~\ref{teo:aggregation}. To allow for imperfect proportionality, elementary areas are drawn separately. In the baseline, $S_p$ is positively correlated with the elementary-location count, as in the natural case where larger administrative units contain more land and more light-emitting locations. We also consider independent and negatively correlated areas. The negative-correlation case is a useful stress test for settings with compact, dense urban units and large, sparsely populated rural units. Two alternatives to Model~\eqref{eq:lightOnIncomeMC} are considered:
\begin{equation}
	\log (Y_{p}/S_p) = \alpha + \beta \log  (NTL_{p}/S_p) + \gamma \log S_p + \epsilon_{p},
	\label{eq:mcModelDensityArea}
\end{equation}
which expresses variables as spatial densities and adds an explicit area control, and
\begin{equation}
	\log (Y_{p}/S_p) = \alpha + \beta \log  (NTL_{p}/S_p + \ell_0) + \gamma \log S_p + \epsilon_{p},
	\label{eq:mcModelNLM}
\end{equation}
which further allows a luminosity shift $\ell_0$ estimated from the data. These are the Monte Carlo counterparts of the Baseline Model and of the Nonlinear Model of Section~\ref{sec:personalIncome}.

\paragraph{Bias in the estimates}
\begin{figure}[!htbp]
	\centering
	\includegraphics[width=0.7\textwidth]{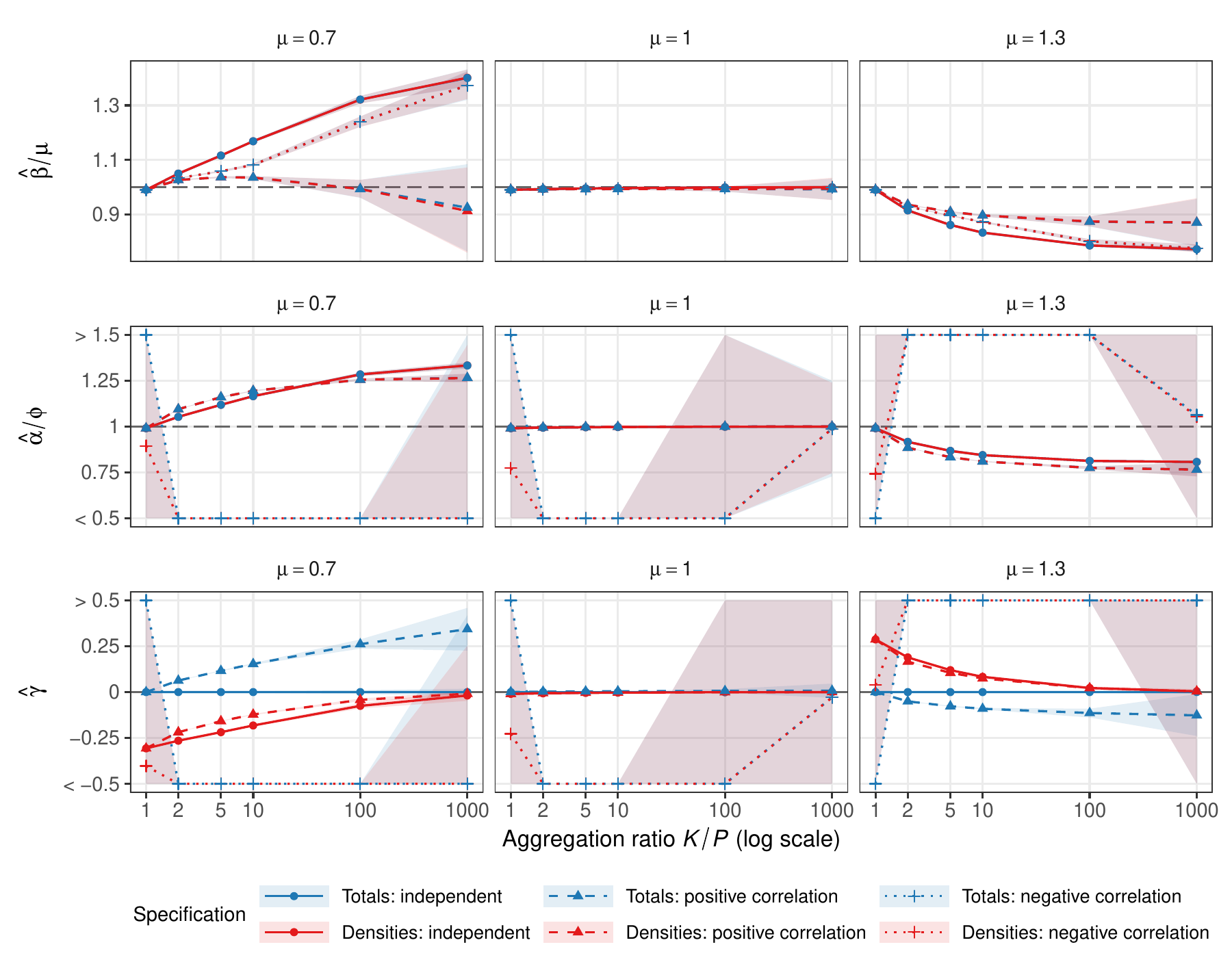}
	\caption{Monte Carlo estimates of $\hat\beta$, $\hat\alpha$, and $\hat\gamma$ under the area-control specifications, $\sigma_\eps=0.1$ (Eq.~\eqref{eq:mcModelDensityArea} for densities and its totals counterpart; median across 1000 replications). Rows: $\hat\beta/\mu$, $\hat\alpha/\phi$, and $\hat\gamma$ (reference $0$); columns: $\mu=0.7,1,1.3$. Lines distinguish totals vs densities and whether unit size is independent of, positively, or negatively correlated with the elementary-location count. For readability the $\hat\alpha/\phi$ and $\hat\gamma$ rows are clipped to $[0.5,1.5]$ and $[-0.5,0.5]$ ($<$/$>$ mark clipped values); full estimates in Section~\ref{sec:tabelleMonteCarlo} of the Online Appendix.}
	\label{fig:mcBetaGammaSixPanel}
\end{figure}
Figure~\ref{fig:mcBetaGammaSixPanel} reports the in-sample evidence for the area-control specification, with variables in totals or densities and under the three area--size correlation regimes. The first row shows that the density specification with an area control keeps $\hat\beta$ close to the true elasticity at every aggregation level, removing the contraction toward one that afflicts the totals regression. The second row reports recovery of the scale parameter, and the third row reports the residual role of unit size after density normalisation. The area control is therefore substantive: when area tracks the aggregation support, it prevents the NTL coefficient from absorbing variation due to observational-unit size.
Adding noise to the picture, under aggregation the naive $\hat\beta/\mu$ is pulled toward the fixed point $1/\mu$, while the density-and-area ratio stays close to the reliability ratio $\lambda$, with residual drift only at the highest noise levels. Densities and the area control therefore largely neutralise the first-order aggregation channel, but not the attenuation channel; attenuation fades only as aggregation averages the shocks out
(Figure~\ref{fig:mcNoise} in the Online Appendix).

\paragraph{Out-of-sample prediction}
Finally, the specifications are compared on the criterion that matters for using lights as a proxy: out-of-sample prediction. Within each replication, units are split into five folds. Each specification is estimated on four folds and used to predict the fifth. Performance is measured by the root mean squared out-of-fold error, which is comparable across specifications because predicting $\log Y_p$ or $\log(Y_p/S_p)$ is equivalent once area is observed. To isolate the pure aggregation channel, the independent and negative area--size regimes are evaluated at low noise ($\sigma_\eps=0.1$).
\input{tables/table_montecarlo_oos.tex}
Table~\ref{tab:mcOOS} reports the results. Three facts stand out. First, the density-with-area specification weakly dominates the totals regression in every regime. It coincides with the totals regression wherever area carries no size information, and predicts systematically better wherever it does. The gains are largest at coarse scales and high noise under positive area--size correlation: for example, at $\mu=0.7$, $\sigma_\eps=0.5$, and $K/P=1000$, the RMSE falls from $0.135$ to $0.082$. The same logic applies under negative correlation: at $\mu=0.7$ and $K/P=10$, the RMSE falls from $0.183$ to $0.150$. The area control exploits the area--size relation in either direction.

Second, density normalisation alone is not sufficient. The density specification without the area control is dominated at fine scales whenever normalisation injects area noise into the regression. Removing the unit-size channel requires the control, not just the transformation. Conversely, a totals specification with the area control is not reported because it is observationally equivalent to the density specification with the area control: the two share the same regressor span and their out-of-fold predictions differ exactly by the observed $\log S_p$, so the errors coincide. What matters econometrically is controlling for area; densities remain preferable for interpretation and comparability.

Third, the Nonlinear Model performs exactly like the Baseline Model in this DGP\@. Since there is no luminosity floor, the estimated shift $\ell_0$ correctly collapses toward zero, so allowing for it costs nothing out of sample. Its value emerges on real data at low luminosity, as documented in Section~\ref{sec:nonparametric}. At $\mu=1$ all specifications are equivalent, the knife-edge case in which aggregation is harmless. These results provide the controlled-environment justification for adopting the density-with-area specification as the Baseline Model of Section~\ref{sec:empirical}, with the Nonlinear Model as a flexible extension.

\section{Empirical application}\label{sec:empirical}

This section takes the aggregation framework to the data. We first describe the harmonised light and economic panel data and construct empirical counterparts of the aggregation terms in Theorem~\ref{teo:aggregation} (Section~\ref{sec:dataset}). We then estimate the predictive light--activity relationship at each available spatial support (Section~\ref{sec:personalIncome}). The final two subsections assess transferability: out-of-sample validation tests whether estimated relationships predict held-out units, scales, and years (Section~\ref{sec:validation}), and a pooled finest-level calibration evaluates how far a common proxying rule can be pushed in benchmark economies (Section~\ref{sec:pooledCalibration}).

\subsection{Data and empirical aggregation diagnostics}
\label{sec:dataset}

We combine a common high-resolution signal, nighttime lights (NTL), with official local measures of GDP or personal income for Brazil, Italy, the United States, Indonesia, and Kenya. The sample period is 2012--2019, the common window over which consistently processed VIIRS data and the required local economic series are available, and it stops before the COVID-19 disruption. The countries are deliberately heterogeneous: the United States and Italy are high-income benchmarks, Brazil is an upper-middle-income economy with large territorial disparities, Indonesia is a lower-middle-income archipelagic economy, and Kenya is a lower-income agricultural economy where lights are more likely to miss non-electrified, agricultural, or informal production.

NTL are sourced from NASA's \textit{Black Marble} project,\footnote{\url{https://blackmarble.gsfc.nasa.gov/}.} based on VIIRS observations from the Suomi NPP satellite at approximately 500-meter resolution. Black Marble corrects for atmospheric effects, terrain, vegetation, snow cover, lunar illumination, stray light, saturation, and calibration issues \citep{roman2018nasa}. These corrections are important for the proxy problem because measurement noise contributes directly to reverse-regression attenuation. Raw NTL maps are reported in Section~\ref{app:literatureTable} of the Online Appendix.

For every country, year, and spatial support, we construct total NTL by summing grid-cell radiance within administrative boundaries. Brazil is observed at municipality, microregion, mesoregion, and federal-state levels; Italy at municipality, local labour area, NUTS3, and NUTS2 levels; the United States at county, commuting-zone, and state levels; Kenya at county level; and Indonesia at kabupaten/kota and province levels. For compact notation, we use MUN for municipality, MICRO for microregion, MESO for mesoregion, STATE for federal state or U.S. state, LLA for local labour area, CZ for commuting zone, COUNTY for county, KAB/KOTA for kabupaten/kota, and PROV for province. NUTS2 and NUTS3 denote European Union statistical regions at levels 2 and 3.  Figure~\ref{fig:dataLowerAdmLevel} in the Online Appendix reports NTL density and economic density at the lowest available administrative level for each country.

Local economic activity is measured from official sources. Brazilian municipal GDP comes from IBGE and is aggregated to microregions, mesoregions, and federal states.\footnote{\url{https://www.ibge.gov.br/estatisticas/economicas/contas-nacionais/9088-produto-interno-bruto-dos-municipios.html}.} U.S. county- and state-level GDP comes from the BEA regional accounts.\footnote{\url{https://apps.bea.gov/iTable/}.} Italian municipal personal income is derived from Agenzia delle Entrate tax declarations, while NUTS2 and NUTS3 GDP comes from ARDECO.\footnote{\url{https://www1.finanze.gov.it/finanze/pagina_dichiarazioni/public/dichiarazioni.php}. For ARDECO, see \url{https://knowledge4policy.ec.europa.eu/territorial/ardeco-database_en\#database}.} Kenya's Gross County Product comes from the Kenya National Bureau of Statistics,\footnote{\url{https://www.knbs.or.ke/all-reports/}.} and Indonesian real GDP at province and kabupaten/kota levels comes from INDO-DAPOER.\footnote{Indonesia Database for Policy and Economic Research, World Bank DataBank: \url{https://databank.worldbank.org/source/indonesia-database-for-policy-and-economic-research}.} The effective time coverage and unit counts reported below refer to the merged analytical panel after sample restrictions.\footnote{Alaska, Hawaii, and U.S. territories are excluded from county-level maps and regressions because their extreme geography and low population density would distort county-level estimates. State-level regressions retain aggregate state observations, including the District of Columbia where available, as coarse aggregation benchmarks.}

Monetary variables are first expressed in 2015 prices using CPI or GDP deflators and are then converted into international dollars using World Bank PPP conversion factors. This harmonisation supports cross-country comparison of levels, while the main elasticity estimates rely on within-country variation. Where available, we also apply within-country cost-of-living adjustments: NUTS2 PPP data for Italy from \cite{cannari2009differenze}, and state-level PPP data for the United States from the BEA\@. No comparable intra-country PPP adjustment is available for Brazil, Kenya, or Indonesia. The baseline regressions use density measures: NTL, GDP, income, and area are all harmonised at the same administrative support, with economic activity and NTL divided by square kilometres.

As defined in Section~\ref{sec:introduction}, ``local economic activity'' is the real outcome available at the relevant spatial support: GDP for Brazil, the United States, Kenya, and Indonesia, and taxable personal income for Italian municipalities. Only for Italy does the finest-level variable differ from GDP; we therefore read the Italian municipal estimates as income-density estimates and use the NUTS2 and NUTS3 GDP data to check whether they carry over to GDP-based activity at coarser supports (Section~\ref{sec:betaElasticity}).

Table~\ref{tab:statDescr} reports descriptive statistics for the analytical samples. The table also provides a first comparison of economic density across countries after PPP conversion.\footnote{As a consistency check, we aggregate local monetary values to the country-year level and compare them with World Bank GDP in current international dollars (NY.GDP.MKTP.PP.CD). The ratios are close to one for Brazil and Indonesia, about 0.94 for the contiguous United States, consistent with the exclusion of Alaska, Hawaii, and U.S. territories, and about 0.91 for Kenya, where Gross County Product does not exactly match national GDP\@. For Italy, the municipal variable is personal income rather than GDP; aggregating municipalities gives about 42--45 percent of GDP in current international dollars, i.e.\ roughly one half.} Within-country dispersion is large, especially at fine scales, and cross-country comparisons should be interpreted with caution where subnational cost-of-living adjustments are unavailable. Administrative labels are also not directly comparable across countries: municipalities in Brazil and Italy, counties in the United States and Kenya, and kabupaten/kota in Indonesia are different spatial and institutional objects.
\input{tables/table_descriptiveStats.tex}

We next construct empirical diagnostics for the aggregation terms in Theorem~\ref{teo:aggregation}. In particular, Table~\ref{tab:deltaRhoEmpirics} reports three empirical counterparts of the theory. For this diagnostic only, the elementary units are lit Black Marble grid cells. The variables below continue to aggregate total radiance within administrative boundaries; the lit-cell restriction is used only to compute within-unit light shares and coefficients of variation in a way that corresponds to the positive-light elementary locations of Section~\ref{subsec:mcAnalyticalBenchmark}. Including fully dark cells would make $\mathrm{CV}_p^2$ dominated by large, sparsely lit units and would turn the diagnostic into a measure of darkness rather than of dispersion among light-emitting locations. The cell-level variance of log radiance, $\sigma_n^2$, is a proxy for the variance of elementary (log) local economic activity under the light-generation equation when measurement error is negligible. The projection coefficient $\delta_P$ captures the unit-size channel: luminous aggregate units tend to contain more lit elementary locations. The projection coefficient $\rho_P$ captures the within-unit dispersion channel: conditional on luminosity, units with more unequal light shares can have different aggregate elasticities.

Two regularities are important for the estimates below. First, $\delta_P$ is positive and sizeable in all five countries, already at the finest available administrative levels. The unit-size channel is therefore empirically relevant even before moving to coarse regions. Second, $\rho_P$ is heterogeneous. It is moderate in Brazil, Italy, the United States, and Kenya, but large in Indonesia, where archipelagic units combine bright urban cores with extensive dim areas; it can also turn negative at coarse scales, as in Brazilian mesoregions, where large Amazonian units combine low total light with high dispersion. Because both coefficients enter the aggregation component multiplied by $(1-\mu)$, they generate little bias when the elementary elasticity is close to one, but they contract estimated elasticities toward unity when $\mu$ differs from one. Section~\ref{sec:betaElasticity} uses these diagnostics to interpret the country and scale patterns of $\hat\beta$.

\input{tables/table_deltaRhoEmpirics.tex}

\subsection{Estimates}
\label{sec:personalIncome}

We now estimate the predictive NTL--activity relationship across all available countries and administrative levels. The baseline specification is
\begin{equation}
\log\text{(EconomicActivity\_km2)}_{it} = \alpha_{t} + \beta \log\text{(NTL\_km2)}_{it} + \gamma \log\text{(Km2)}_i + \epsilon_{it}
\label{eq:lightOnIncome}
\end{equation}
where $\text{EconomicActivity\_km2}$ is GDP density per square kilometre for Brazil, the United States, Indonesia, and Kenya, and personal-income density per square kilometre for Italy. $\text{NTL\_km2}$ is nighttime lights density per square kilometre; $\alpha_t$ is a time-varying intercept, $\beta$ is the NTL elasticity, and $\gamma$ captures the residual role of observational-unit area.

The three parameters of this \textit{Baseline Model} (BM) map directly into the proxy problem. The intercept $\alpha_t$ anchors levels. The elasticity $\beta$ governs proportional variation: when it is close to one, NTL density approximates relative differences and growth rates in economic density. The area coefficient $\gamma$ implements the aggregation correction motivated by Section~\ref{subsec:mcAnalyticalBenchmark}: if unit size is correlated with luminosity and economic density, omitting area can contaminate the NTL elasticity. The nonlinear specification introduced in Section \ref{sec:nonparametric} relaxes the constant-elasticity restriction and lets the NTL elasticity vary with luminosity.

\subsubsection{Predictive elasticity}\label{sec:betaElasticity}
Figure~\ref{fig:estimatedElasticity_countries} reports $\hat\beta$ for all countries and administrative levels. Confidence intervals use standard errors clustered by spatial unit; Section~\ref{app:completeResultLinear} of the Online Appendix reports the corresponding BM estimates and higher-geography clustering checks where available.

The estimates separate the benchmark economies from the lower-income cases. Italy and the United States are close to unit elasticity at the finest available scales; Brazil is slightly below one at the municipal level and moves above one at coarser levels. In these more luminous settings, proportional variation in NTL density is therefore a useful approximation to proportional variation in GDP or income density. This is a statement about economic magnitudes, not about failing to reject $\beta=1$: with large panels, small deviations can be statistically precise.

Kenya and Indonesia instead have elasticities below one, especially Kenyan counties and Indonesian kabupaten/kota. This is consistent with low luminosity, agriculture, informality, and weaker electrification, although data quality may also contribute: subnational accounts in these settings are noisier and sometimes partly imputed. Indonesia is particularly informative because it is observed at two scales. Moving from kabupaten/kota to provinces, the elasticity rises toward one, matching the aggregation benchmark for $\mu<1$; Section~\ref{app:indonesiaAggregation} in the Online Appendix finds the same pattern under random aggregation of kabupaten/kota. A near-unit estimate at a coarse scale can therefore mask a lower fine-scale elasticity.

That the near-unit estimates in Brazil, Italy, and the United States are already present at the finest levels, where the density specification removes the first-order aggregation channel (Eq.~\eqref{eq:mcBiasApprox}), indicates they are not merely an aggregation artefact, which would appear mainly at coarse scales, but a property of the local NTL--activity relationship, with residual deviations measuring the approximation error of a one-for-one proxy.
The inference is robust to clustering by unit, higher-geography clustering, and a wild cluster bootstrap (Section~\ref{app:completeResultLinear} in the Online Appendix). The equivalence tests in Table~\ref{tab:beta_tost} sharpen the magnitude claim: under conservative clustering, Italian municipalities and U.S. counties are equivalent to one within $\pm 10\%$, Brazilian municipalities within $\pm 20\%$, while Kenya and Indonesia reject equivalence at either tolerance.
For Italy, where the finest-level variable is income rather than GDP, we re-estimate the model at NUTS3 and NUTS2 using GDP\@. The elasticities are very similar, especially at NUTS3, supporting the use of personal income at the municipal level while keeping the GDP interpretation for coarser Italian scales.

Finally, high in-sample fit should not be overinterpreted: both luminosity and economic activity density are strongly related to population density, so fit can be high even when the proxy is imperfect. Section~\ref{sec:validation} therefore assesses proxy quality out of sample.

\begin{figure}[!htbp]
\centering
\includegraphics[width=1\textwidth]{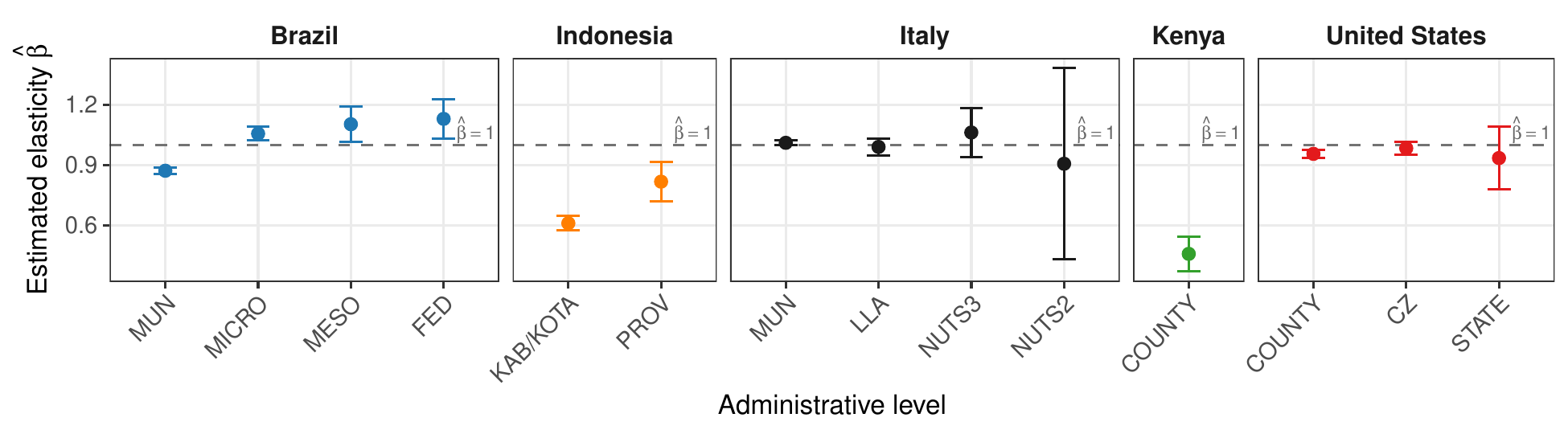}
\caption{Estimated elasticity ($\beta$) of Eq.~\eqref{eq:lightOnIncome} by country and administrative aggregation. Error bars: 95\% confidence intervals, standard errors clustered by spatial unit. Administrative levels are shown as acronyms on the horizontal axis, and panel width is proportional to the number of administrative levels of each country.
}
\label{fig:estimatedElasticity_countries}
\end{figure}

The near-unit fine-scale estimates also receive indirect support from the aggregation formula. Under negligible measurement error ($\kappa_\infty\approx0$), Theorem~\ref{teo:aggregation} can be inverted to recover an implied elementary, cell-level elasticity $\mu$ from the naive totals slope and the empirical projection coefficients in Table~\ref{tab:deltaRhoEmpirics}. This is a diagnostic exercise rather than an estimator we rely on for calibration, because it imposes the strong assumption that reverse-regression attenuation is negligible. The inversion (Table~\ref{tab:muCorrected} in Section~\ref{app:muCorrected} of the Online Appendix) delivers $\hat\mu\approx1.00$ for Italian municipalities, $0.91$ for U.S. counties, $0.76$ for Brazilian municipalities, and $0.70$ for Indonesian kabupaten/kota. The ordering matches the BM estimates. Kenya is informative through failure: its naive slope lies below $\delta_P$, a configuration no non-negative $\mu$ can generate under $\kappa_\infty\approx0$, signalling that measurement error is non-negligible where lights miss much agricultural and informal activity.

\subsubsection{Level conversion}\label{sec:timeincomelights}
The time-varying intercept $\alpha_t$ is the level conversion between NTL density and economic density. It absorbs common changes in the baseline light--activity relationship over time and is the key parameter when NTL are used to predict levels rather than proportional differences.

Figure~\ref{fig:EstimatedAlpha_countries} reports the estimated intercepts. At the finest available levels, Brazil, Italy, and the United States have broadly comparable conversions, whereas Indonesia and especially Kenya require much higher intercepts. In logs, these are large multiplicative differences: conditional on NTL density and area, a unit of observed luminosity corresponds to substantially more measured GDP density in the lower-luminosity economies. This is consistent with a larger agricultural, informal, rural, or otherwise low-emission component of activity \citep{henderson2018global,jean2016combining,abbes2024deepwealth}.
One qualification: at the finest scale Italy is measured with personal income, not GDP (municipal income is about 42--45\% of GDP; Section~\ref{sec:dataset}). The like-for-like NUTS3 and NUTS2 checks suggest the mismatch shifts the Italian intercept by less than the raw income-to-GDP ratio would imply; still, the Brazil--Italy--United States level comparison should be read as approximate, and the pooled calibration (Section~\ref{sec:pooledCalibration}) handles it through country-year fixed effects.

Within countries, $\alpha_t$ is fairly stable over time. Brazil shows a mild decline, especially between 2012 and 2017, while Italy and the United States show small increases at finer levels. This stability indicates that annual recalibration is not the main difficulty; cross-country level conversion is.
\begin{figure}[hbtp]
\centering
\includegraphics[width=1\textwidth]{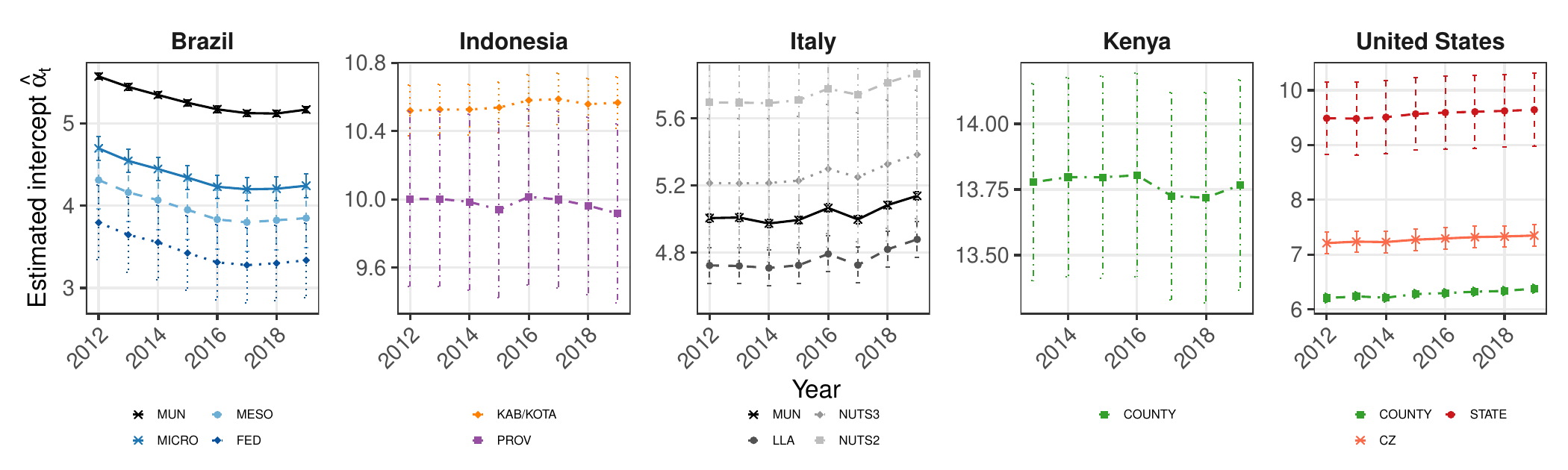}
\caption{Estimated time-varying intercept ($\alpha_t$) of Eq.~\eqref{eq:lightOnIncome} by country and administrative aggregation. GDP/income densities are measured in PPP-adjusted international dollars for comparability. Error bars represent 95\% confidence intervals. Each country panel has its own legend, placed below the panel, identifying the administrative levels (acronyms).
}
\label{fig:EstimatedAlpha_countries}
\end{figure}

\subsubsection{Area correction}\label{sec:gammaArea}
Figure~\ref{fig:estimatedGamma_countries} in the Online Appendix reports the area coefficient $\gamma$, which captures the residual association between observational-unit size and economic density after conditioning on NTL density and time effects. In the density specification, the first-order unit-size term is removed by construction, so $\gamma$ should be close to zero unless area remains correlated with residual aggregation effects.

This is what we find for Brazil, Italy, and the United States, where $\hat\gamma$ is generally small. Kenya and Indonesia instead display larger and statistically different area coefficients, suggesting that unit size and within-unit heterogeneity still matter after density normalisation. The area correction is therefore marginal in the benchmark economies but material in lower-luminosity and geographically uneven settings.


Taken together, $\beta$, $\alpha_t$, and $\gamma$ point to two calibration problems in poorer and lower-luminosity economies. Levels require a country-specific conversion factor, and growth rates require an elasticity that may be well below one. This complements \citet{gibson2024luminosity}: NTL remain informative in developing settings, but they need explicit calibration for both levels and proportional changes.

\subsubsection{Nonlinear elasticity at low luminosity\label{sec:nonparametric}}
We finally allow the NTL elasticity to vary with luminosity using the shifted-log Nonlinear Model (NLM):
\begin{equation}
\log\text{(EconomicActivity\_km2)}_{it} = \alpha_{t} + \beta \log\!\left(\text{NTL\_km2}_{it}+\ell_0\right) + \gamma \log\text{(Km2)}_i + \epsilon_{it},
\label{eq:lightOnIncomeQuadratic}
\end{equation}
where the luminosity shift $\ell_0$ governs curvature and satisfies $\text{NTL\_km2}_{it}+\ell_0>0$ throughout the sample. The implied local elasticity is $\beta\,\text{NTL\_km2}_{it}/(\text{NTL\_km2}_{it}+\ell_0)$. When $\ell_0>0$, elasticity is lower in the darkest units and rises toward $\beta$ as luminosity increases; when $\ell_0<0$, the opposite pattern can arise over the admissible range.

Figure~\ref{fig:elasticityByNTL} in the Online Appendix reports the implied local elasticity $\varepsilon(x)=\hat\beta\,x/(x+\hat\ell_0)$, with estimates in Section~\ref{app:kappaEstimates} of the Online Appendix. In Brazil, Italy, and the United States, elasticity is close to one around typical luminosity levels but falls in the darkest units. Kenya and Indonesian kabupaten/kota remain below one over most of their support. The nonlinear model therefore refines the linear evidence by locating the low-luminosity regime in which the constant-elasticity approximation becomes fragile.


\subsection{Out-of-sample validation and transferability}
\label{sec:validation}

For a proxy, the relevant test is not in-sample fit but prediction for units, scales, or years not used in estimation. We compare the Baseline Model (BM), the Nonlinear Model (NLM), and a deliberately \textit{No-Information Model} (NIM), which omits NTL density and keeps only the time-varying intercept and $\log \text{Km2}$. The gain over NIM measures the information in lights beyond unit size and time effects, while the comparison between BM and NLM tests whether low-luminosity curvature improves prediction.

The validation separates two uses of the proxy. For growth rates or relative comparisons, the key object is the slope: with $\beta$ close to one, light density moves nearly proportionally with economic density. For levels, the slope is not enough: the intercept must anchor the light-to-activity conversion. The exercises below therefore ask how well slopes transfer across units, scales, and years, and how much accuracy is lost when the level correction is unavailable or observed only at an aggregate scale. This design is feasible because the model uses common time-varying intercepts rather than unit fixed effects, so held-out units can still be predicted.


We consider three complementary exercises.

\textit{(A) Cross-sectional (leave-regions-out).} We hold out groups of spatial units, estimate the model on the remaining units, and predict all years of the held-out units. At the coarsest levels, where the number of units is small, this becomes leave-one-unit-out. This tests transferability across units.

\textit{(B) Cross-scale (downscaling).} We estimate the model at a coarse administrative level and use the coefficients, together with fine-level NTL, to predict fine-level income or GDP density. This tests whether aggregate estimates extrapolate to finer spatial resolutions.

\textit{(C) Temporal (leave-years-out).} We hold out one year and predict it using slopes estimated on the remaining years and the average estimated intercept $\bar{\hat\alpha}$. The error captures any drift in the light-to-activity conversion.

In every exercise, the forecast error is
\begin{equation}
\begin{split}
e_{it}^{(m)}
&= \log\!\left(\text{EconomicActivity\_km2}\right)_{it} \\
&\quad - \left[\, \hat\alpha^{(-\ell,m)}_t
+ \hat\beta^{(-\ell,m)} f_m\!\left(\text{NTL\_km2}\right)_{it}
+ \hat\gamma^{(-\ell,m)}\log\text{(Km2)}_i \,\right],
\end{split}
\label{eq:cvError}
\end{equation}
where the three forecasting functions are
\[
f_{\mathrm{BM}}(x)=\log x, \qquad
f_{\mathrm{NLM}}(x)=\log(x+\ell_0), \qquad
f_{\mathrm{NIM}}(x)=0,
\]
and $(-\ell,m)$ denotes coefficients of model $m$ estimated without held-out fold $\ell$. Predictive performance is summarised by the out-of-sample $R^2$,
\begin{equation}
\begin{split}
R^2_{\text{oos}}
= 1 -
\frac{\sum_t\sum_i e_{it}^2}
{\sum_t\sum_i\left[\log(\text{EconomicActivity\_km2})_{it}
- \overline{\log(\text{EconomicActivity\_km2})}\right]^2},
\end{split}
\label{eq:R2oos}
\end{equation}
by root mean square error (RMSE), and by the distribution of $e_{it}$. The increment over NIM, $\Delta R^2_{\text{oos}}$, is the information content of NTL beyond area and time effects; the gain from NLM is measured relative to BM. 

Table~\ref{tab:cvCrossSectional} reports the main cross-sectional validation exercise by country and administrative level.
\begin{table}[!htbp]
	\centering
	\small
	\caption{Cross-sectional (leave-regions-out) out-of-sample performance by country and administrative level: $R^2_{\text{oos}}$ for the Baseline (BM), Nonlinear (NLM), and No-Information (NIM) models, the NLM-vs-BM gain, $\Delta R^2_{\text{oos}}$ relative to NIM, RMSE, and $L_4=(N^{-1}\sum_i |e_i|^4)^{1/4}$.
	}
	\label{tab:cvCrossSectional}
	\resizebox{\textwidth}{!}{%
		\input{tables/table_cvCrossSectional.tex}
	}
\end{table}
 Figure~\ref{fig:cvErrorDistribution} shows the distribution of temporal forecast errors conditional on NTL density at the finest available level: municipalities for Brazil and Italy, kabupaten/kota for Indonesia, and counties for Kenya and the United States. The complete cross-scale and temporal validation tables are reported in Section~\ref{app:additionalValidation} of the Online Appendix.
The results answer the three validation questions. First, NTL add substantial predictive content over NIM. Across countries and scales, the NTL models explain much more out-of-sample variation than area and time effects alone, so predictive power is not merely a mechanical artefact of unit size. The nonlinear correction, however, changes average fit only modestly; its role is mainly in the tails of the error distribution.
\begin{figure}[htbp]
	\centering
	\includegraphics[width=0.7\textwidth]{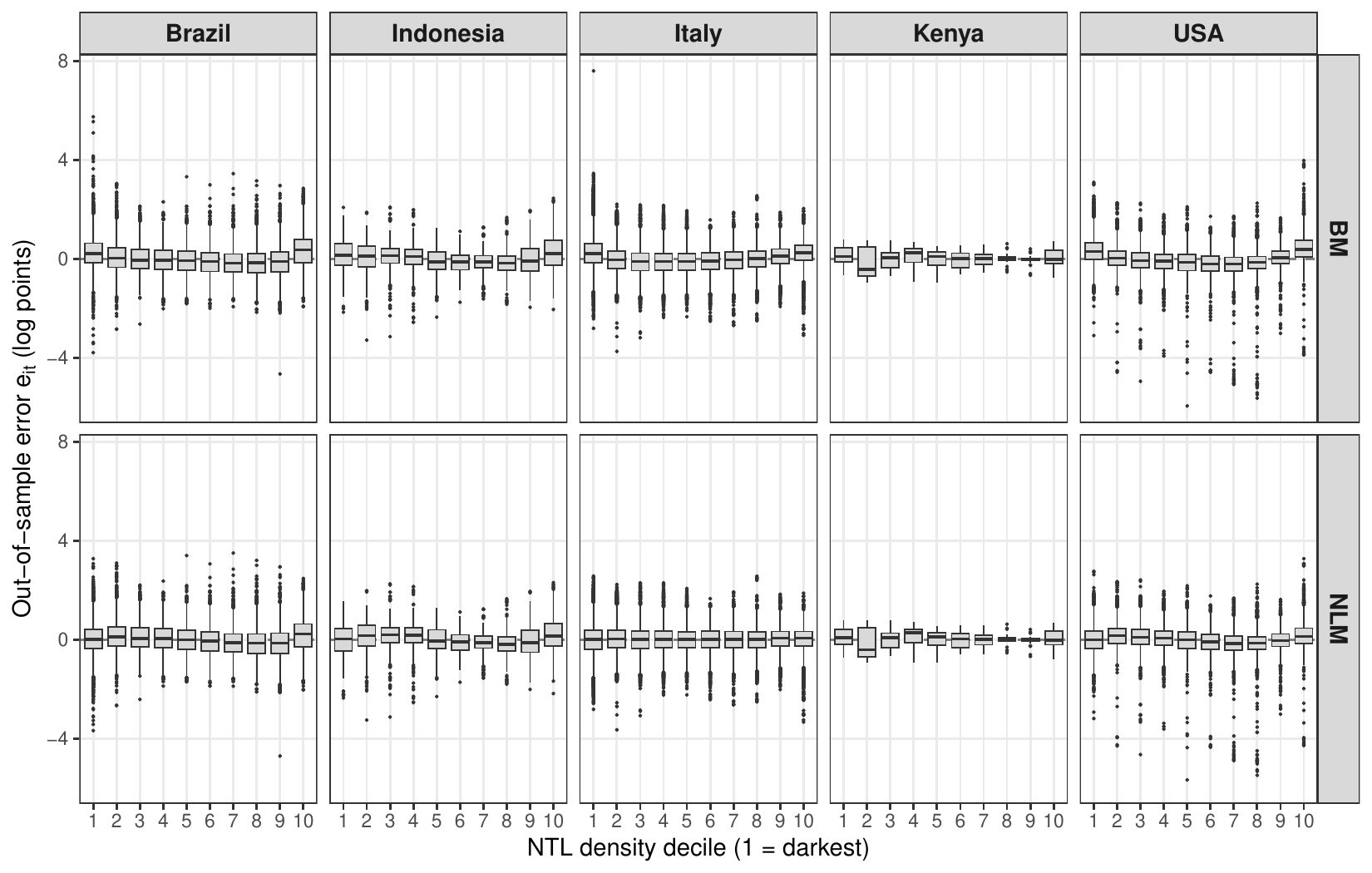}
	\caption{Distribution of temporal out-of-sample forecast errors $e_{it}$ at the finest level, by country and model, conditional on NTL density deciles (decile 1 = darkest). Finest level: MUN (Brazil, Italy), KAB/KOTA (Indonesia), COUNTY (Kenya, United States). Kenya is a single-scale low-income benchmark
		. Boxes show the interquartile range; the BM--NLM comparison highlights the low-luminosity regime where the nonlinear term absorbs part of the error dispersion.}
	\label{fig:cvErrorDistribution}
\end{figure}
Second, coarse-to-fine prediction supports downscaling, but only for the slope structure. Once the systematic level shift between scales is removed, coefficients estimated at coarse levels recover the cross-sectional pattern of finer-level density well. What does not transfer automatically is the intercept: the naive level prediction becomes biased as the distance between estimation and target scales grows. Aggregate estimates therefore extrapolate to finer resolutions in relative terms, but level prediction requires re-anchoring with some local information.

Third, the error distribution shows where the proxy degrades. Forecast errors widen in the lowest NTL deciles and are tighter elsewhere (Figure~\ref{fig:cvErrorDistribution}). This is why Table~\ref{tab:cvCrossSectional} reports $L_4$ in addition to RMSE: the nonlinear correction is most visible in tail errors rather than average fit. These patterns mirror the curvature documented in Section~\ref{sec:nonparametric}: the proxy is least reliable in the darkest, least dense units, where the constant-elasticity approximation begins to fail. Overall, NTL are informative out of sample, but accuracy degrades predictably at low luminosity and downscaling requires re-anchoring the level.

\subsection{Pooled finest-level calibration}
\label{sec:pooledCalibration}

The validation exercises show that NTL contain out-of-sample information and identify where the proxy becomes fragile. We now ask a more operational question: can the benchmark-country evidence be summarised by a pooled finest-level calibration usable when local economic data are unavailable? This is not a new identification design, but a calibration exercise that turns the preceding estimates into a proxying rule and measures the cost of different level anchors.

We restrict the benchmark pool to Brazil, Italy, and the United States, where finest-level elasticities are close enough to unity to justify a common slope (Table~\ref{tab:beta_tost}). The pooled sample uses municipalities for Brazil and Italy and counties for the United States. Italy contributes income density rather than GDP density. Including it in a GDP-based pool is therefore an empirical decision, not an assumption: it rests on the finding that the income- and GDP-based elasticities are near-identical where both are observed (Italian NUTS3 and NUTS2), so that the Italian slope is informative about GDP even though the finest-level variable is income. Country-year intercepts then absorb the residual level difference between income and GDP. Kenya and Indonesia are excluded because their elasticities and level conversions are not comparable to this benchmark group.

We compare three level anchors. The first uses country-year intercepts estimated on the pooled local sample. The second imposes a single common intercept. The third keeps the pooled slopes but recovers the intercept from aggregate country-level totals, either by country-year or by country. Predicted local income or GDP density is then evaluated in log errors and by deciles of observed economic density.
\begin{table}[!htbp]
\centering
\small
\caption{Benchmark pooled finest-level calibration on the three homogeneous-elasticity countries: municipalities for Brazil and Italy, counties for the United States. Kenya is excluded because its level conversion and elasticity differ sharply from the benchmark economies; Indonesia is left for a robustness/developing-country check. The country-year-FE and common-intercept rows estimate the intercept on the pooled local sample; the aggregate-intercept rows keep the country-year-FE slopes but recover the intercept from aggregate country-level totals (by country-year or by country). Error statistics from predicted local income or GDP density.}
\label{tab:pooledFinestPrediction}
\resizebox{\textwidth}{!}{%
\input{tables/table_pooledFinestPrediction.tex}
}
\end{table}

The main lesson is that slopes transfer better than levels. With country-year intercepts, the pooled specification has mean errors essentially equal to zero by construction, and the nonlinear term modestly improves the error distribution (Table~\ref{tab:pooledFinestPrediction}). A balanced-sample robustness check in Section~\ref{app:pooledRobustness} of the Online Appendix gives essentially unchanged results, so the calibration is not driven by the larger number of Brazilian and Italian municipalities.

A common intercept is less damaging than expected in this benchmark group, but intercepts recovered only from aggregate country-level totals perform substantially worse. Thus $\beta$, $\gamma$, and $\ell_0$ are reasonably transferable within the benchmark economies, whereas an $\alpha$ calibrated at the country scale does not automatically anchor local density. Aggregate information can provide a pilot estimate, but quantitative fine-scale prediction still requires some local-level anchoring. The pooled calibration should therefore be used primarily to recover relative economic density across local units. When the objective is level prediction in a new country, the intercept must be re-anchored using aggregate national or regional accounts and, where possible, some local observations. Without this re-anchoring, the systematic level shifts documented in the validation exercises are likely to translate into biased local predictions.

\section{Concluding remarks \label{sec:concludingRemarks}}

This paper has studied nighttime lights as an application of a more general econometric problem: when can a high-resolution signal, observed on a fine spatial grid and aggregated to administrative units, be used to predict an economic variable that is not observed locally? The answer depends on two distinct margins. The first is reverse-regression attenuation: lights are generated by economic activity, but the empirical objective is often to predict economic activity from observed lights. The second is spatial aggregation: the relationship estimated at an administrative support need not coincide with the elementary relationship operating at finer locations.

The analytical framework clarifies how these two margins interact. The probability limit of the predictive elasticity differs from the elementary elasticity because of a noise component in the light signal and an aggregation component driven by unit size and within-unit dispersion. A central implication is that aggregation tends to contract elasticities toward unity. Unit elasticity is therefore not merely a convenient empirical regularity; it is the only benchmark under which the elasticity is invariant to aggregation. The Monte Carlo exercises confirm this decomposition. They also show that density specifications with area controls reduce the first-order unit-size channel, but do not remove all aggregation effects when the underlying relationship is nonlinear or when the light signal is noisy.

The empirical results are consistent with this conditional view of proxy validity. Using Black Marble VIIRS nighttime lights and official local GDP or income data, we find that Brazil, Italy, and the United States form a relatively stable benchmark group: at the finest available spatial supports, predictive elasticities are economically close to one and remain broadly stable across scales and over time once densities and area are used. Out-of-sample validation shows that lights add substantial predictive content beyond area and time effects alone, although prediction errors remain larger among low-luminosity units. In this group, pooled calibration is informative about relative local economic density, provided that levels are appropriately anchored.

The results for Indonesia and Kenya show the limits of mechanical transfer. Their predictive elasticities are substantially below one at fine scales, implying that differences in luminosity translate less than proportionally into differences in measured economic activity. Aggregation may partly conceal this pattern by moving estimated elasticities toward unity at coarser scales. Level conversion is even more fragile: for a given amount of light and area, lower-income, more agricultural, more informal, or weakly illuminated contexts may contain more economic activity than a benchmark-country calibration would imply.

The practical implication is that nighttime lights should be treated as a calibrated proxy rather than as a universal measure of local output. Applications should estimate the elasticity at the relevant spatial support, work with densities and unit area when possible, test whether the predictive elasticity is close enough to one for the intended use, and validate predictions out of sample. When the slope is transferable but levels are not, lights are better suited to recovering relative economic density than to producing unanchored level estimates. When the elasticity itself differs across contexts, country- or context-specific calibration is necessary.

Future work should extend this framework to longer panels and a wider set of low- and middle-income countries, where official local economic data are scarce and the value of a reliable proxy is highest. The same logic can also be applied to other high-resolution spatial signals used in economics, such as remotely sensed land use, building footprints, mobile-phone activity, or transaction-based indicators. In all these cases, the key question is not only whether the signal correlates with the target variable, but whether its predictive relationship survives aggregation, noise, and transfer across contexts.

\bibliographystyle{chicago}
\bibliography{biblio}

\clearpage
\appendix

\textbf{\LARGE Appendix}

\section*{Proof of the Aggregation and Attenuation Decomposition}
\label{app:aggregationBias}

This appendix proves Theorem~\ref{teo:aggregation}. The proof follows the decomposition used in Section~\ref{subsec:mcAnalyticalBenchmark}: first isolate the component generated by nonlinear aggregation of elementary locations, then isolate the reverse-regression attenuation component generated by output-to-light shocks, and finally show that the cross terms are higher order under the maintained independence assumption.

\begin{proof}[Proof of Theorem~\ref{teo:aggregation}]
Use the residual split in~\eqref{eq:mcResidualSplit}, with components defined in~\eqref{eq:mcResidualComponents}.

Along the local path $s_{ip}(\chi)=\abs{m_p}^{-1}(1+\chi d_{ip})$, with $\sum_{i\in m_p}d_{ip}=0$, define
\begin{equation*}
g_p(\chi)\equiv\sum_{i\in m_p}s_{ip}(\chi)^\mu = \abs{m_p}^{-\mu}\sum_{i\in m_p}(1+\chi d_{ip})^\mu.
\end{equation*}
Applying Taylor's theorem to $x\mapsto(1+x)^\mu$ around $x=0$ gives
\begin{equation}
\label{eq:proofgP}
g_p(\chi)=\abs{m_p}^{-\mu}\sum_{i\in m_p}\left[1+\mu\chi d_{ip}+\frac{\mu(\mu-1)}{2}\chi^2 d_{ip}^2+O(\chi^3)\right].
\end{equation}
Equivalently,
\begin{equation*}
\sum_{i\in m_p}s_{ip}(\chi)^\mu = \abs{m_p}^{1-\mu}\left[1+\mu\chi\frac{1}{\abs{m_p}}\sum_{i\in m_p}d_{ip}+\frac{\mu(\mu-1)}{2}\chi^2\frac{1}{\abs{m_p}}\sum_{i\in m_p}d_{ip}^2+O(\chi^3)\right].
\end{equation*}
The first-order term is zero because $\sum_{i\in m_p}d_{ip}=0$. The $O(\chi^3)$ remainder is uniform along the cross-section sequence because $\sup_P\max_{1\le p\le P}\max_{i\in m_p}\abs{d_{ip}}<\infty$: for $\chi$ small enough, all arguments $\chi d_{ip}$ remain in a common neighbourhood of zero, where the third derivative of $x\mapsto(1+x)^\mu$ is bounded.

Since $ntl_i=s_{ip}(\chi)NTL_p$ and $\overline{ntl}_p=NTL_p/\abs{m_p}$,
\begin{equation*}
\mathrm{CV}_p^2(\chi) = \frac{1}{\abs{m_p}}\sum_{i\in m_p}\left(\frac{ntl_i-\overline{ntl}_p}{\overline{ntl}_p}\right)^2 = \chi^2\frac{1}{\abs{m_p}}\sum_{i\in m_p}d_{ip}^2.
\end{equation*}
Therefore
\begin{equation*}
\Lambda_p^{\mathrm{agg}} = (1-\mu)\log\abs{m_p}+\frac{\mu(\mu-1)}{2}\mathrm{CV}_p^2(\chi)+O(\chi^3).
\end{equation*}
Projecting this term on aggregate lights yields
\begin{equation*}
\frac{\mathrm{Cov}(\Lambda_p^{\mathrm{agg}},\log NTL_p)}{\mathrm{Var}(\log NTL_p)} = (1-\mu)\left[\delta_P-\frac{\mu}{2}\rho_P\right]+O(\chi^3),
\end{equation*}
where $\delta_P$ and $\rho_P$ are defined in~\eqref{eq:mcDeltaRhoDef}.
It remains to account for the shock component. Recall from Section~\ref{subsec:mcAnalyticalBenchmark} the noise-free aggregate light index $A_p$, the noise-free within-unit shares $s_{ip}^{(0)}$, and the aggregate shock $\tilde\eps_p=\sum_{i\in m_p}s_{ip}^{(0)}\eps_i$. To derive the first-order expansion $\log NTL_p=A_p-\tilde\eps_p/\mu+O_p(\sigma_\eps^2)$, write aggregate lights as a function of the shock vector:
\begin{equation*}
H_p(\eps)\equiv\log NTL_p(\eps)=\log\!\left(\sum_{i\in m_p}\exp\!\left\{\frac{\nu_i-\phi-\eps_i}{\mu}\right\}\right).
\end{equation*}
Then $H_p(0)=A_p$. Moreover, for each $j\in m_p$,
\begin{equation*}
\left.\frac{\partial H_p(\eps)}{\partial\eps_j}\right|_{\eps=0} = -\frac{1}{\mu}\frac{\exp\!\left\{\frac{\nu_j-\phi}{\mu}\right\}}{\sum_{i\in m_p}\exp\!\left\{\frac{\nu_i-\phi}{\mu}\right\}} = -\frac{1}{\mu}s_{jp}^{(0)}.
\end{equation*}
Therefore the first-order multivariate Taylor expansion around $\eps=0$ gives
\begin{equation*}
\log NTL_p = H_p(0)+\sum_{j\in m_p}\left.\frac{\partial H_p(\eps)}{\partial\eps_j}\right|_{\eps=0}\eps_j+O_p(\sigma_\eps^2),
\end{equation*}
and hence
\begin{equation*}
\log NTL_p = A_p-\frac{1}{\mu}\sum_{j\in m_p}s_{jp}^{(0)}\eps_j+O_p(\sigma_\eps^2)=A_p-\frac{1}{\mu}\tilde\eps_p+O_p(\sigma_\eps^2).
\end{equation*}

To expand the noise component, first expand the weights. Along the local path, by the same calculation as in Eq.~\eqref{eq:proofgP},
\begin{equation*}
s_{ip}(\chi)^\mu = \abs{m_p}^{-\mu}\left[1+\mu\chi d_{ip}+O(\chi^2)\right].
\end{equation*}
Moreover,
\begin{equation*}
\sum_{j\in m_p}s_{jp}(\chi)^\mu = \abs{m_p}^{1-\mu}\left[1+\mu\chi\frac{1}{\abs{m_p}}\sum_{j\in m_p}d_{jp}+O(\chi^2)\right].
\end{equation*}
Since $\sum_{j\in m_p}d_{jp}=0$, this becomes
\begin{equation*}
\sum_{j\in m_p}s_{jp}(\chi)^\mu = \abs{m_p}^{1-\mu}\left[1+O(\chi^2)\right].
\end{equation*}
Therefore
\begin{equation*}
w_{ip}=\frac{s_{ip}(\chi)^\mu}{\sum_{j\in m_p}s_{jp}(\chi)^\mu}=\abs{m_p}^{-1}\left[1+\mu\chi d_{ip}+O(\chi^2)\right]=\abs{m_p}^{-1}+O(\chi),
\end{equation*}
with remainders uniform along the cross-section sequence by the boundedness assumption on $d_{ip}$.\footnote{We used the trivial identity $1/(1+O(\chi^2)) = 1+O(\chi^2)$.}
Using $e^{\eps_i}=1+\eps_i+O(\eps_i^2)$ and $\sum_{i\in m_p}w_{ip}=1$,
\begin{equation*}
\sum_{i\in m_p}w_{ip}e^{\eps_i}=1+\sum_{i\in m_p}w_{ip}\eps_i+O_p(\sigma_\eps^2).
\end{equation*}
Using $\bar\eps_{p,w}\equiv\sum_{i\in m_p}w_{ip}\eps_i$, and since $\bar\eps_{p,w}=O_p(\sigma_\eps)$, applying $\log(1+x)=x+O(x^2)$ to $x=\bar\eps_{p,w}+O_p(\sigma_\eps^2)$ gives
\begin{equation*}
\Lambda_p^{\mathrm{noise}}=\log\!\left(1+\bar\eps_{p,w}+O_p(\sigma_\eps^2)\right)=\bar\eps_{p,w}+O_p(\sigma_\eps^2).
\end{equation*}
Thus the noise component is driven, to first order, by the weighted aggregate shock $\bar\eps_{p,w}$. The maintained independence between activity and output-to-light shocks implies that $A_p$ is orthogonal to this first-order shock term. Moreover,
\begin{equation*}
\mathrm{Var}(\log NTL_p)=\frac{1}{\mu^2}\left[\mathrm{Var}(\mu A_p)+\mathrm{Var}(\tilde\eps_p)\right]+O_{p}(\sigma_\eps^2).
\end{equation*}
Therefore the shock contribution to the projection is
\begin{equation*}
-\frac{1}{\mu}\frac{\mathrm{Cov}(\Lambda_p^{\mathrm{noise}},\tilde\eps_p)}{\mathrm{Var}(\log NTL_p)} = -\mu\kappa_P+O_p(\sigma_\eps^2),
\end{equation*}
where $\kappa_P$ is defined in~\eqref{eq:mcReliability}.

It remains to justify the two cross terms. Write the leading aggregation term as
\begin{equation*}
B_p\equiv (1-\mu)\log\abs{m_p}+\frac{\mu(\mu-1)}{2}\mathrm{CV}_p^2(\chi),
\end{equation*}
so that $\Lambda_p^{\mathrm{agg}}=B_p+O(\chi^3)$. The objects $B_p$ and $s_{ip}^{(0)}$ are functions of the elementary log outputs and of the grouping. Hence, using the independence of the elementary shocks,
\begin{equation*}
\begin{split}
\mathrm{Cov}(B_p,\tilde\eps_p)&=\EE{(B_p-\EE{B_p})\sum_{i\in m_p}s_{ip}^{(0)}\eps_i} \\
&=\sum_{i\in m_p}\EE{(B_p-\EE{B_p})s_{ip}^{(0)}}\EE{\eps_i}=0.
\end{split}
\end{equation*}
Therefore
\begin{equation*}
\mathrm{Cov}(\Lambda_p^{\mathrm{agg}},\tilde\eps_p)=O(\chi^3\sigma_\eps).
\end{equation*}
For the other cross term, use $\Lambda_p^{\mathrm{noise}}=\bar\eps_{p,w}+O_p(\sigma_\eps^2)$, with $\bar\eps_{p,w}\equiv\sum_{i\in m_p}w_{ip}\eps_i$. Along the local share path, $A_p$ and $w_{ip}$ are functions of the elementary log outputs and of the grouping. Hence, using the independence of the elementary shocks,
\begin{equation*}
\begin{split}
\mathrm{Cov}(\bar\eps_{p,w},A_p)&=\EE{(A_p-\EE{A_p})\sum_{i\in m_p}w_{ip}\eps_i} \\
&=\sum_{i\in m_p}\EE{(A_p-\EE{A_p})w_{ip}}\EE{\eps_i}=0.
\end{split}
\end{equation*}
Thus
\begin{equation*}
\mathrm{Cov}(\Lambda_p^{\mathrm{noise}},A_p)=O_p(\sigma_\eps^2).
\end{equation*}
The product term behind the interaction remainder can be seen directly from the two local expansions
\begin{equation*}
w_{ip}=\abs{m_p}^{-1}+\chi a_{ip}+O(\chi^2),\qquad e^{\eps_i}=1+\eps_i+O_p(\sigma_\eps^2),
\end{equation*}
where $a_{ip}$ is uniformly bounded along the cross-section sequence. Multiplying the two expansions inside the noise component gives
\begin{equation*}
\sum_{i\in m_p}w_{ip}e^{\eps_i}=1+\sum_{i\in m_p}\abs{m_p}^{-1}\eps_i+\chi\sum_{i\in m_p}a_{ip}\eps_i+O_p(\sigma_\eps^2)+O_p(\chi^2\sigma_\eps).
\end{equation*}
The third term is the interaction between the local share perturbation and the shock expansion; since $a_{ip}=O(1)$ and $\eps_i=O_p(\sigma_\eps)$, it is $O_p(\chi\sigma_\eps)$. This is the source of the $O_p(\sigma_\eps\chi)$ term included in $R(\chi,\sigma_\eps)$.

Combining the aggregation and shock projections and taking limits of the finite-$P$ projection coefficients, $\kappa_P\to\kappa_\infty$, $\delta_P\to\delta_\infty$, and $\rho_P\to\rho_\infty$, gives the local expansion of the population slope,
\begin{equation*}
\plim_{P\to\infty}\hat\beta_P-\mu = -\mu\kappa_\infty+(1-\mu)\left[\delta_\infty-\frac{\mu}{2}\rho_\infty\right]+R(\chi,\sigma_\eps),
\end{equation*}
where
\begin{equation*}
R(\chi,\sigma_\eps)=O(\chi^3)+O(\sigma_\eps^2)+O(\sigma_\eps\chi).
\end{equation*}
This is~\eqref{eq:mcBiasApprox}.
\end{proof}

\clearpage

\renewcommand{\thesection}{OA-\Alph{section}}
\setcounter{section}{0}
\begin{center}
{\LARGE\bfseries Online Appendix}
\end{center}

This Online Appendix contains Monte Carlo details, robustness exercises, and supplementary descriptive material supporting the results in the paper. The Monte Carlo tables are placed first because they provide the methodological backbone of the analysis; the subsequent sections report additional empirical diagnostics and complete robustness tables.

\section{Monte Carlo Implementation and Tables}
\label{sec:tabelleMonteCarlo}

This appendix reports the numerical details behind the Monte Carlo figures of Section~\ref{sec:MonteCarloExperiments}. The tables are not additional evidence separate from the figures; rather, they document the estimates, intervals, and approximation errors underlying the graphical summary in the main text.

Table~\ref{tab:mcBetaAlphaTotalEstimates} reports the full estimates of the totals specification behind Figure~\ref{fig:mcBetaAlphaTotal}; Table~\ref{tab:mcAreaBias} summarises how the area control changes the average bias in $\hat\beta$; Table~\ref{tab:mcBetaGammaEstimates} reports the full Monte Carlo estimates of $\hat\alpha$, $\hat\beta$, and $\hat\gamma$ under the area-control specifications behind Figure~\ref{fig:mcBetaGammaSixPanel}; Table~\ref{tab:mcNoise} isolates reverse-regression attenuation as idiosyncratic output-to-light noise varies; and Table~\ref{tab:mcBiasApproxVerification} reports the full numerical verification of the bias approximation of Theorem~\ref{teo:aggregation} behind Figure~\ref{fig:mcBiasApproxVerification}.

\input{tables/table_montecarlo_alpha_beta_total.tex}

\input{tables/table_montecarlo_area_bias.tex}

\input{tables/table_montecarlo_alpha_beta_gamma.tex}

\input{tables/table_montecarlo_noise.tex}

\input{tables/table_montecarlo_biasApprox.tex}

\clearpage

\section{Nonlinear Model Details}
\label{app:kappaEstimates}

Table~\ref{tab:kappaEstimates} reports the NLM parameters used to construct Figure~\ref{fig:elasticityByNTL}. Time-varying intercepts are included in estimation but omitted from the table, so the comparison focuses on $\hat{\beta}$, $\hat{\ell}_0$, and $\hat{\gamma}$ across countries and spatial scales.

\input{tables/table_kappa.tex}

\begin{figure}[!hbtp]
\centering
\includegraphics[width=\textwidth]{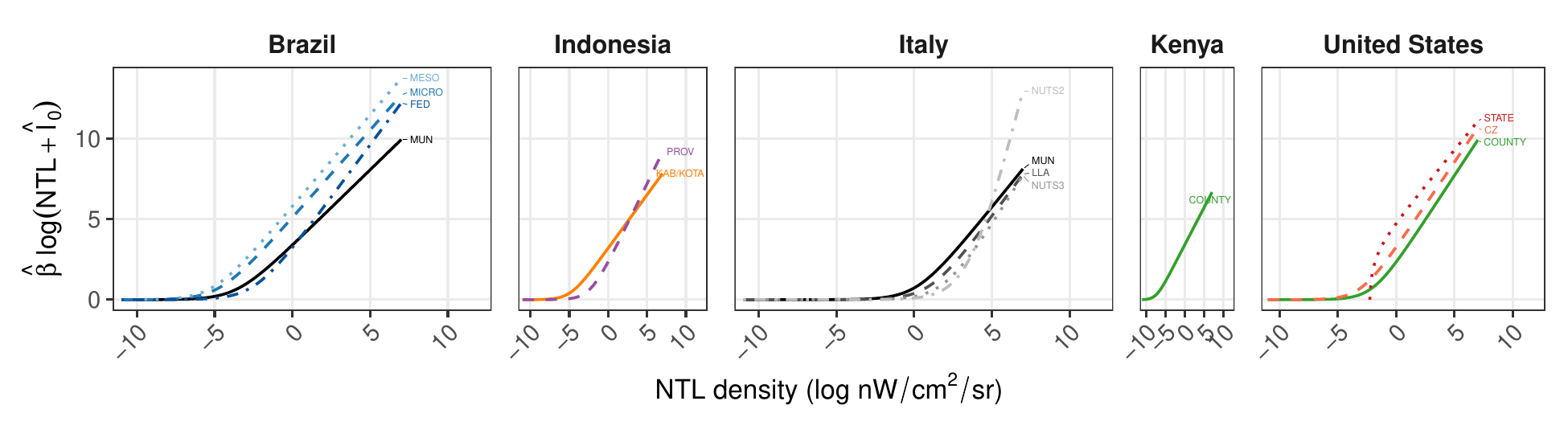}
\caption{Estimated nonlinear component of Eq.~\eqref{eq:lightOnIncomeQuadratic} by country and administrative level. Curves report $\hat{\beta}\log(NTL+\hat{\ell}_0)$ over the observed NTL-density range and are centred at the left edge of each panel for comparability. Administrative levels are labelled directly on each curve (acronyms), and panel width is proportional to the number of administrative levels of each country.}
\label{fig:EstimatedQuadraticEffect_countries}
\end{figure}

\clearpage

\section{Additional Validation Results}
\label{app:additionalValidation}

This appendix reports the cross-scale and temporal validation exercises summarised in Section~\ref{sec:validation}. The purpose is to separate two objects that are deliberately kept distinct in the paper: the transferability of slopes across spatial supports and years, and the transferability of levels without re-anchoring the intercept.

\begin{table}[!htbp]
\centering
\small
\caption{Cross-scale (downscaling) out-of-sample performance. Coarse-level coefficients predict finest-level density for a \textit{large} or \textit{small} scale jump. The demeaned $R^2_{\text{oos}}$ removes the systematic level shift between scales; level bias is the mean log forecast error.}
\label{tab:cvCrossScale}
\resizebox{\textwidth}{!}{%
\input{tables/table_cvCrossScale.tex}
}
\end{table}

\begin{table}[!htbp]
\centering
\small
\caption{Temporal (leave-years-out) out-of-sample performance at the finest level, predicting each held-out year from the average estimated intercept.}
\label{tab:cvTemporal}
\resizebox{\textwidth}{!}{%
\input{tables/table_cvTemporal.tex}
}
\end{table}

\clearpage

\section{Pooled Calibration Robustness}
\label{app:pooledRobustness}

This appendix reports robustness checks for the pooled finest-level calibration in Section~\ref{sec:pooledCalibration}. The checks focus on whether the pooled results are driven by the unbalanced number of local units across Brazil, Italy, and the United States, and on how aggregate-country intercept anchoring changes level predictions.

\begin{table}[!htbp]
\centering
\small
\caption{Balanced-sample robustness for the benchmark pooled finest-level calibration. For each country-year, the sample draws the same number of local units for Brazil, Italy, and the United States, equal to the minimum available count across the three countries (3,014 units, corresponding to the U.S. county sample). The random draw uses a fixed seed. Error statistics are in log points and are computed from predicted local income or GDP density.}
\label{tab:pooledFinestPredictionBalanced}
\resizebox{\textwidth}{!}{%
\input{tables/table_pooledFinestPredictionBalanced.tex}
}
\end{table}

\begin{table}[!htbp]
\centering
\small
\caption{Country-specific intercepts used in the aggregate-country calibration of Table~\ref{tab:pooledFinestPrediction}. The values are obtained from country-level totals using the pooled slopes estimated in the country-year FE specification. This table is included only as a diagnostic check.}
\label{tab:pooledFinestAggregateAlpha}
\input{tables/table_pooledFinestAggregateAlpha.tex}
\end{table}

\begin{figure}[!htbp]
\centering
\includegraphics[width=\textwidth]{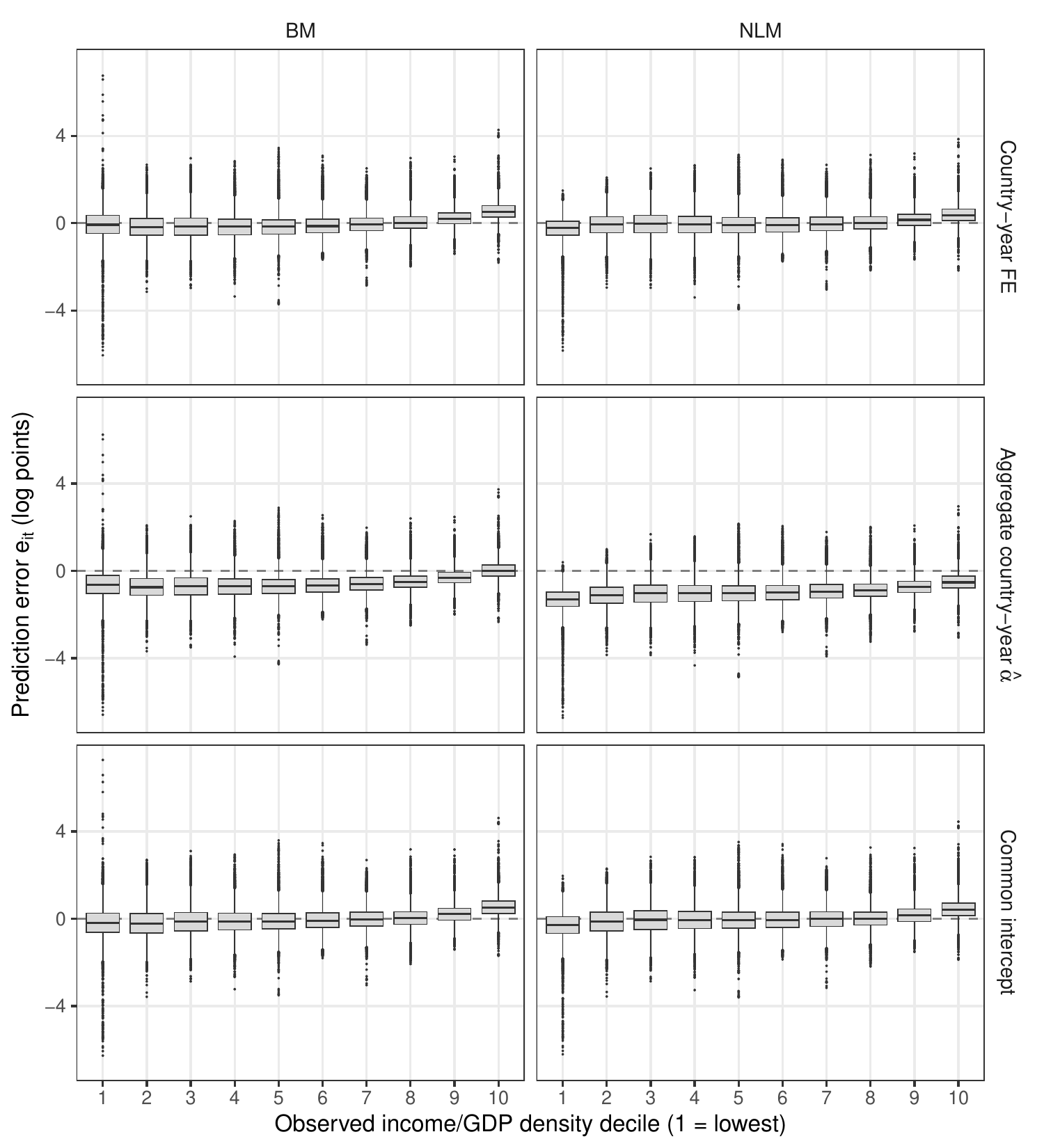}
\caption{Distribution of prediction errors from the benchmark pooled finest-level calibration, by model, intercept calibration, and decile of observed local income or GDP density (decile 1 = lowest). The exercise pools Brazil MUN, Italy MUN, and USA COUNTY observations.}
\label{fig:pooledFinestPredictionErrorByIncome}
\end{figure}

\clearpage

\section{Indonesia Aggregation Check}
\label{app:indonesiaAggregation}

Indonesia provides an empirical check of the aggregation mechanism because GDP and NTL are observed at two administrative scales and the kabupaten/kota units can be further aggregated into coarser random partitions. Starting from the kabupaten/kota sample, we repeatedly form partitions of increasing size and re-estimate BM on each partition. Figure~\ref{fig:betaByScaleIDN} reports the resulting elasticity estimates together with the administrative anchors.

The pattern is consistent with the analytical benchmark in the Appendix of the paper and with the Monte Carlo results in Section~\ref{sec:montecarlo}. The estimated elasticity is below one at the kabupaten/kota level, rises when observations are aggregated to provinces, and drifts toward unity under progressively coarser random partitions. Thus, in a setting where the finest-scale elasticity is below one, aggregation can make the relationship appear closer to the unit-elastic benchmark than it is at the local scale.

\begin{figure}[!htbp]
\centering
\includegraphics[width=0.7\linewidth]{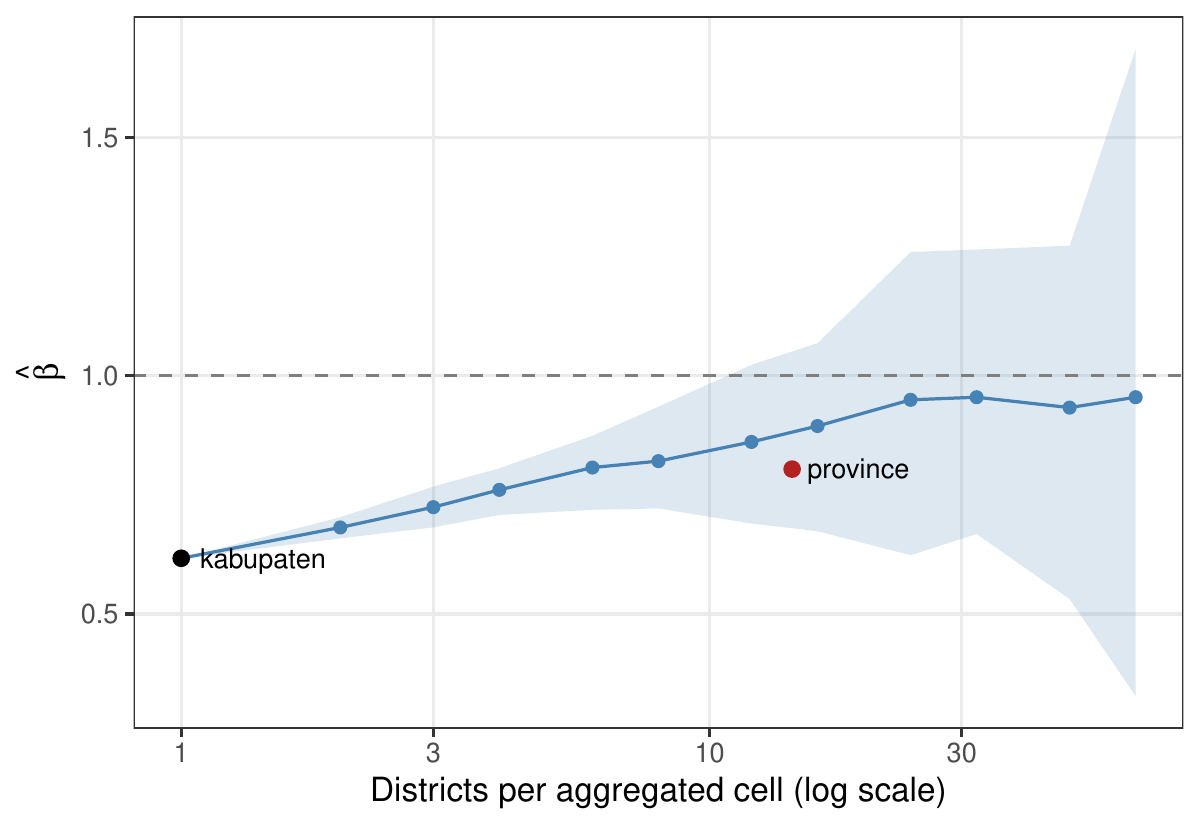}
\caption{Indonesia: estimated NTL elasticity $\hat\beta$ as the 487 kabupaten/kota are aggregated into units of increasing size (random partitions, 95\% band), with the administrative anchors (kabupaten, province). The elasticity drifts from the low finest-scale estimate toward unity, as predicted by Theorem~\ref{teo:aggregation} for $\mu<1$.}
\label{fig:betaByScaleIDN}
\end{figure}

\clearpage

\section{An Aggregation-Corrected Elementary Elasticity}
\label{app:muCorrected}

This appendix develops the supporting evidence referenced in Section~\ref{sec:betaElasticity}. The third implication of Theorem~\ref{teo:aggregation} discussed in Section~\ref{subsec:mcAnalyticalBenchmark} can be taken to the data: assuming negligible measurement error ($\kappa_\infty=0$), Eq.~\eqref{eq:mcBiasApprox} can be inverted to recover the elementary (cell-level) elasticity $\mu$ from the estimated slope and the empirical projection coefficients of Table~\ref{tab:deltaRhoEmpirics}. Since the theorem applies to the regression of log totals on log totals, we first estimate the naive slope $\hat\beta^{\mathrm{biased}}$ at each country and administrative level, and then solve $\hat\beta^{\mathrm{biased}} = \mu + (1-\mu)\left[\delta_P - \tfrac{\mu}{2}\rho_P\right]$ for $\mu$. Table~\ref{tab:muCorrected} reports the resulting $\hat\mu$. For comparison, the first column reports the Baseline-Model slope $\hat\beta^{\mathrm{BM}}$ on the same samples: the BM corrects for aggregation \emph{by design}, through densities and the area control, so the gap between $\hat\beta^{\mathrm{BM}}$ and $\hat\beta^{\mathrm{biased}}$ visualises the correction that the inversion must deliver \emph{ex post}, and the two routes can be checked against each other.

Three lessons emerge. First, at the finest scales the corrected elasticities confirm and sharpen the cross-country pattern of Figure~\ref{fig:estimatedElasticity_countries}: $\hat\mu$ is $1.00$ for Italian municipalities, $0.91$ for U.S. counties, $0.76$ for Brazilian municipalities, and $0.70$ for Indonesian kabupaten/kota. The ordering matches income levels and luminosity, and the ex post corrections sit close to the by-design ones: $\hat\mu$ is within a few hundredths of $\hat\beta^{\mathrm{BM}}$ in Italy, the United States, and Indonesia, while for Brazilian municipalities it lies somewhat below ($0.76$ against $0.87$), consistent with $\hat\mu$ being a cell-level object: the cell-to-municipality aggregation step, which the municipal BM cannot remove, is already at work there. For Indonesia the full $(\delta,\rho)$ correction is essential: ignoring the within-unit dispersion term $\rho_P$ would deliver a meaningless negative value, while the full inversion yields an elementary elasticity consistent with the below-unity estimates of Section~\ref{sec:betaElasticity}. Second, the correction degrades exactly where Theorem~\ref{teo:aggregation} predicts: at state and regional scales the quadratic has no real root or diverges, mirroring the large-remainder region of the Monte Carlo verification (Figure~\ref{fig:mcBiasApproxVerification}). The inversion is therefore a fine-scale tool, not a substitute for estimation at the appropriate resolution. Third, Kenya is informative through failure: the naive county slope ($0.52$) lies below $\delta_P=0.77$, a configuration that no non-negative $\mu$ can generate under $\kappa_\infty=0$. The maintained assumption of negligible measurement error is untenable there, consistent with lights missing a substantial share of agricultural and informal activity in low-luminosity settings.

\input{tables/table_muCorrected.tex}

\clearpage

\section{Inference Robustness}
\label{app:completeResultLinear}

This appendix documents the robustness of the inference for \emph{all} model parameters, ensuring that our conclusions are not artefacts of specific assumptions about the residual covariance structure. Throughout we report $95\%$ confidence intervals rather than standard errors, so that the precision of each estimate can be read directly.

Table~\ref{tab:bm_full_robust} gives the complete set of Baseline Model parameters --- the time-varying intercepts $\alpha_t$, the year effects, the elasticity $\beta$, and the area coefficient $\gamma$ --- and Table~\ref{tab:nlm_robust} the Nonlinear Model parameters, including the luminosity shift $\ell_0$; a consistent battery of inferential methods is applied to each specification.

The first inferential baseline constructs confidence intervals by clustering standard errors at the individual spatial unit level. This allows the residuals of a given municipality, county, or province to be arbitrarily serially correlated over time, which represents the standard approach in short administrative panels and corresponds to the inference displayed in Figures~\ref{fig:estimatedElasticity_countries} and~\ref{fig:estimatedGamma_countries}.

The second baseline implements a more conservative clustering strategy at a higher spatial level --- Brazilian states for municipalities, Italian NUTS3 provinces for municipalities, U.S. states for counties, and Indonesian provinces for kabupaten/kota. This multi-level clustering framework relaxes the assumption of spatial independence across neighbouring borders by allowing for unrestricted cross-sectional correlation among all nested sub-units. A non-parametric Conley spatial-HAC estimator offers a natural complement, at the cost of choosing coordinate centroids and a distance bandwidth. Our higher-level clustering is a convenient alternative that sidesteps these tuning choices and absorbs localised spatial spillovers and policy shocks operating along regional administrative boundaries; because it allows unrestricted correlation within large administrative blocks, it is if anything more conservative than a Conley estimator with a moderate bandwidth. We view the two as complementary rather than as substitutes. For large panels, the cluster-robust variance matrix uses the standard CR1S sandwich estimator, while for smaller panels the CR2 small-sample correction is applied where computationally feasible.

As an additional diagnostic check --- particularly reliable when the number of aggregate spatial clusters is relatively small --- we compute wild cluster bootstrap confidence intervals (using Rademacher weights, a percentile-$t$ approach, and $1{,}999$ replications) applied directly at the higher geographic cluster tier. The resulting bootstrap intervals remain remarkably close to the analytical cluster-robust counterparts throughout the analysis --- for instance, yielding $[0.79,0.96]$ against $[0.79,0.95]$ at the Brazilian municipal scale, and $[0.89,1.02]$ against $[0.89,1.02]$ for U.S. counties. This tight alignment confirms that our inferential conclusions are robust to the choice of variance estimator and asymptotic approximations.

Finally, Table~\ref{tab:beta_tost} operationalises the discussion on the NTL elasticity via a formal \emph{equivalence testing} framework. In large administrative panels containing thousands of observations, conventional hypothesis testing can be misleading: because of high statistical power, even economically trivial departures from unity will reject the point null hypothesis $H_0: \beta=1$. To shift the focus from mere statistical power to genuine economic significance, we implement the Two One-Sided Tests (TOST) procedure. This framework tests the null hypothesis of non-equivalence against the alternative that $\beta$ falls entirely within a pre-specified stable range $(1-\delta, 1+\delta)$. We evaluate this at two conservative tolerance thresholds, $\delta=10\%$ and $\delta=20\%$.

The TOST results confirm that equivalence to unit elasticity holds at the $5\%$ significance level (where the $90\%$ confidence interval is entirely contained within the tolerance bounds) under both baseline and higher-geographic clustering for the benchmark group. Specifically, Italian municipalities and U.S. counties exhibit formal equivalence within the strict $\pm10\%$ interval, while all three benchmark economies --- including Brazilian municipalities --- satisfy equivalence within the $\pm20\%$ range. Conversely, Kenya and Indonesia systematically reject the hypothesis of equivalence at either tolerance threshold, showing that near-unit elasticity is a distinct empirical feature of the benchmark contexts rather than a mechanical byproduct of the spatial data.

\input{tables/table_bm_full_robust.tex}

\input{tables/table_nlm_robust.tex}

\input{tables/table_beta_tost.tex}

\clearpage

\section{Selected Literature and Raw NTL Maps}
\label{app:literatureTable}

This final appendix collects descriptive material that supports the motivation and data discussion in Sections~\ref{sec:introduction} and~\ref{sec:dataset}. Table~\ref{tab:nightlights_gdp} summarises the heterogeneity in economic target, satellite source, NTL transformation, and estimation strategy across selected studies. Figure~\ref{app:rawMaps} reports the raw 2019 Black Marble NTL maps used as descriptive background in Section~\ref{sec:dataset}.

\begin{table}[!htbp]
\centering
\tiny
\caption{Selected studies using NTL as a proxy for economic activity and related socio-economic outcomes.}
\label{tab:nightlights_gdp}
\begin{tabular}{lllll}
	\hline
	\hline
	\textbf{Study} & \textbf{Economic measure} & \textbf{Source of } & \textbf{NTL} & \textbf{Estimation}  \\
	 & & \textbf{satellite data} & & \textbf{details} \\
	\midrule
	Sutton \& Costanza (2002) & GDP per $km^2$ & DMSP-OLS & Light density & U.S. states, 1995 \\
	Chen \& Nordhaus (2011) & GDP density/growth & DMSP-OLS & Light density & Grid cells, 1992--2008 \\
	Henderson et al. (2012) & Total GDP & DMSP-OLS & Area-weighted light & Countries, panel FE \\
	Galimberti (2020) & GDP growth & DMSP-OLS & Light growth & Countries, panel/AR models \\
	Hu \& Yao (2022) & GDP growth & DMSP-OLS & Sum of light & Countries, panel FE \\
	Martinez (2022) & Total GDP & DMSP-OLS/VIIRS & Light density & Countries, panel FE \\
	Gibson (2021) & Total GDP & DMSP-OLS/VIIRS & Sum of light & European NUTS2 \\
	Bluhm et al. (2022) & GDP density & DMSP-OLS & Light density & Subnational, multi-scale \\
	Gibson et al. (2024) & Local GDP growth & DMSP-OLS/VIIRS & Corrected light growth & Subnational, quantile regression \\
	Rossi-Hansberg \& Zhang (2025) & GDP and growth & VIIRS Black Marble & Light density/growth & Global gridded GDP \\
	Jean et al. (2016) & Poverty/wealth proxy & DMSP-OLS + daytime & Mean light & CNN, African countries \\
	Mirza et al. (2021) & Income inequality & DMSP-OLS/VIIRS & Light-intensity Gini & Subnational panel \\
	Abbes et al. (2023) & Wealth proxy & VIIRS + daytime & Mean light & Deep learning, Africa \\
	Lehnert et al. (2023) & Total GDP & Landsat + NTL & Sum of light & Machine learning, Germany \\
	Giannini (2025) & Personal income & Black Marble & Sum of light & Recurrent neural network, Italy \\
	\hline
	\hline
\end{tabular}
\end{table}

\begin{figure}[!hbtp]
\centering
\caption{Raw NTL in 2019 at approximately 500-meter resolution for Brazil, Italy, the United States, Kenya, and Indonesia.}
\label{app:rawMaps}
\subcaptionbox{Brazil}{\includegraphics[height=3.2cm]{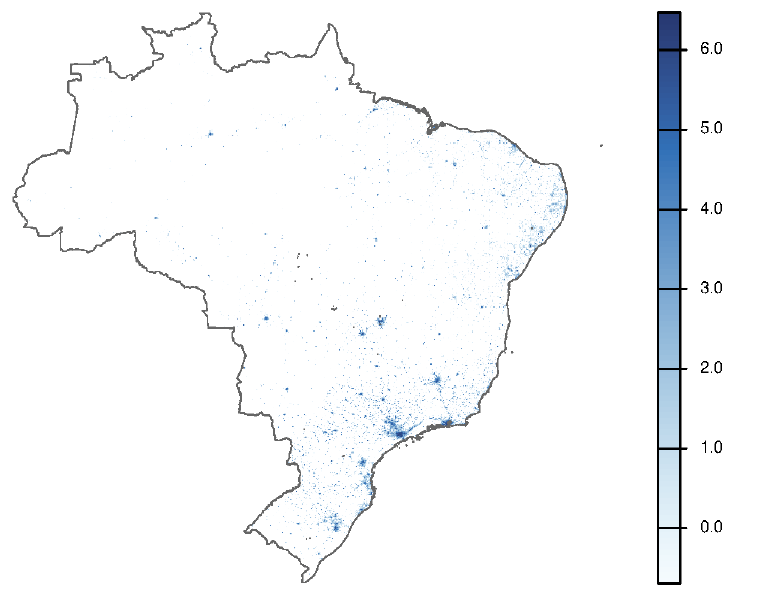}}\hfill
\subcaptionbox{Italy}{\includegraphics[height=3.2cm]{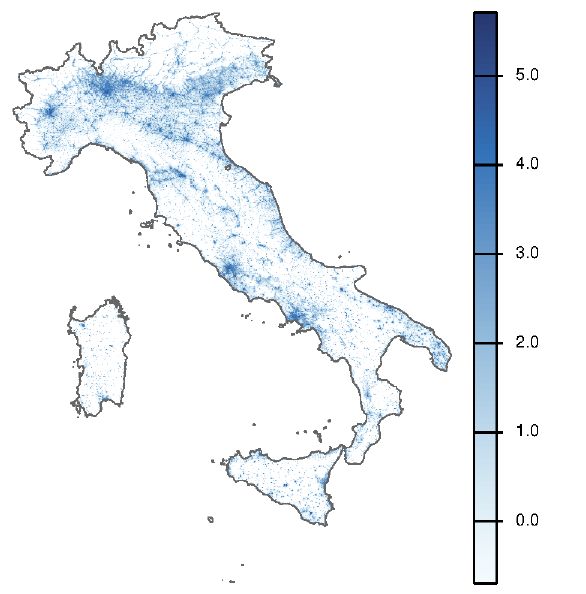}}\hfill
\subcaptionbox{Kenya}{\includegraphics[height=3.2cm]{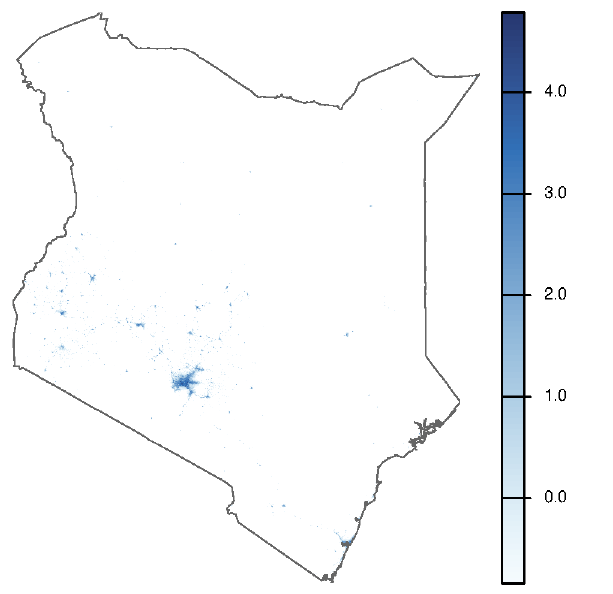}}\\
\subcaptionbox{USA}{\includegraphics[height=3.2cm]{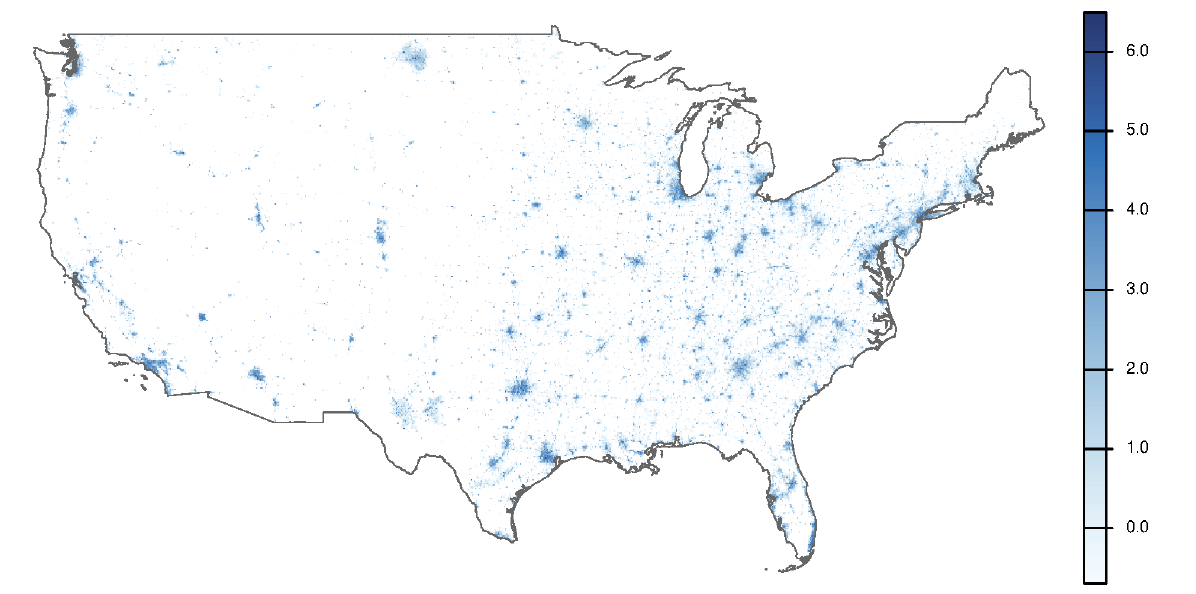}}\hfill
\subcaptionbox{Indonesia}{\includegraphics[height=3.2cm]{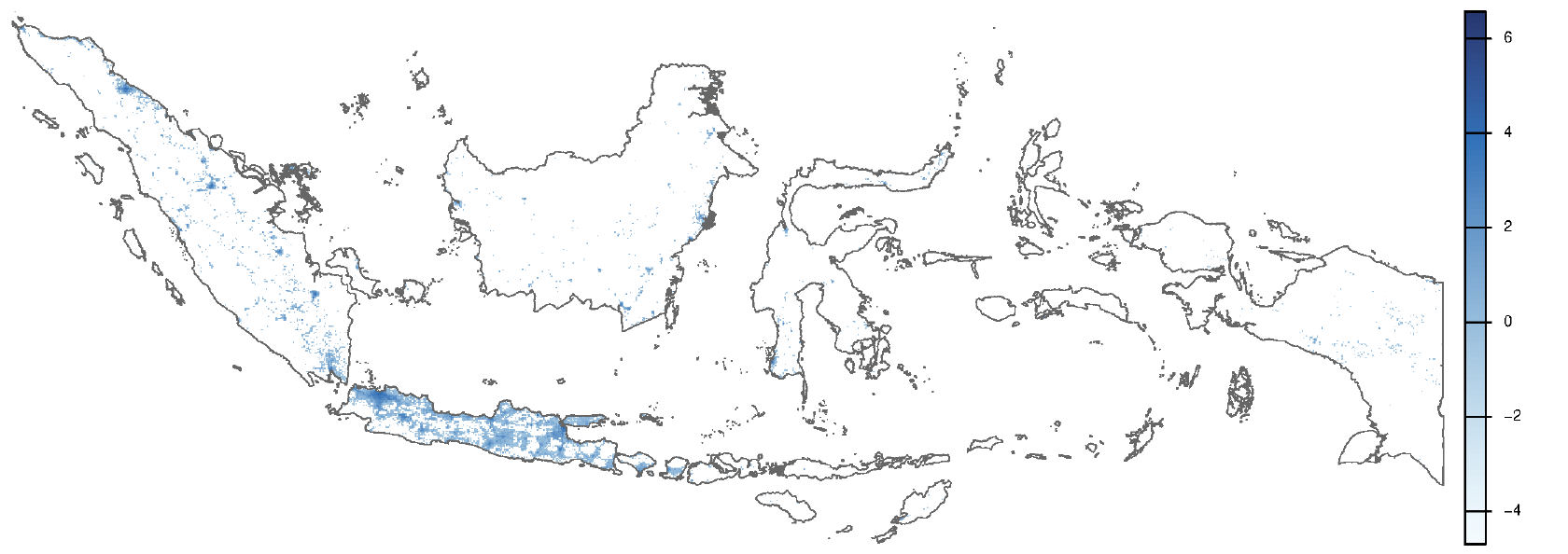}}
\begin{flushleft}
	\textit{\small Source: NASA Black Marble NTL.}
\end{flushleft}
\end{figure}

\clearpage

\section{Additional Monte Carlo Figures}
\label{app:mcFigures}
This appendix collects the Monte Carlo figures referenced in Section~\ref{sec:MonteCarloExperiments} that complement the summary evidence retained in the main text: the totals-regression estimates of the elasticity and scale parameter (Figure~\ref{fig:mcBetaAlphaTotal}) and the reverse-regression attenuation experiments for the naive and preferred specifications (Figures~\ref{fig:mcNoiseNaive} and~\ref{fig:mcNoise}).

\begin{figure}[!htbp]
\centering
\includegraphics[width=\textwidth]{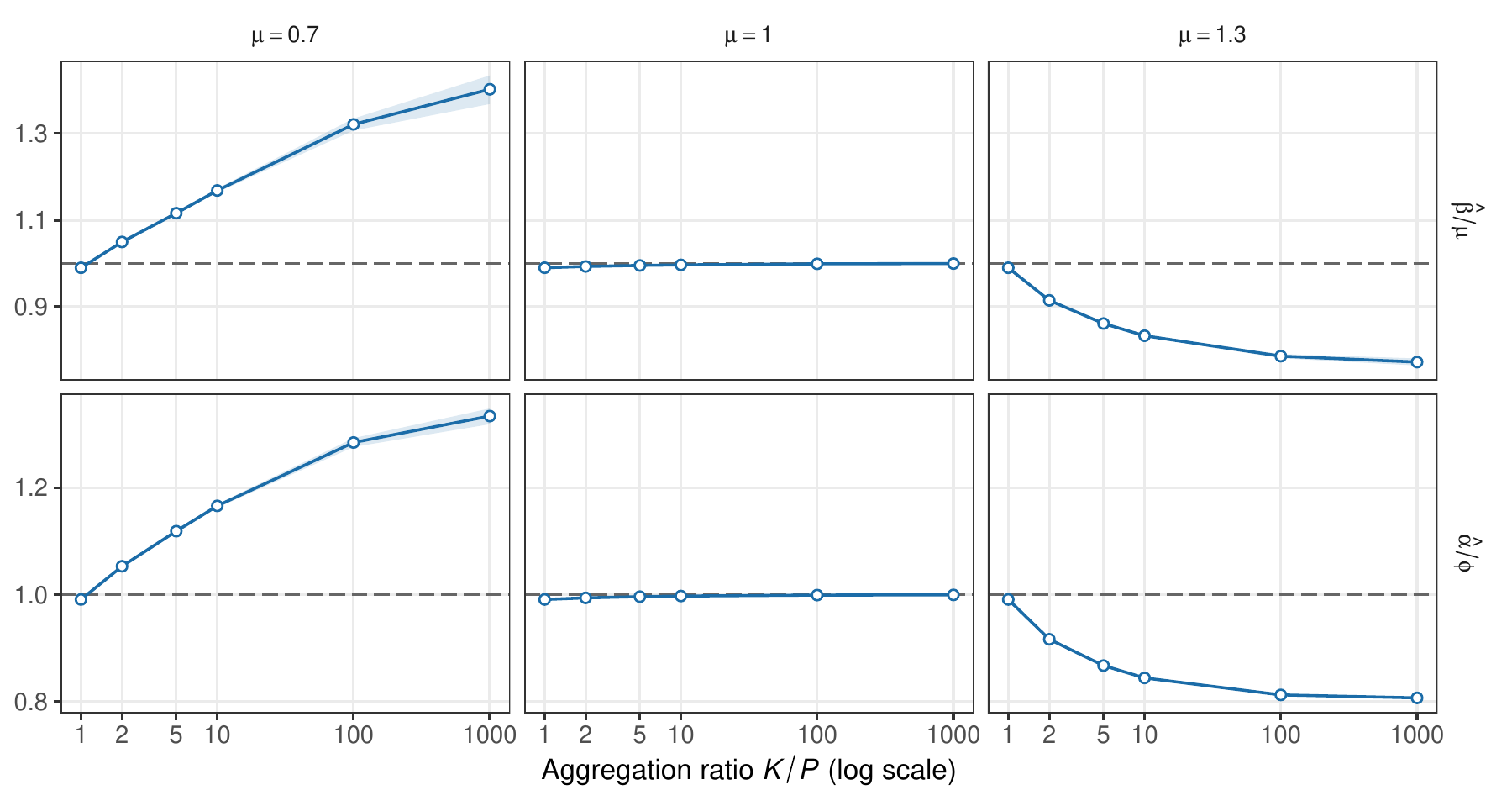}
\caption{Monte Carlo estimates of Model~\eqref{eq:lightOnIncomeMC} (log totals on log totals), $\sigma_\eps=0.1$. Lines/points: median across the 1000 replications; bands: 90\% intervals combining Monte Carlo and estimation uncertainty. First row: $\hat\beta/\mu$; second row: $\hat\alpha/\phi$. Columns: true elasticities $\mu=0.7,1,1.3$; horizontal axis: $K/P$ (log scale).}
\label{fig:mcBetaAlphaTotal}
\end{figure}

\begin{figure}[!htbp]
\centering
\includegraphics[width=\textwidth]{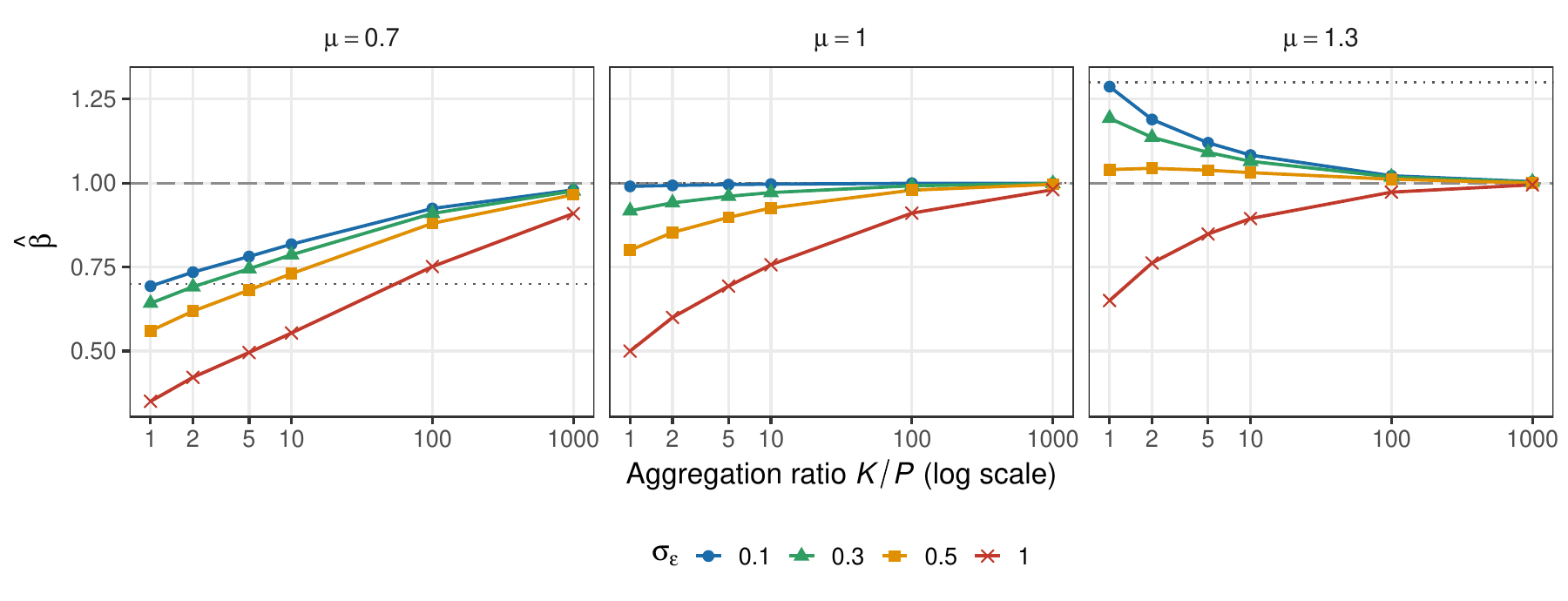}
\caption{Monte Carlo: reverse-regression attenuation and spatial aggregation in Model~\eqref{eq:lightOnIncomeMC} (log totals on log totals) as $\sigma_\eps$ varies. Columns are the true elasticity $\mu$; colours are the idiosyncratic-shock standard deviation $\sigma_\eps$. The dotted line is the true $\mu$ and the dashed line the aggregation fixed point $\hat\beta=1$. At the finest scale ($K/P=1$) $\hat\beta\approx\lambda\mu$ with $\lambda=\sigma_\nu^2/(\sigma_\nu^2+\sigma_\eps^2)$; under aggregation the slope is contracted toward one.}
\label{fig:mcNoiseNaive}
\end{figure}

\begin{figure}[!htbp]
\centering
\includegraphics[width=\textwidth]{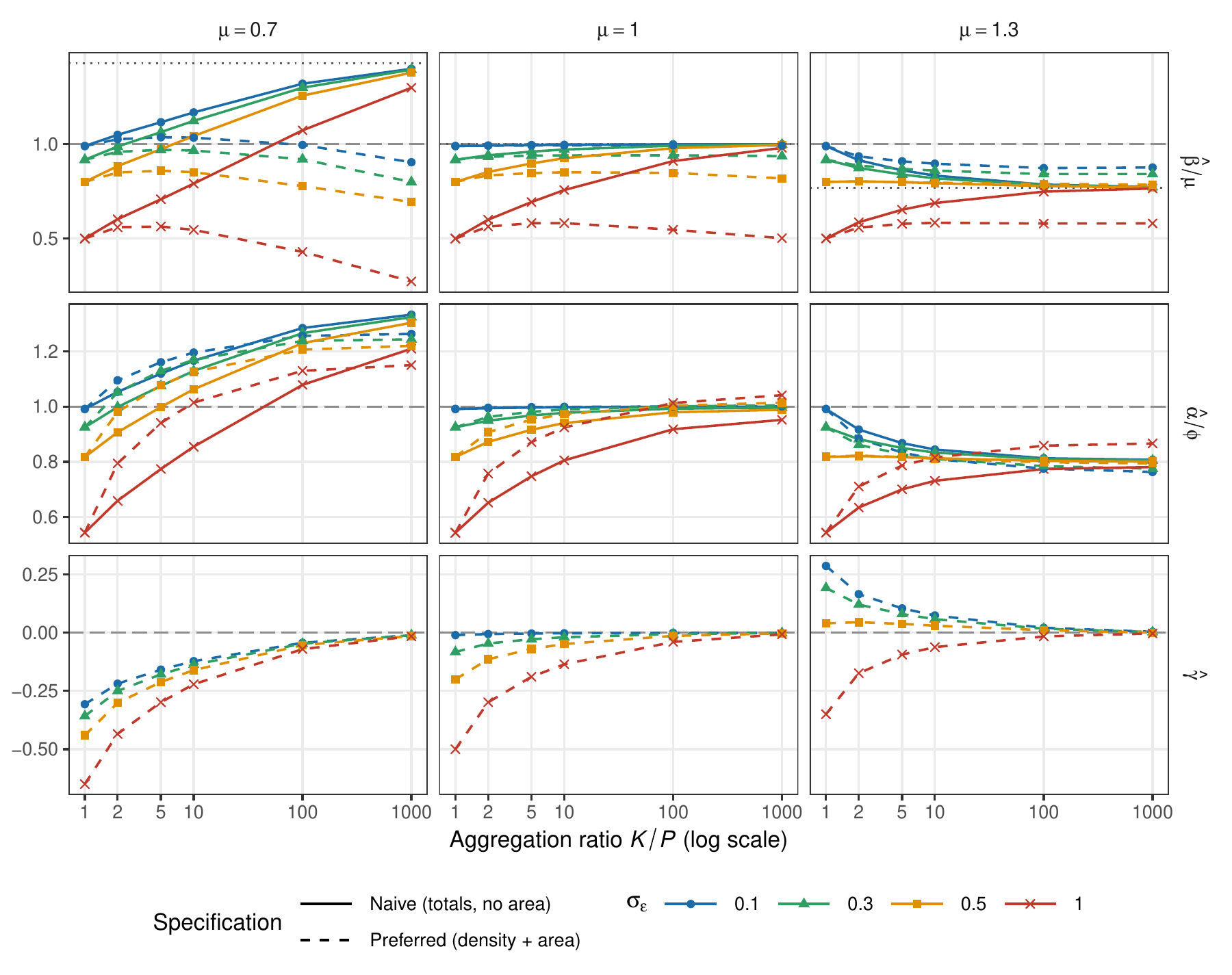}
\caption{Monte Carlo: reverse-regression attenuation and spatial aggregation as $\sigma_\eps$ varies, for the naive specification (log totals on log totals, solid lines) and the preferred specification (densities with an area control, Equation~\eqref{eq:mcModelDensityArea}, dashed lines). As in Figure~\ref{fig:mcBetaGammaSixPanel}, the rows report $\hat\beta/\mu$, $\hat\alpha/\phi$, and $\hat\gamma$; the area coefficient $\hat\gamma$ exists only for the preferred specification. Columns are the true elasticity $\mu$; colours are the idiosyncratic-shock standard deviation $\sigma_\eps$; averages over Monte Carlo replications, with unit area positively correlated with the elementary-location count (the benchmark regime of Figure~\ref{fig:mcBetaGammaSixPanel}). The long-dashed line marks the reference value ($1$ for the ratios, $0$ for $\hat\gamma$); the dotted line in the first row marks the aggregation fixed point $\hat\beta/\mu=1/\mu$. At the finest scale ($K/P=1$) $\hat\beta/\mu\approx\lambda$ with $\lambda=\sigma_\nu^2/(\sigma_\nu^2+\sigma_\eps^2)$; under aggregation the naive ratio is pulled toward $1/\mu$, while the density-and-area ratio stays close to $\lambda$, with a residual drift at the highest noise levels.}
\label{fig:mcNoise}
\end{figure}

\clearpage

\section{Additional Empirical Figures}
\label{app:empFigures}
This appendix collects descriptive and supplementary empirical figures referenced in Section~\ref{sec:empirical}: the map of NTL and economic density at the finest administrative level (Figure~\ref{fig:dataLowerAdmLevel}), the estimated area coefficient (Figure~\ref{fig:estimatedGamma_countries}), and the implied local elasticity of the Nonlinear Model (Figure~\ref{fig:elasticityByNTL}).

\begin{figure}[!hbtp]
\centering
\begin{subfigure}[b]{0.18\linewidth}
\includegraphics[width=\textwidth]{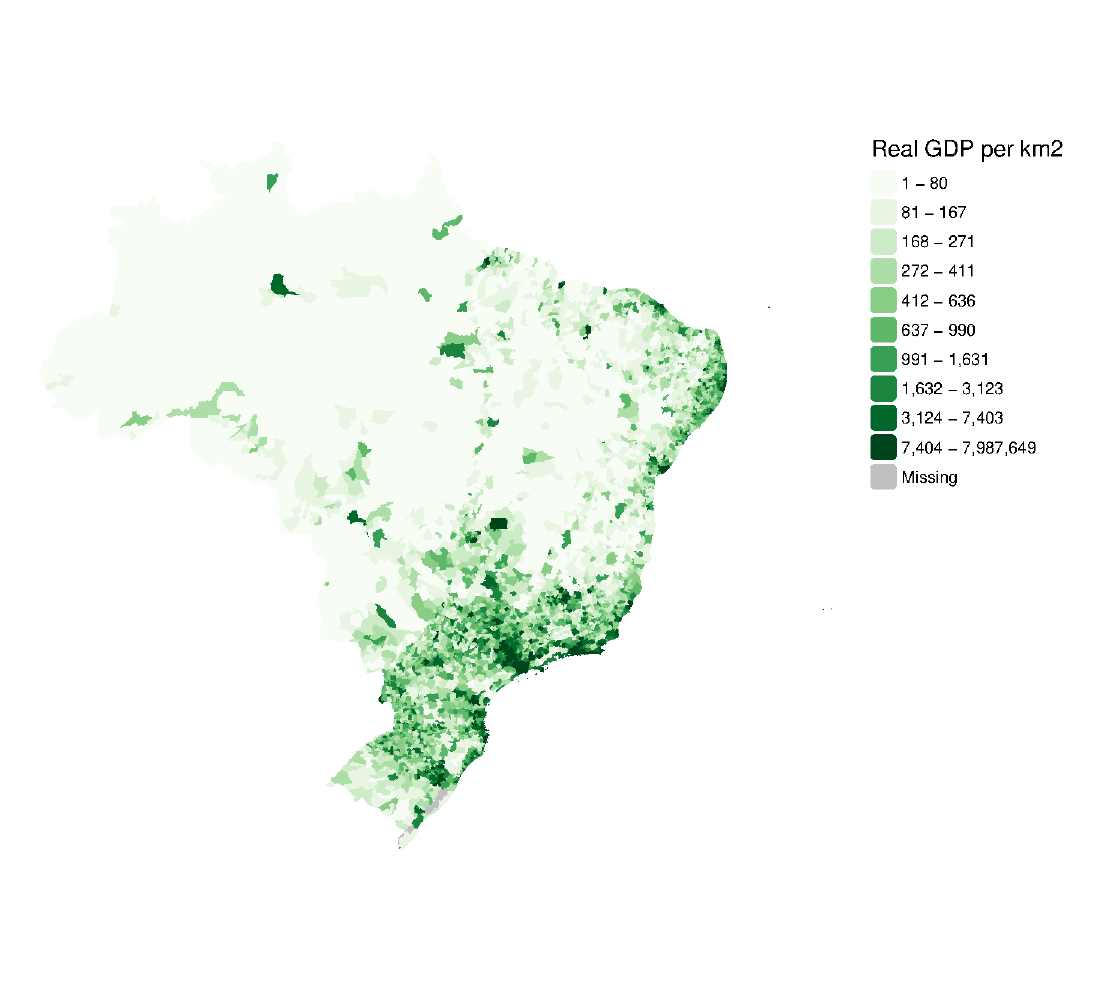}
\caption{Brazil: GDP}
\end{subfigure}
\begin{subfigure}[b]{0.15\linewidth}
\includegraphics[width=\textwidth]{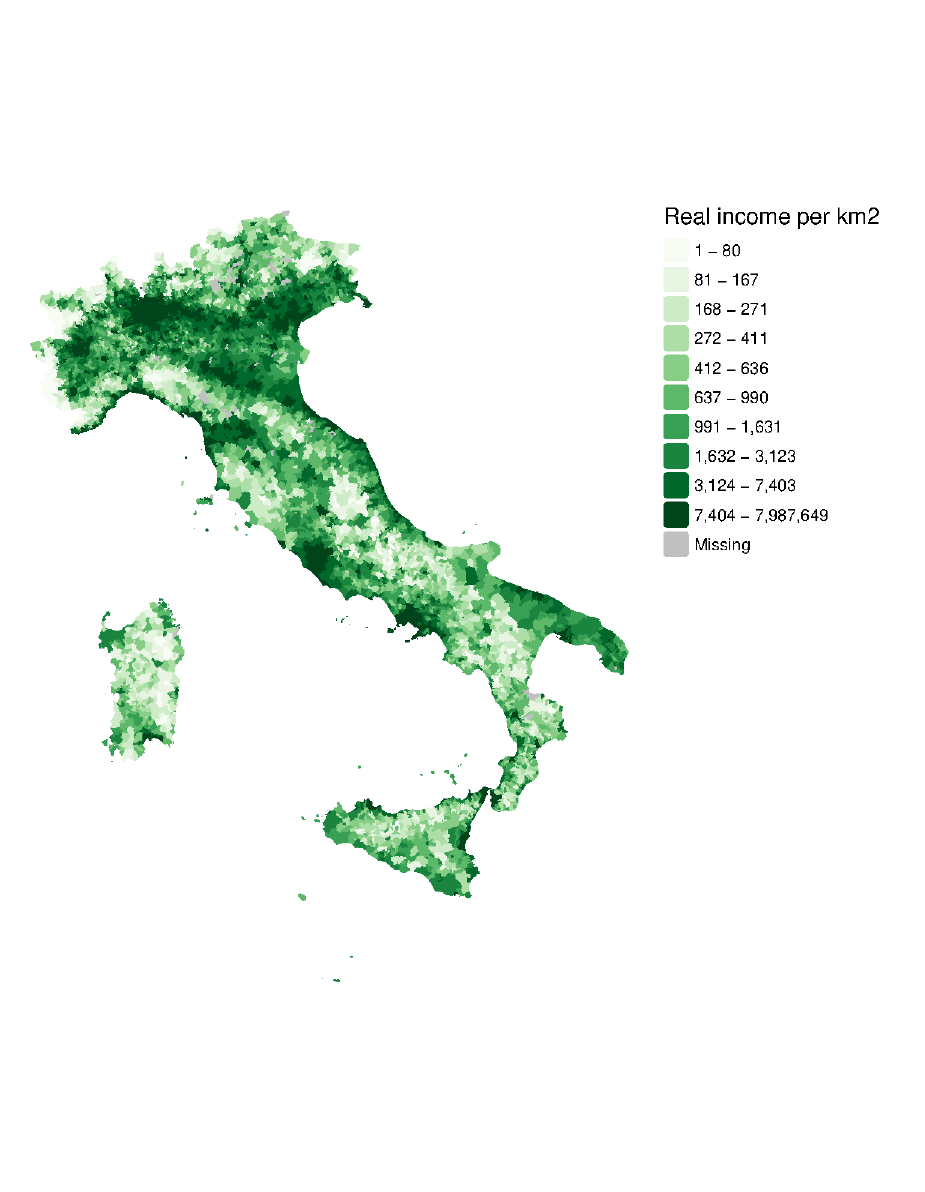}
\caption{Italy: income}
\end{subfigure}
\begin{subfigure}[b]{0.22\linewidth}
\includegraphics[width=\textwidth]{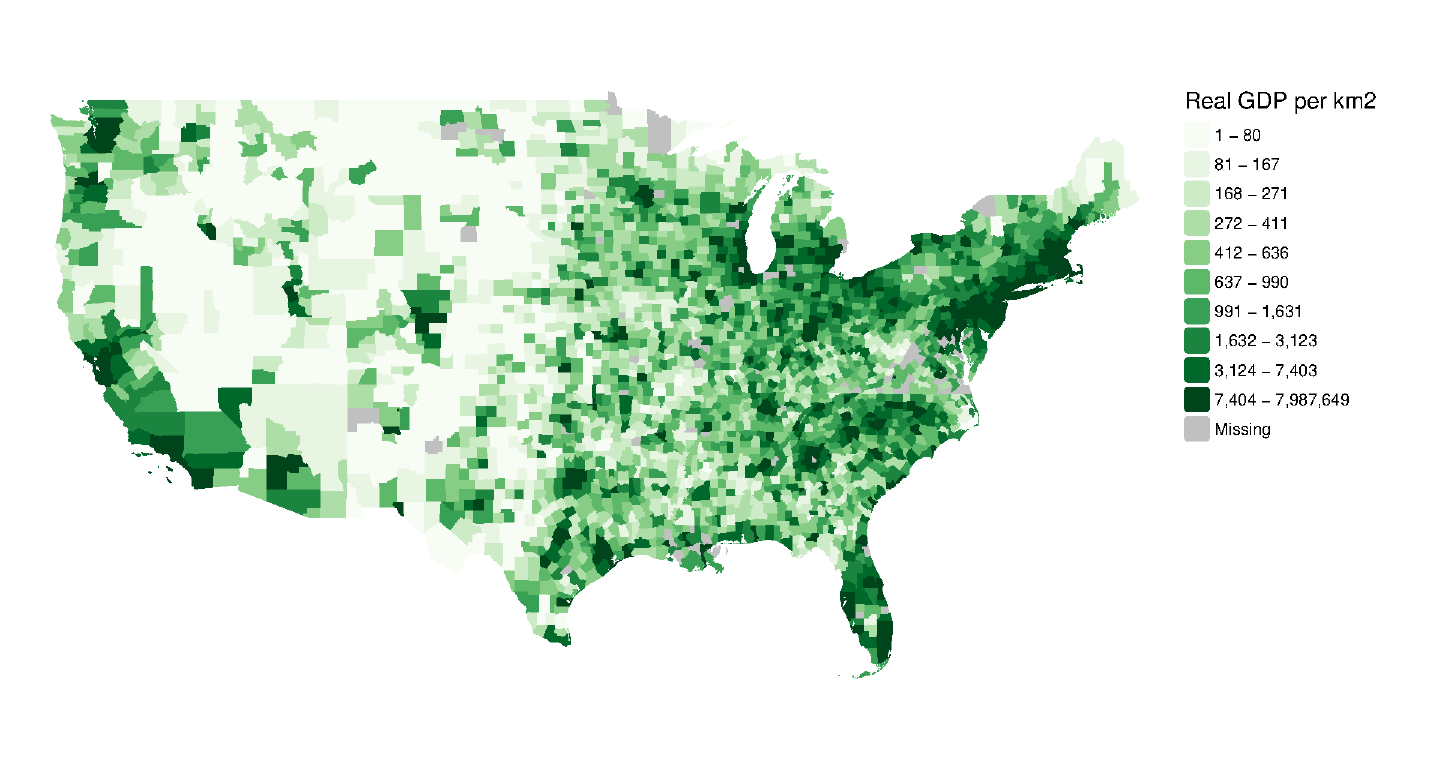}
\caption{U.S.: GDP}
\end{subfigure}
\begin{subfigure}[b]{0.16\linewidth}
\includegraphics[width=\textwidth]{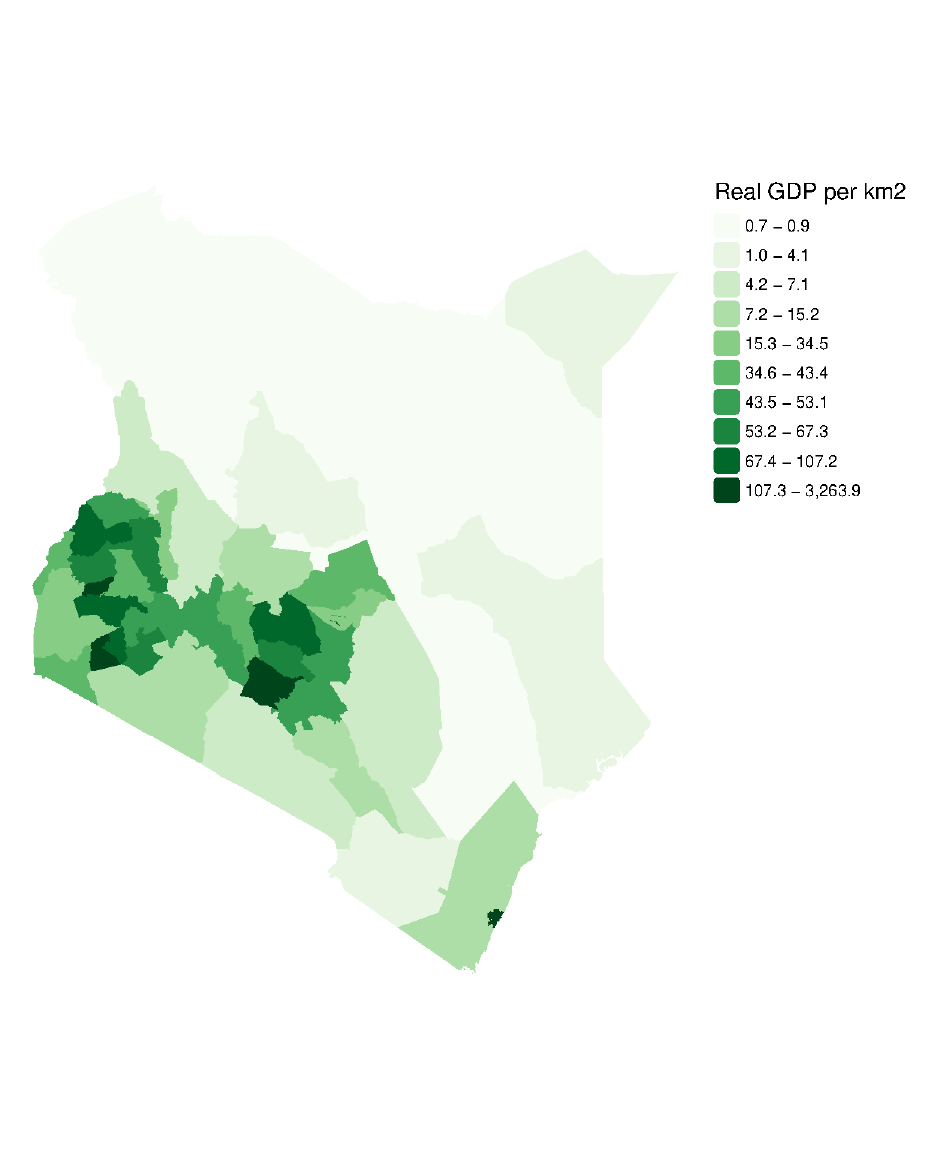}
\caption{Kenya: GDP}
\end{subfigure}
\begin{subfigure}[b]{0.24\linewidth}
\includegraphics[width=\textwidth]{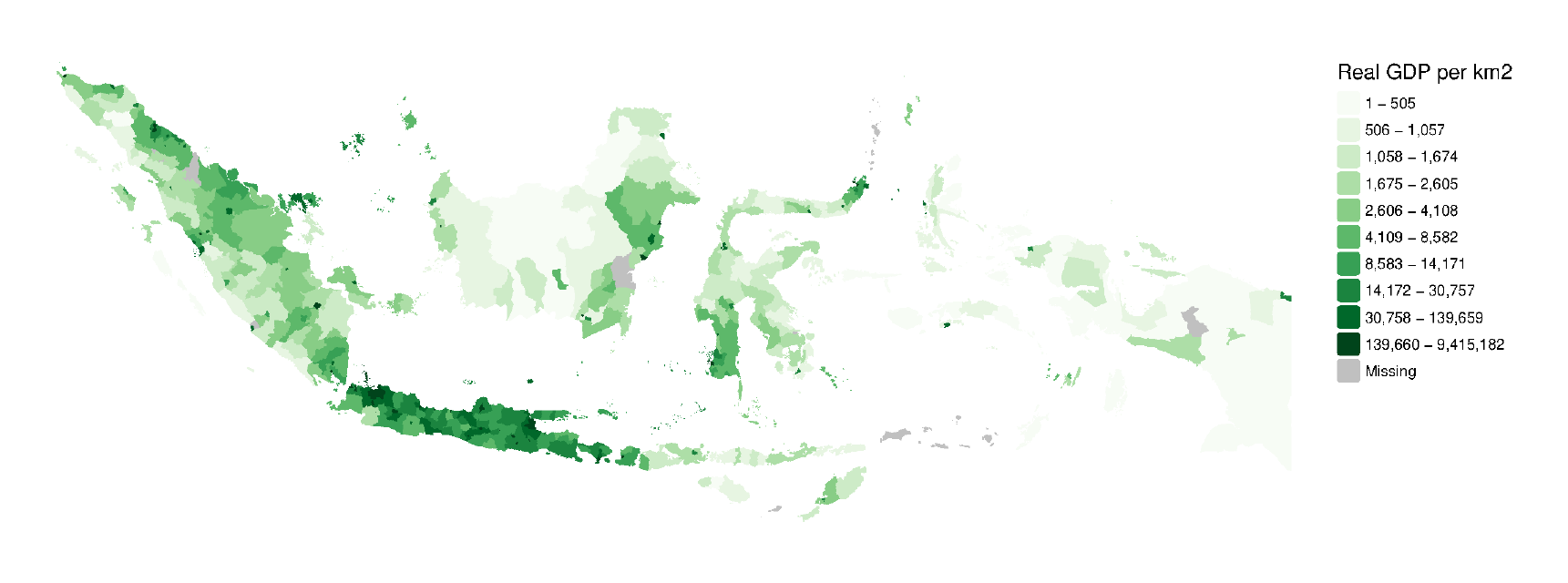}
\caption{Indonesia: GDP}
\end{subfigure}
\begin{subfigure}[b]{0.18\linewidth}
\includegraphics[width=\textwidth]{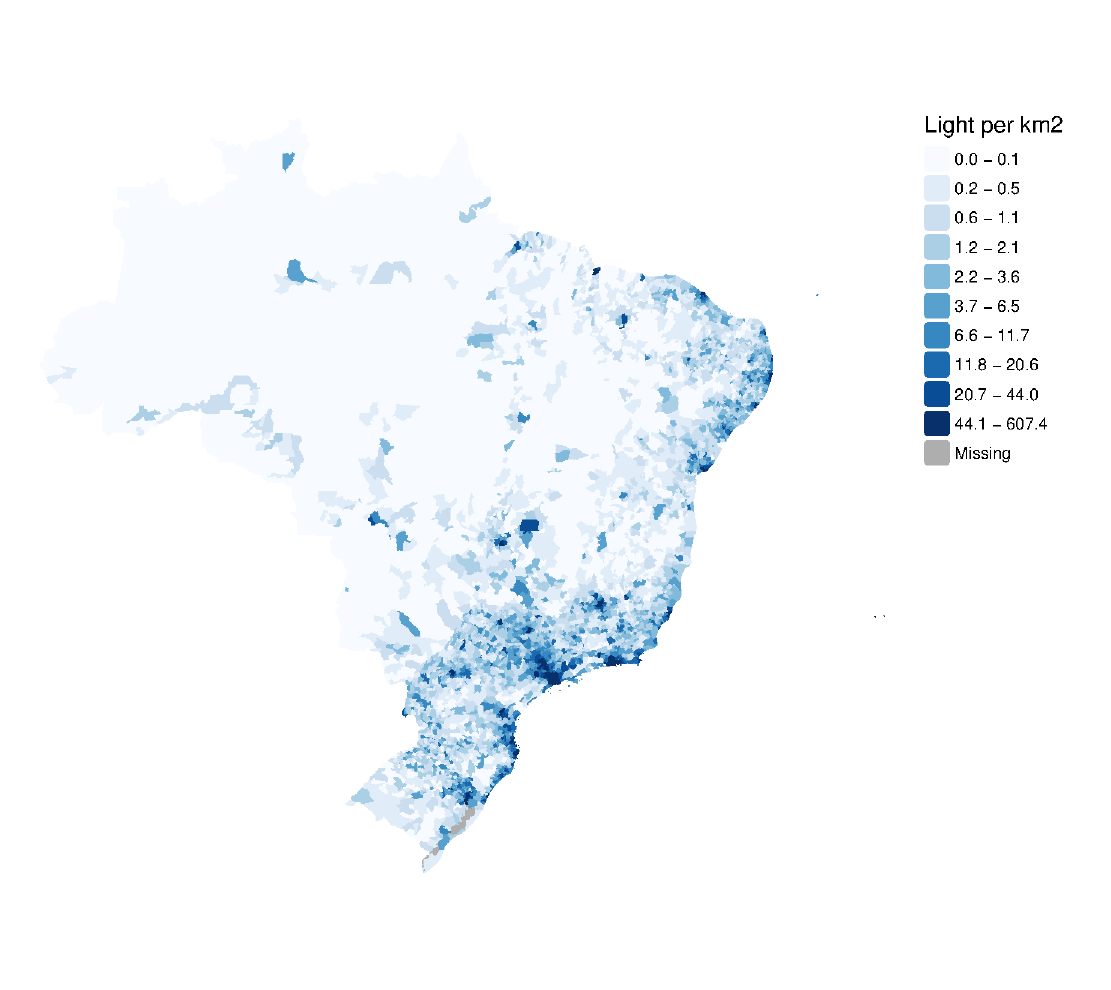}
\caption{Brazil: NTL}
\end{subfigure}
\begin{subfigure}[b]{0.15\linewidth}
\includegraphics[width=\textwidth]{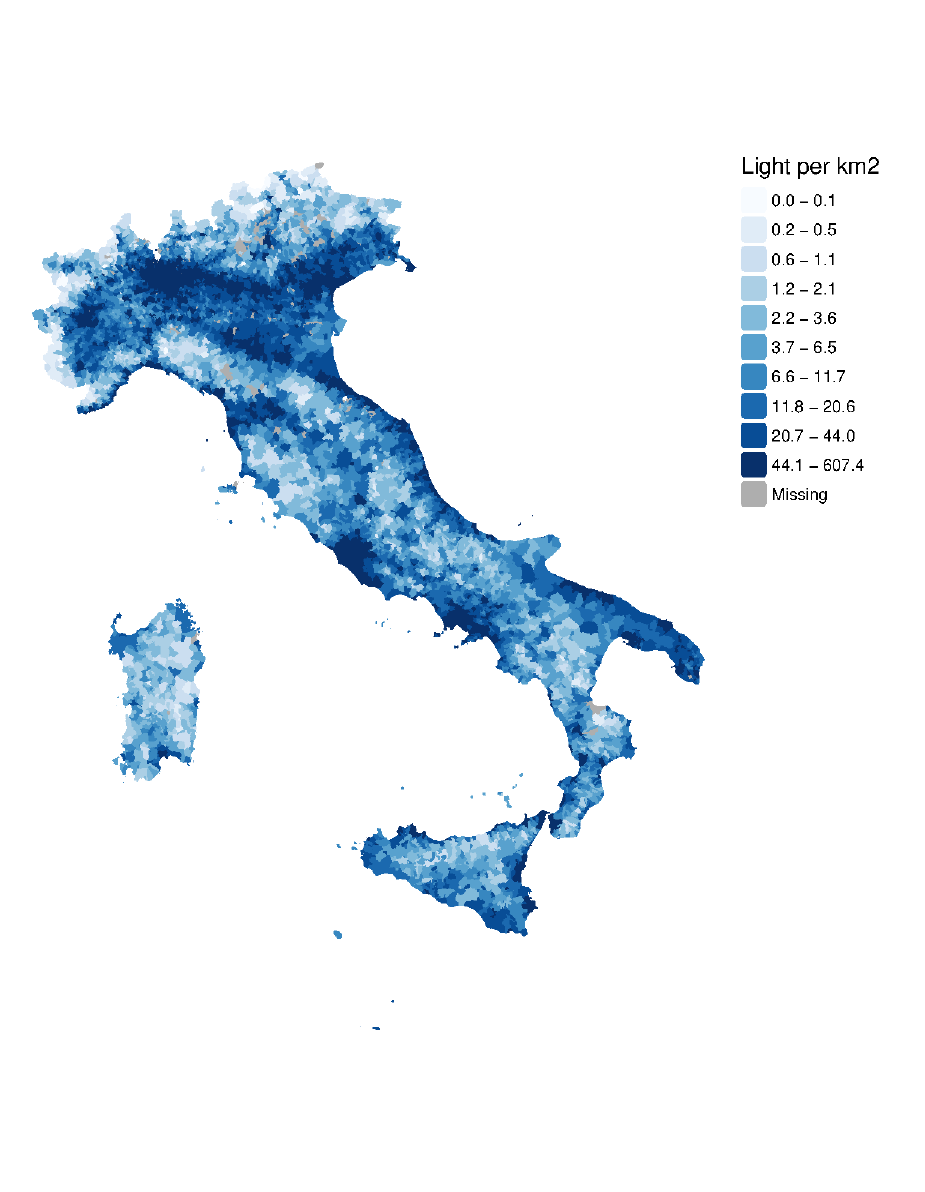}
\caption{Italy: NTL}
\end{subfigure}
\begin{subfigure}[b]{0.22\linewidth}
\includegraphics[width=\textwidth]{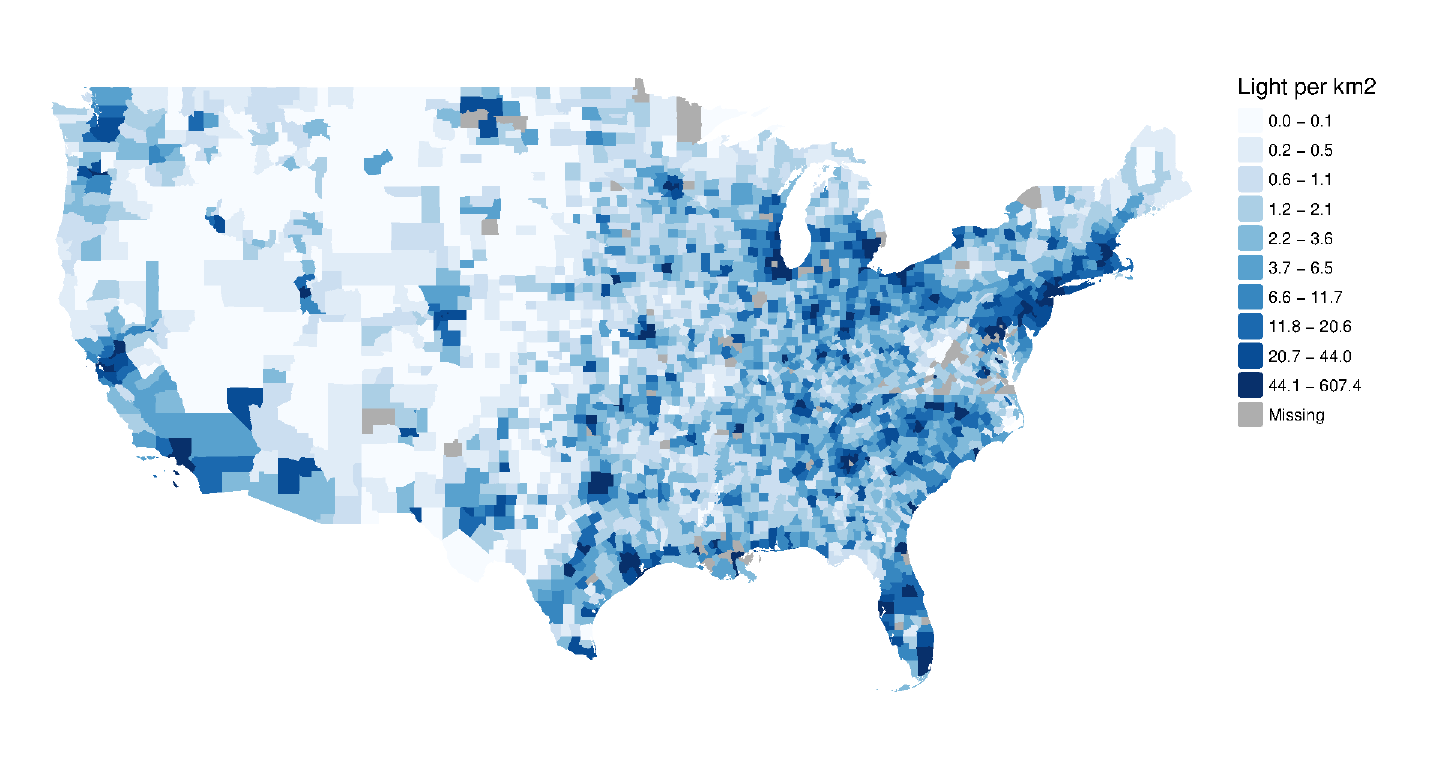}
\caption{U.S.: NTL}
\end{subfigure}
\begin{subfigure}[b]{0.16\linewidth}
\includegraphics[width=\textwidth]{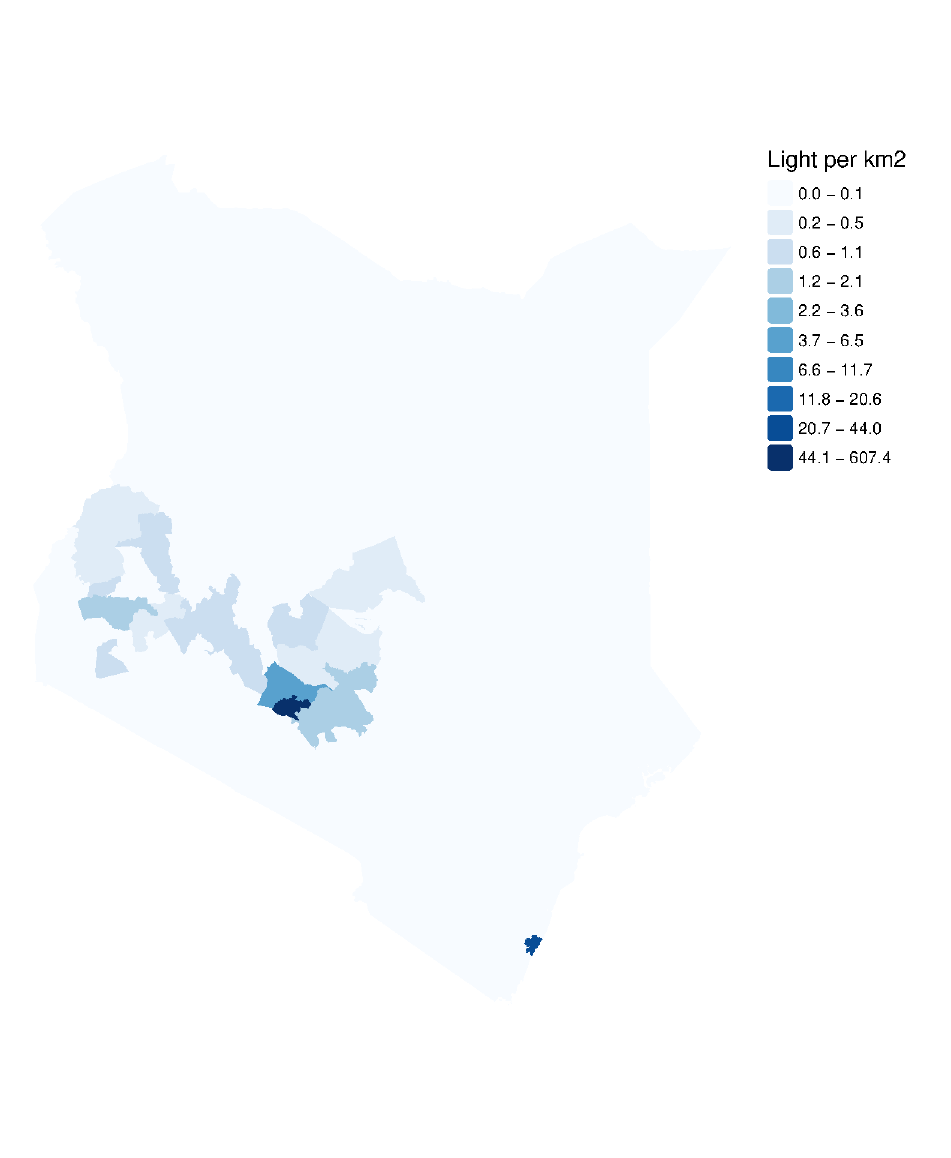}
\caption{Kenya: NTL}
\end{subfigure}
\begin{subfigure}[b]{0.24\linewidth}
\includegraphics[width=\textwidth]{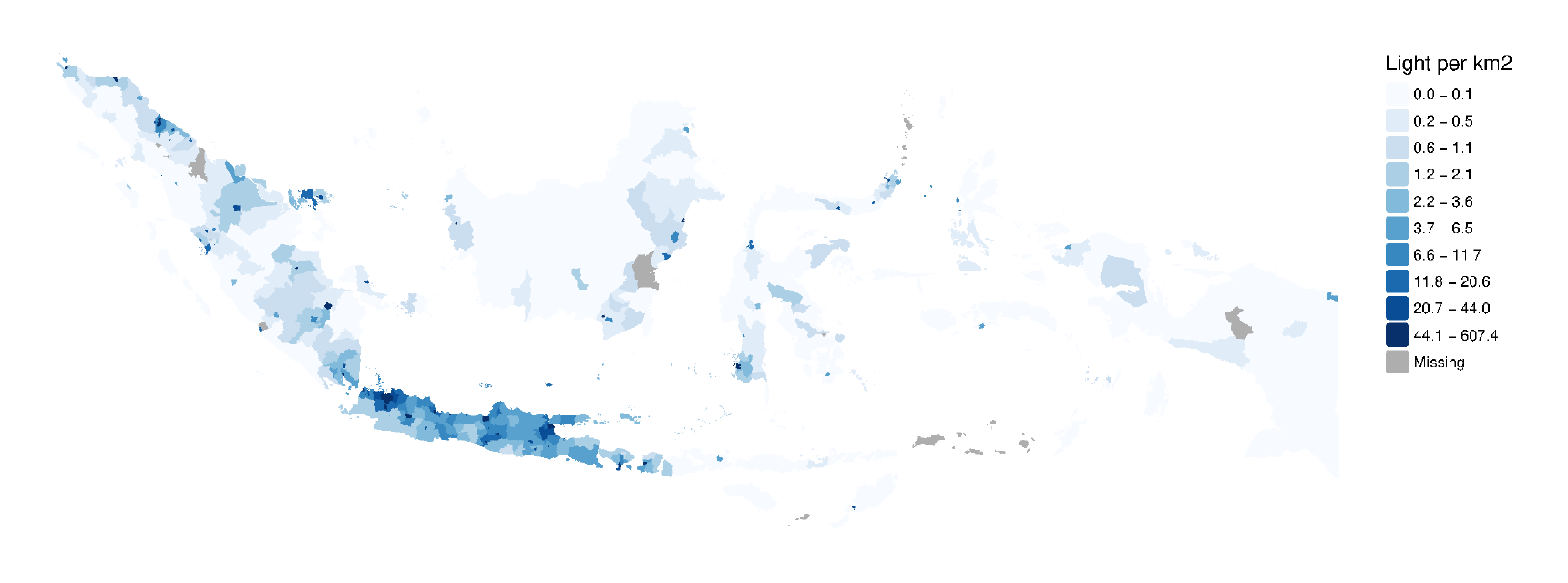}
\caption{Indonesia: NTL}
\end{subfigure}
\caption{NTL, real GDP, and income density per km$^2$ in 2019 at the lowest administrative level: municipalities for Brazil and Italy, counties for the United States and Kenya, and kabupaten/kota for Indonesia. NTL density uses a common cross-country colour scale; GDP/income density uses a common scale for Brazil, Italy, and the United States, and country-specific deciles for Kenya and Indonesia.}
\label{fig:dataLowerAdmLevel}
\begin{flushleft}
\textit{Sources: data from ARDECO, Agenzia delle Entrate, BEA, IBGE, KNBS, BPS/INDO-DAPOER, and World Development Indicators.}
\end{flushleft}
\end{figure}

\begin{figure}[!htbp]
\centering
\includegraphics[width=\textwidth]{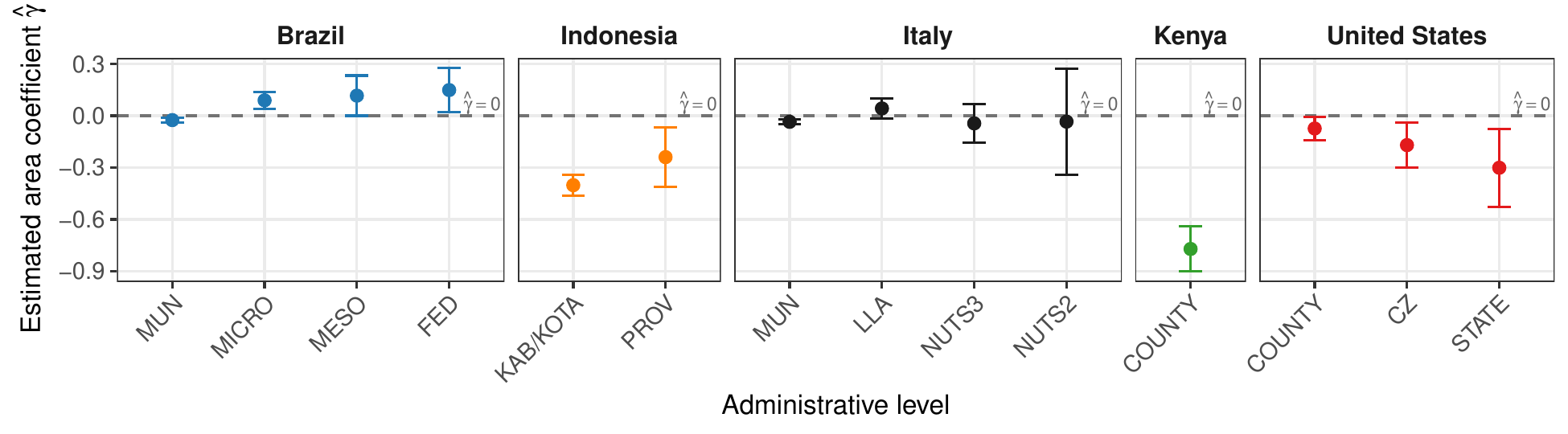}
\caption{Estimated area coefficient ($\gamma$) of Eq.~\eqref{eq:lightOnIncome} by country and administrative aggregation. Error bars represent 95\% confidence intervals based on standard errors clustered by spatial unit. Aggregations are labelled directly in each panel: Italy reports municipalities, local labour areas, NUTS 3 provinces, and NUTS 2 regions; Brazil reports municipalities, microregions, mesoregions, and federative units; the United States reports counties, commuting zones, and states; Kenya reports counties; Indonesia reports kabupaten/kota and provinces.}
\label{fig:estimatedGamma_countries}
\end{figure}

\begin{figure}[!htbp]
\centering
\includegraphics[width=\textwidth]{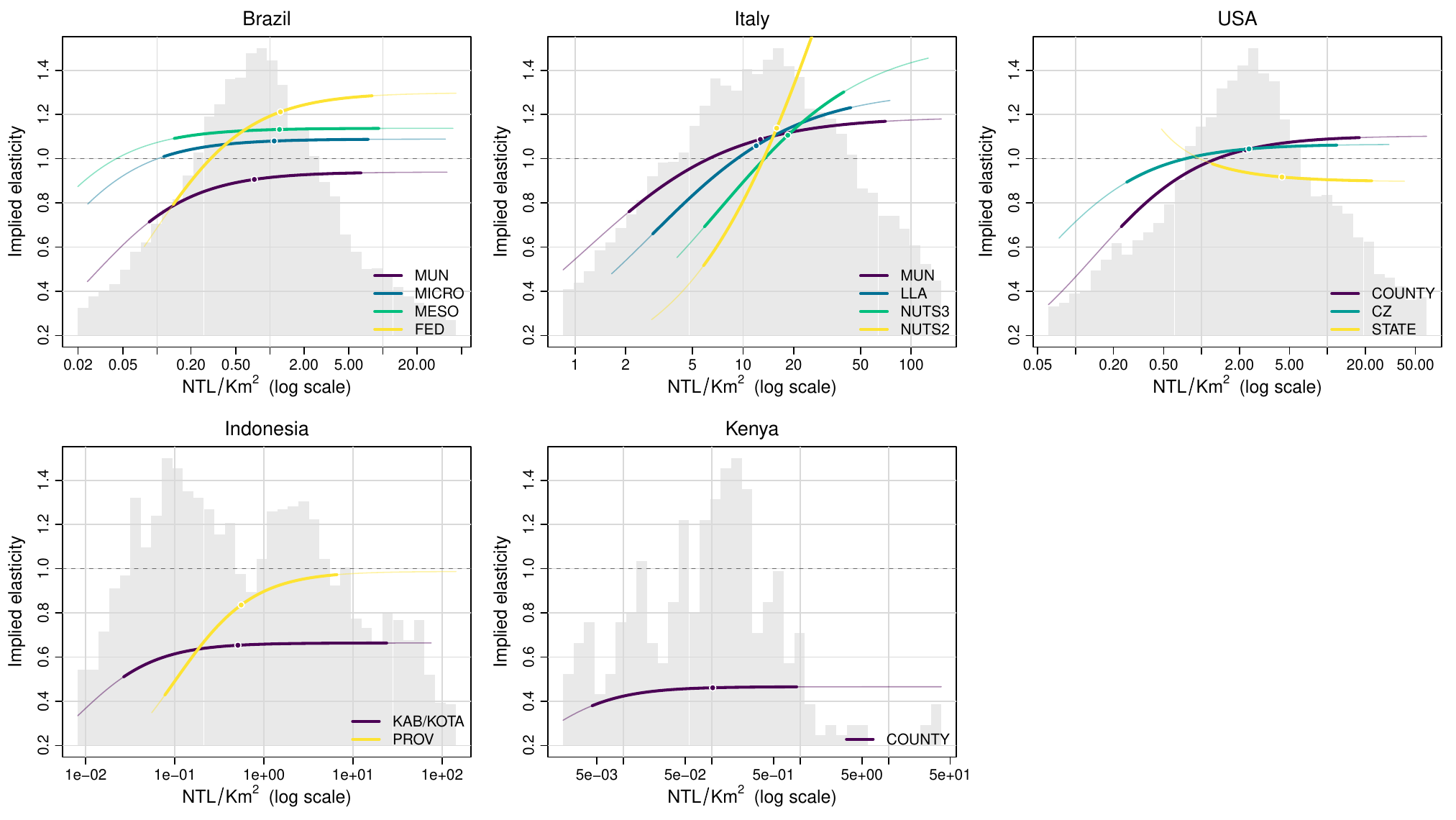}
\caption{Implied local NTL elasticity of GDP/income density from the Nonlinear Model, $\varepsilon(x)=\hat\beta\,x/(x+\hat\ell_0)$ with $x=NTL/\text{Km2}$, by country and administrative level. For each level the curve is drawn solid over the observed interdecile range (p10--p90) of NTL density and thin over the p2.5--p97.5 tails; the dot marks the median NTL density and the dashed line the unit-elastic benchmark. The shaded grey histogram in each panel shows the distribution of observed NTL density at the finest administrative level, indicating where the bulk of the units lie.}
\label{fig:elasticityByNTL}
\end{figure}

\end{document}

%% file: tables/table_montecarlo_oos.tex
\begin{table}[htbp]\centering
\footnotesize
\caption{Monte Carlo out-of-sample comparison of the empirical specifications. For each true elasticity $\mu$, area--size correlation regime, shock standard deviation $\sigma_\eps$, and aggregation ratio $K/P$, units are split into five folds within each replication; every specification is estimated on four folds and predicts the fifth, and the table reports the root mean squared out-of-fold prediction error (averages over 1000 or 100 replications depending on the block, $K=100{,}000$). Specifications: \emph{T} regresses log totals on log totals (Equation~\eqref{eq:lightOnIncomeMC}); \emph{D} log densities on log densities; \emph{D+A} adds the log-area control (the Baseline Model of Section~\ref{sec:personalIncome}, Equation~\eqref{eq:mcModelDensityArea}); \emph{D+A+$\ell_0$} replaces $\log(NTL_p/S_p)$ with $\log(NTL_p/S_p+\ell_0)$, with $\ell_0$ estimated on the training folds (the Nonlinear Model, Equation~\eqref{eq:mcModelNLM}). RMSEs are comparable across specifications because predicting $\log Y_p$ and $\log(Y_p/S_p)$ is equivalent given the observed $S_p$.}
\label{tab:mcOOS}
\begin{adjustbox}{width=\textwidth}
\begin{tabular}{lcc rrrr rrrr rrrr}
\hline\hline
 & & & \multicolumn{4}{c}{$\mu=0.7$} & \multicolumn{4}{c}{$\mu=1$} & \multicolumn{4}{c}{$\mu=1.3$}\\
\cline{4-7}\cline{8-11}\cline{12-15}
Area--size & $\sigma_\eps$ & $K/P$ & T & D & D+A & D+A+$\ell_0$ & T & D & D+A & D+A+$\ell_0$ & T & D & D+A & D+A+$\ell_0$ \\
\hline
Positive & 0.1 & 1 & \textbf{0.100} & 0.271 & \textbf{0.100} & \textbf{0.100} & \textbf{0.100} & \textbf{0.100} & \textbf{0.100} & \textbf{0.100} & \textbf{0.099} & 0.202 & \textbf{0.099} & \textbf{0.099} \\
 &  & 2 & 0.150 & 0.258 & \textbf{0.135} & \textbf{0.135} & \textbf{0.086} & \textbf{0.086} & \textbf{0.086} & \textbf{0.086} & 0.134 & 0.186 & \textbf{0.125} & \textbf{0.124} \\
 &  & 5 & 0.176 & 0.225 & \textbf{0.142} & \textbf{0.142} & \textbf{0.069} & \textbf{0.069} & \textbf{0.069} & \textbf{0.069} & 0.133 & 0.156 & \textbf{0.119} & \textbf{0.118} \\
 &  & 10 & 0.183 & 0.196 & \textbf{0.137} & \textbf{0.137} & \textbf{0.057} & \textbf{0.057} & \textbf{0.057} & \textbf{0.057} & 0.123 & 0.133 & \textbf{0.108} & \textbf{0.107} \\
 &  & 100 & 0.146 & 0.108 & \textbf{0.094} & \textbf{0.094} & \textbf{0.027} & \textbf{0.027} & \textbf{0.027} & \textbf{0.028} & 0.071 & 0.068 & \textbf{0.062} & \textbf{0.063} \\
 &  & 1000 & 0.079 & 0.050 & \textbf{0.049} & 0.051 & \textbf{0.011} & \textbf{0.011} & \textbf{0.012} & 0.013 & 0.031 & \textbf{0.029} & \textbf{0.029} & 0.031 \\
\addlinespace
Positive & 0.5 & 1 & \textbf{0.447} & 0.582 & \textbf{0.447} & \textbf{0.447} & \textbf{0.447} & 0.471 & \textbf{0.447} & \textbf{0.447} & \textbf{0.447} & 0.448 & \textbf{0.447} & \textbf{0.447} \\
 &  & 2 & 0.434 & 0.522 & \textbf{0.420} & \textbf{0.420} & 0.398 & 0.409 & \textbf{0.395} & \textbf{0.395} & \textbf{0.386} & 0.388 & \textbf{0.386} & \textbf{0.386} \\
 &  & 5 & 0.399 & 0.430 & \textbf{0.360} & \textbf{0.360} & 0.325 & 0.326 & \textbf{0.317} & \textbf{0.317} & \textbf{0.304} & 0.307 & \textbf{0.304} & \textbf{0.304} \\
 &  & 10 & 0.371 & 0.365 & \textbf{0.314} & \textbf{0.314} & 0.274 & 0.270 & \textbf{0.265} & \textbf{0.265} & \textbf{0.251} & 0.253 & \textbf{0.251} & \textbf{0.251} \\
 &  & 100 & 0.256 & 0.188 & \textbf{0.177} & \textbf{0.177} & 0.139 & 0.131 & \textbf{0.130} & \textbf{0.130} & \textbf{0.119} & 0.120 & \textbf{0.120} & 0.120 \\
 &  & 1000 & 0.135 & \textbf{0.082} & \textbf{0.082} & 0.085 & 0.059 & \textbf{0.055} & 0.056 & 0.058 & \textbf{0.049} & 0.050 & 0.051 & 0.054 \\
\addlinespace
Independent & 0.1 & 1 & \textbf{0.099} & 0.271 & \textbf{0.099} & \textbf{0.099} & \textbf{0.100} & \textbf{0.100} & \textbf{0.100} & \textbf{0.100} & \textbf{0.100} & 0.202 & \textbf{0.100} & \textbf{0.100} \\
 &  & 2 & \textbf{0.149} & 0.265 & \textbf{0.149} & \textbf{0.149} & \textbf{0.086} & \textbf{0.087} & \textbf{0.086} & \textbf{0.086} & \textbf{0.134} & 0.185 & \textbf{0.134} & \textbf{0.134} \\
 &  & 5 & \textbf{0.176} & 0.251 & \textbf{0.176} & \textbf{0.176} & \textbf{0.069} & \textbf{0.069} & \textbf{0.069} & \textbf{0.069} & \textbf{0.134} & 0.158 & \textbf{0.134} & \textbf{0.134} \\
 &  & 10 & \textbf{0.183} & 0.235 & \textbf{0.183} & \textbf{0.182} & \textbf{0.057} & \textbf{0.057} & \textbf{0.057} & \textbf{0.057} & \textbf{0.123} & 0.137 & \textbf{0.123} & \textbf{0.123} \\
 &  & 100 & 0.146 & 0.158 & 0.146 & \textbf{0.145} & \textbf{0.027} & \textbf{0.027} & \textbf{0.027} & \textbf{0.028} & \textbf{0.071} & 0.074 & \textbf{0.072} & \textbf{0.072} \\
 &  & 1000 & \textbf{0.078} & 0.080 & \textbf{0.078} & 0.080 & \textbf{0.012} & \textbf{0.011} & \textbf{0.012} & 0.014 & \textbf{0.032} & \textbf{0.032} & \textbf{0.032} & 0.034 \\
\addlinespace
Negative & 0.1 & 1 & \textbf{0.099} & \textbf{0.099} & \textbf{0.099} & \textbf{0.099} & \textbf{0.100} & \textbf{0.100} & \textbf{0.100} & \textbf{0.100} & \textbf{0.100} & \textbf{0.100} & \textbf{0.100} & \textbf{0.100} \\
 &  & 2 & 0.150 & 0.150 & \textbf{0.139} & \textbf{0.139} & \textbf{0.086} & \textbf{0.086} & \textbf{0.086} & \textbf{0.086} & 0.134 & 0.134 & \textbf{0.127} & \textbf{0.127} \\
 &  & 5 & 0.176 & 0.176 & \textbf{0.149} & \textbf{0.149} & \textbf{0.069} & \textbf{0.069} & \textbf{0.069} & \textbf{0.069} & 0.133 & 0.133 & \textbf{0.122} & \textbf{0.122} \\
 &  & 10 & 0.183 & 0.183 & \textbf{0.150} & \textbf{0.150} & \textbf{0.057} & \textbf{0.057} & \textbf{0.057} & \textbf{0.057} & 0.123 & 0.123 & \textbf{0.112} & \textbf{0.112} \\
 &  & 100 & 0.147 & 0.147 & \textbf{0.131} & \textbf{0.131} & \textbf{0.027} & \textbf{0.027} & \textbf{0.027} & \textbf{0.027} & 0.071 & 0.071 & \textbf{0.068} & \textbf{0.068} \\
 &  & 1000 & 0.080 & 0.080 & \textbf{0.077} & 0.084 & \textbf{0.011} & \textbf{0.011} & \textbf{0.012} & 0.014 & \textbf{0.031} & \textbf{0.031} & \textbf{0.032} & 0.036 \\
\hline\hline
\multicolumn{15}{l}{\parbox{1.25\textwidth}{\textit{Note}: bold marks the lowest RMSE in each cell (ties within $0.0005$ all in bold). ``Area--size'' is the correlation between unit area and the firm count in the DGP\@. A totals specification with the area control is not reported: it shares the regressor span of D+A, so its out-of-fold predictions differ from those of D+A exactly by the observed $\log S_p$ and the prediction errors coincide.}}\\
\end{tabular}
\end{adjustbox}
\end{table}

%% file: tables/table_descriptiveStats.tex
\begin{table}[ht]
\centering
\caption{Descriptive statistics by country and aggregation level. The number of units is the effective number of spatial units in the analytical sample after merging economic, area, and NTL data. Economic activity density is expressed in thousands of international dollars per km$^2$ using World Bank PPP conversion factors for GDP (PA.NUS.PPP), and refers to GDP for Brazil, the United States, Kenya, and Indonesia, and to personal income for Italy. SD denotes the standard deviation; CV denotes the coefficient of variation, defined as the ratio between the standard deviation and the mean.}
\label{tab:statDescr}
\resizebox{\textwidth}{!}{%
\begin{tabular}{llrrrrrrrrrr}
  \hline
Country & Level & \shortstack{Number\\of units} & \shortstack{Area\\mean} & \shortstack{GDP/income\\mean} & \shortstack{GDP/income\\median} & \shortstack{GDP/income\\SD} & \shortstack{GDP/income\\CV} & \shortstack{NTL\\mean} & \shortstack{NTL\\median} & \shortstack{NTL\\SD} & \shortstack{NTL\\CV} \\
  \hline
BRA & MUN & 5570 & 1525.5 & 1776.97 & 183.24 & 13908.34 & 7.83 & 4.59 & 0.73 & 20.49 & 4.47 \\
   & MICRO & 558 & 15228.2 & 1877.89 & 269.76 & 9535.72 & 5.08 & 4.43 & 1.09 & 13.75 & 3.10 \\
   & MESO & 137 & 62024.3 & 1766.00 & 301.49 & 6043.59 & 3.42 & 4.28 & 1.21 & 10.15 & 2.37 \\
   & FED & 27 & 315200.7 & 1517.14 & 434.29 & 3804.01 & 2.51 & 3.90 & 1.24 & 8.77 & 2.25 \\
  ITA & MUN & 7802 & 38.0 & 5803.82 & 1808.06 & 12503.71 & 2.15 & 27.74 & 12.66 & 43.18 & 1.56 \\
   & LLA & 609 & 495.7 & 3600.31 & 1858.38 & 5222.77 & 1.45 & 19.17 & 11.97 & 20.96 & 1.09 \\
   & NUTS3 & 107 & 2816.1 & 5412.43 & 3100.31 & 8467.77 & 1.56 & 24.46 & 18.44 & 25.82 & 1.06 \\
   & NUTS2 & 21 & 14348.5 & 3439.80 & 2841.70 & 2404.35 & 0.70 & 16.46 & 15.84 & 9.34 & 0.57 \\
  USA & COUNTY & 3014 & 2821.1 & 7347.74 & 543.36 & 148013.45 & 20.14 & 8.78 & 2.24 & 27.71 & 3.15 \\
   & CZ & 706 & 12941.5 & 2347.99 & 591.20 & 9808.87 & 4.18 & 5.20 & 2.38 & 8.93 & 1.72 \\
   & STATE & 51 & 182905.2 & 18326.89 & 1712.60 & 99893.86 & 5.45 & 17.52 & 4.35 & 69.24 & 3.95 \\
  IDN & KAB/KOTA & 487 & 3833.0 & 22414.72 & 1025.58 & 128549.91 & 5.74 & 8.54 & 0.51 & 24.56 & 2.88 \\
   & PROV & 34 & 55787.6 & 23481.03 & 886.85 & 122347.34 & 5.21 & 6.23 & 0.56 & 26.67 & 4.28 \\
  KEN & COUNTY & 47 & 12581.8 & 3166.08 & 728.90 & 11576.39 & 3.66 & 2.06 & 0.09 & 8.96 & 4.35 \\
   \hline
\end{tabular}%
}
\end{table}

%% file: tables/table_deltaRhoEmpirics.tex
\begin{table}[!htbp]\centering
\small
\caption{Cell-level variance of log radiance and empirical projection coefficients of Equation~\eqref{eq:mcDeltaRhoDef} by country and administrative level. The elementary units are the lit Black Marble grid cells ($\approx 500$~m): $\sigma^2_n$ is the cross-sectional variance of log radiance across all lit cells within the estimation sample of each country. All cell-level statistics are computed on the 2012 and 2019 rasters (2013 and 2019 for Kenya) and averaged; $P$ is the average number of units per year in the estimation samples of Section~\ref{sec:personalIncome}.}
\label{tab:deltaRhoEmpirics}
\scriptsize{
\begin{tabular}{lclrrr}
\hline\hline
Country & $\sigma^2_n$ & Level & $P$ & $\delta_P$ & $\rho_P$ \\
\hline
Brazil        & 1.90 & MUN      & 5569 & 0.74 & 0.56 \\
              &      & MICRO    &  558 & 0.70 & 0.02 \\
              &      & MESO     &  137 & 0.74 & $-$4.37 \\
              &      & STATE    &   27 & 0.94 & $-$0.12 \\
\addlinespace
Italy         & 1.57 & MUN      & 7802 & 0.54 & 0.38 \\
              &      & LLA      &  609 & 0.73 & 0.27 \\
              &      & NUTS3    &  107 & 0.65 & $-$0.29 \\
              &      & NUTS2    &   20 & 0.78 & $-$0.36 \\
\addlinespace
United States & 1.16 & COUNTY   & 3014 & 0.74 & 0.54 \\
              &      & CZ       &  706 & 0.84 & 0.58 \\
              &      & STATE    &   51 & 1.04 & 0.09 \\
\addlinespace
Indonesia     & 0.90 & KAB/KOTA &  487 & 0.73 & 2.85 \\
              &      & PROV     &   34 & 0.85 & 4.35 \\
\addlinespace
Kenya         & 1.15 & COUNTY   &   47 & 0.77 & 0.23 \\
\hline\hline
\end{tabular}
}
\end{table}

%% file: tables/table_cvCrossSectional.tex
\begin{tabular}{llcccccccccc}
  \hline
Country & Level & \shortstack{$R^2_{\text{oos}}$\\(BM)} & \shortstack{$R^2_{\text{oos}}$\\(NLM)} & \shortstack{Gain\\NLM--BM} & \shortstack{$R^2_{\text{oos}}$\\(NIM)} & \shortstack{$\Delta R^2_{\text{oos}}$\\(BM)} & \shortstack{$\Delta R^2_{\text{oos}}$\\(NLM)} & \shortstack{RMSE\\(BM)} & \shortstack{RMSE\\(NLM)} & \shortstack{$L_4$\\(BM)} & \shortstack{$L_4$\\(NLM)} \\ 
  \hline
Brazil & MUN & 0.867 & 0.873 & 0.006 & 0.319 & 0.549 & 0.555 & 0.607 & 0.593 & 0.880 & 0.831 \\ 
  Brazil & MICRO & 0.938 & 0.940 & 0.002 & 0.486 & 0.453 & 0.454 & 0.440 & 0.434 & 0.630 & 0.631 \\ 
  Brazil & MESO & 0.959 & 0.957 & -0.002 & 0.600 & 0.359 & 0.357 & 0.374 & 0.382 & 0.565 & 0.615 \\ 
  Brazil & FED & 0.969 & 0.978 & 0.008 & 0.380 & 0.590 & 0.598 & 0.297 & 0.255 & 0.419 & 0.352 \\ 
  Indonesia & KAB/KOTA & 0.909 & 0.910 & 0.001 & 0.671 & 0.238 & 0.239 & 0.653 & 0.648 & 0.942 & 0.944 \\ 
  Indonesia & PROV & 0.943 & 0.947 & 0.004 & 0.363 & 0.580 & 0.584 & 0.410 & 0.396 & 0.627 & 0.558 \\ 
  Italy & MUN & 0.854 & 0.866 & 0.012 & 0.159 & 0.695 & 0.707 & 0.566 & 0.542 & 0.858 & 0.786 \\ 
  Italy & LLA & 0.859 & 0.877 & 0.018 & 0.019 & 0.840 & 0.858 & 0.416 & 0.388 & 0.638 & 0.568 \\ 
  Italy & NUTS3 & 0.860 & 0.881 & 0.021 & 0.294 & 0.567 & 0.588 & 0.339 & 0.313 & 0.514 & 0.431 \\ 
  Italy & NUTS2 & 0.506 & 0.655 & 0.150 & -0.095 & 0.601 & 0.751 & 0.496 & 0.414 & 0.730 & 0.572 \\ 
  Kenya & COUNTY & 0.951 & 0.950 & -0.001 & 0.846 & 0.105 & 0.104 & 0.431 & 0.434 & 0.543 & 0.552 \\ 
  USA & COUNTY & 0.886 & 0.898 & 0.012 & 0.188 & 0.698 & 0.710 & 0.589 & 0.557 & 1.138 & 1.090 \\ 
  USA & CZ & 0.859 & 0.863 & 0.004 & 0.034 & 0.825 & 0.829 & 0.617 & 0.608 & 1.350 & 1.336 \\ 
  USA & STATE & 0.836 & 0.834 & -0.003 & 0.564 & 0.272 & 0.270 & 0.643 & 0.648 & 1.130 & 1.137 \\ 
   \hline
\end{tabular}

%% file: tables/table_pooledFinestPrediction.tex
\begin{tabular}{llcccccc}
  \hline
Calibration & Model & Mean error & MAE & RMSE & P10 & Median & P90 \\
  \hline
  Country-year FE & BM & 0.000 & 0.449 & 0.592 & -0.713 & -0.005 & 0.708 \\
  Country-year FE & NLM & 0.000 & 0.433 & 0.573 & -0.688 & 0.005 & 0.673 \\
  Common intercept & BM & 0.000 & 0.478 & 0.629 & -0.762 & 0.001 & 0.744 \\
  Common intercept & NLM & 0.000 & 0.470 & 0.620 & -0.750 & 0.006 & 0.728 \\
  Aggregate country-year intercept & BM & -0.541 & 0.655 & 0.803 & -1.255 & -0.547 & 0.168 \\
   Aggregate country-year intercept & NLM & -0.954 & 0.984 & 1.117 & -1.658 & -0.944 & -0.272 \\
   Aggregate country intercept & BM & -0.593 & 0.695 & 0.845 & -1.318 & -0.596 & 0.128 \\
   Aggregate country intercept & NLM & -1.089 & 1.111 & 1.239 & -1.803 & -1.078 & -0.398 \\

   \hline
\end{tabular}

%% file: tables/table_montecarlo_alpha_beta_total.tex
\begin{table}[!htbp]\centering
\footnotesize
\caption{Monte Carlo estimates of $\alpha$ and $\beta$ in the totals specification without controls, $\log Y_p = \alpha + \beta \log NTL_p + \epsilon_p$ (the specification covered by Theorem~\ref{teo:aggregation}), with $\sigma_\eps = 0.1$. Each cell reports the median estimate across Monte Carlo replications, with the empirical 90\% interval in brackets. The true scale parameter is $\phi=11.5$; an asterisk indicates that the interval excludes the true value.}
\label{tab:mcBetaAlphaTotalEstimates}
\begin{tabular}{ccc cc}
\hline\hline
$\mu$ & $\phi$ & $K/P$ & $\hat\alpha$ & $\hat\beta$ \\
\hline
0.70 & 11.50 & 1.000 & \shortstack[c]{11.40$^{*}$\\{\scriptsize [11.39, 11.40]}} & \shortstack[c]{0.6931$^{*}$\\{\scriptsize [0.6927, 0.6934]}} \\
0.70 & 11.50 & 2.000 & \shortstack[c]{12.11$^{*}$\\{\scriptsize [12.10, 12.12]}} & \shortstack[c]{0.7347$^{*}$\\{\scriptsize [0.7341, 0.7354]}} \\
0.70 & 11.50 & 5.000 & \shortstack[c]{12.86$^{*}$\\{\scriptsize [12.85, 12.88]}} & \shortstack[c]{0.7812$^{*}$\\{\scriptsize [0.7799, 0.7827]}} \\
0.70 & 11.50 & 10.00 & \shortstack[c]{13.41$^{*}$\\{\scriptsize [13.38, 13.43]}} & \shortstack[c]{0.8179$^{*}$\\{\scriptsize [0.8159, 0.8200]}} \\
0.70 & 11.50 & 100.0 & \shortstack[c]{14.77$^{*}$\\{\scriptsize [14.69, 14.84]}} & \shortstack[c]{0.9247$^{*}$\\{\scriptsize [0.9162, 0.9320]}} \\
0.70 & 11.50 & 1000. & \shortstack[c]{15.34$^{*}$\\{\scriptsize [15.17, 15.47]}} & \shortstack[c]{0.9814$^{*}$\\{\scriptsize [0.9594, 1.000]}} \\
1.0 & 11.50 & 1.000 & \shortstack[c]{11.40$^{*}$\\{\scriptsize [11.39, 11.40]}} & \shortstack[c]{0.9901$^{*}$\\{\scriptsize [0.9896, 0.9906]}} \\
1.0 & 11.50 & 2.000 & \shortstack[c]{11.43$^{*}$\\{\scriptsize [11.42, 11.44]}} & \shortstack[c]{0.9931$^{*}$\\{\scriptsize [0.9925, 0.9937]}} \\
1.0 & 11.50 & 5.000 & \shortstack[c]{11.46$^{*}$\\{\scriptsize [11.45, 11.46]}} & \shortstack[c]{0.9955$^{*}$\\{\scriptsize [0.9947, 0.9963]}} \\
1.0 & 11.50 & 10.00 & \shortstack[c]{11.47$^{*}$\\{\scriptsize [11.46, 11.48]}} & \shortstack[c]{0.9967$^{*}$\\{\scriptsize [0.9957, 0.9978]}} \\
1.0 & 11.50 & 100.0 & \shortstack[c]{11.49\\{\scriptsize [11.48, 11.50]}} & \shortstack[c]{0.9991\\{\scriptsize [0.9973, 1.001]}} \\
1.0 & 11.50 & 1000. & \shortstack[c]{11.49\\{\scriptsize [11.48, 11.51]}} & \shortstack[c]{0.9998\\{\scriptsize [0.9964, 1.003]}} \\
1.3 & 11.50 & 1.000 & \shortstack[c]{11.40$^{*}$\\{\scriptsize [11.39, 11.40]}} & \shortstack[c]{1.287$^{*}$\\{\scriptsize [1.286, 1.288]}} \\
1.3 & 11.50 & 2.000 & \shortstack[c]{10.54$^{*}$\\{\scriptsize [10.53, 10.55]}} & \shortstack[c]{1.189$^{*}$\\{\scriptsize [1.188, 1.190]}} \\
1.3 & 11.50 & 5.000 & \shortstack[c]{9.976$^{*}$\\{\scriptsize [9.967, 9.986]}} & \shortstack[c]{1.120$^{*}$\\{\scriptsize [1.118, 1.121]}} \\
1.3 & 11.50 & 10.00 & \shortstack[c]{9.711$^{*}$\\{\scriptsize [9.699, 9.723]}} & \shortstack[c]{1.083$^{*}$\\{\scriptsize [1.081, 1.086]}} \\
1.3 & 11.50 & 100.0 & \shortstack[c]{9.346$^{*}$\\{\scriptsize [9.328, 9.364]}} & \shortstack[c]{1.022$^{*}$\\{\scriptsize [1.017, 1.027]}} \\
1.3 & 11.50 & 1000. & \shortstack[c]{9.282$^{*}$\\{\scriptsize [9.273, 9.294]}} & \shortstack[c]{1.004$^{*}$\\{\scriptsize [0.9954, 1.015]}} \\
\hline\hline
\end{tabular}
\end{table}

%% file: tables/table_montecarlo_area_bias.tex
\begin{table}[!htbp]
\centering
\footnotesize
\setlength{\tabcolsep}{4pt}
\caption{Mean absolute bias in $\hat\beta$ with and without the area control. Bias is averaged over the six aggregation ratios for each true elasticity, variable definition, and area-size correlation regime. The last column reports the percentage reduction in mean absolute bias from adding $\log(\mathrm{Km2})$; negative values occur near the aggregation-invariant case $\mu=1$, where the baseline bias is already close to zero.}\label{tab:mcAreaBias}
\begin{adjustbox}{width=\textwidth}
\begin{tabular}{cccccc}
\toprule
True $\mu$ & Variables & Area-size corr. & $|\hat\beta-\mu|$ no area & $|\hat\beta-\mu|$ area & Bias reduction \\
\midrule
0.70 & Totals & Independent & 0.1243 & 0.1243 & 0.0\% \\
0.70 & Totals & Positive & 0.1244 & 0.02215 & 82.2\% \\
0.70 & Totals & Negative & 0.1244 & 0.09296 & 25.3\% \\
0.70 & Densities & Independent & 0.1807 & 0.1244 & 31.2\% \\
0.70 & Densities & Positive & 0.05725 & 0.02308 & 59.7\% \\
0.70 & Densities & Negative & 0.1244 & 0.09295 & 25.3\% \\
1.0 & Totals & Independent & 0.004277 & 0.004293 & -0.4\% \\
1.0 & Totals & Positive & 0.004281 & 0.007695 & -79.7\% \\
1.0 & Totals & Negative & 0.004269 & 0.005233 & -22.6\% \\
1.0 & Densities & Independent & 0.002298 & 0.004274 & -86.0\% \\
1.0 & Densities & Positive & 0.005908 & 0.007473 & -26.5\% \\
1.0 & Densities & Negative & 0.004269 & 0.005201 & -21.8\% \\
1.3 & Totals & Independent & 0.1823 & 0.1824 & -0.0\% \\
1.3 & Totals & Positive & 0.1824 & 0.1139 & 37.6\% \\
1.3 & Totals & Negative & 0.1823 & 0.1593 & 12.6\% \\
1.3 & Densities & Independent & 0.2478 & 0.1823 & 26.4\% \\
1.3 & Densities & Positive & 0.1874 & 0.1132 & 39.6\% \\
1.3 & Densities & Negative & 0.1824 & 0.1592 & 12.7\% \\
\bottomrule
\end{tabular}
\end{adjustbox}
\end{table}

%% file: tables/table_montecarlo_alpha_beta_gamma.tex
\begingroup
\scriptsize
\setlength{\tabcolsep}{4pt}
\begin{longtable}{cccclccc}
\caption{Monte Carlo estimates of $\alpha$, $\beta$, and $\gamma$ under the area-control specification. Each cell reports the median estimate across 1000 Monte Carlo replications, with the empirical 90\% Monte Carlo interval in brackets. The true DGP scale parameter is $\phi=11.50$. An asterisk on $\hat\alpha$ indicates that the interval does not include $\phi$; an asterisk on $\hat\beta$ indicates that the interval does not include the true DGP elasticity $\mu$; an asterisk on $\hat\gamma$ indicates that the interval does not include zero.}\label{tab:mcBetaGammaEstimates}\\
\toprule
$\mu$ & $\phi$ & $K/P$ & Variables & Area-size corr. & $\hat\alpha$ & $\hat\beta$ & $\hat\gamma$ \\
\midrule
\endfirsthead
\caption[]{Monte Carlo estimates of $\alpha$, $\beta$, and $\gamma$ under the area-control specification (continued).}\\
\toprule
$\mu$ & $\phi$ & $K/P$ & Variables & Area-size corr. & $\hat\alpha$ & $\hat\beta$ & $\hat\gamma$ \\
\midrule
\endhead
\midrule
\multicolumn{8}{r}{Continued on next page}\\
\endfoot
\bottomrule
\endlastfoot
0.70 & 11.50 & 1.000 & Totals & Independent & \shortstack[c]{11.40$^{*}$\\{\scriptsize [11.39, 11.40]}} & \shortstack[c]{0.6931$^{*}$\\{\scriptsize [0.6927, 0.6935]}} & \shortstack[c]{0.00001809\\{\scriptsize [-0.0005022, 0.0004994]}} \\
0.70 & 11.50 & 1.000 & Totals & Positive & \shortstack[c]{11.40$^{*}$\\{\scriptsize [11.39, 11.40]}} & \shortstack[c]{0.6931$^{*}$\\{\scriptsize [0.6927, 0.6934]}} & \shortstack[c]{0.000002079\\{\scriptsize [-0.0005151, 0.0005415]}} \\
0.70 & 11.50 & 1.000 & Totals & Negative & \shortstack[c]{22.30\\{\scriptsize [-433.6, 480.8]}} & \shortstack[c]{0.6931$^{*}$\\{\scriptsize [0.6927, 0.6934]}} & \shortstack[c]{0.9468\\{\scriptsize [-38.65, 40.77]}} \\
0.70 & 11.50 & 1.000 & Densities & Independent & \shortstack[c]{11.40$^{*}$\\{\scriptsize [11.39, 11.40]}} & \shortstack[c]{0.6931$^{*}$\\{\scriptsize [0.6927, 0.6934]}} & \shortstack[c]{-0.3069$^{*}$\\{\scriptsize [-0.3076, -0.3063]}} \\
0.70 & 11.50 & 1.000 & Densities & Positive & \shortstack[c]{11.40$^{*}$\\{\scriptsize [11.39, 11.40]}} & \shortstack[c]{0.6931$^{*}$\\{\scriptsize [0.6927, 0.6935]}} & \shortstack[c]{-0.3069$^{*}$\\{\scriptsize [-0.3076, -0.3063]}} \\
0.70 & 11.50 & 1.000 & Densities & Negative & \shortstack[c]{10.29\\{\scriptsize [-441.3, 484.3]}} & \shortstack[c]{0.6931$^{*}$\\{\scriptsize [0.6927, 0.6934]}} & \shortstack[c]{-0.4030\\{\scriptsize [-39.62, 40.77]}} \\
0.70 & 11.50 & 2.000 & Totals & Independent & \shortstack[c]{12.11$^{*}$\\{\scriptsize [12.09, 12.13]}} & \shortstack[c]{0.7348$^{*}$\\{\scriptsize [0.7341, 0.7354]}} & \shortstack[c]{0.00002899\\{\scriptsize [-0.001042, 0.001106]}} \\
0.70 & 11.50 & 2.000 & Totals & Positive & \shortstack[c]{12.59$^{*}$\\{\scriptsize [12.58, 12.60]}} & \shortstack[c]{0.7184$^{*}$\\{\scriptsize [0.7177, 0.7191]}} & \shortstack[c]{0.06289$^{*}$\\{\scriptsize [0.06202, 0.06381]}} \\
0.70 & 11.50 & 2.000 & Totals & Negative & \shortstack[c]{-27683.$^{*}$\\{\scriptsize [-28699., -26583.]}} & \shortstack[c]{0.7231$^{*}$\\{\scriptsize [0.7224, 0.7240]}} & \shortstack[c]{-2560.$^{*}$\\{\scriptsize [-2654., -2458.]}} \\
0.70 & 11.50 & 2.000 & Densities & Independent & \shortstack[c]{12.11$^{*}$\\{\scriptsize [12.09, 12.13]}} & \shortstack[c]{0.7348$^{*}$\\{\scriptsize [0.7341, 0.7355]}} & \shortstack[c]{-0.2652$^{*}$\\{\scriptsize [-0.2665, -0.2639]}} \\
0.70 & 11.50 & 2.000 & Densities & Positive & \shortstack[c]{12.59$^{*}$\\{\scriptsize [12.58, 12.60]}} & \shortstack[c]{0.7184$^{*}$\\{\scriptsize [0.7178, 0.7192]}} & \shortstack[c]{-0.2186$^{*}$\\{\scriptsize [-0.2196, -0.2177]}} \\
0.70 & 11.50 & 2.000 & Densities & Negative & \shortstack[c]{-27691.$^{*}$\\{\scriptsize [-28678., -26525.]}} & \shortstack[c]{0.7231$^{*}$\\{\scriptsize [0.7224, 0.7240]}} & \shortstack[c]{-2561.$^{*}$\\{\scriptsize [-2652., -2453.]}} \\
0.70 & 11.50 & 5.000 & Totals & Independent & \shortstack[c]{12.87$^{*}$\\{\scriptsize [12.84, 12.89]}} & \shortstack[c]{0.7813$^{*}$\\{\scriptsize [0.7799, 0.7826]}} & \shortstack[c]{0.000006056\\{\scriptsize [-0.002050, 0.002071]}} \\
0.70 & 11.50 & 5.000 & Totals & Positive & \shortstack[c]{13.34$^{*}$\\{\scriptsize [13.33, 13.36]}} & \shortstack[c]{0.7257$^{*}$\\{\scriptsize [0.7239, 0.7271]}} & \shortstack[c]{0.1160$^{*}$\\{\scriptsize [0.1141, 0.1179]}} \\
0.70 & 11.50 & 5.000 & Totals & Negative & \shortstack[c]{-21288.$^{*}$\\{\scriptsize [-21909., -20617.]}} & \shortstack[c]{0.7408$^{*}$\\{\scriptsize [0.7392, 0.7426]}} & \shortstack[c]{-2151.$^{*}$\\{\scriptsize [-2213., -2083.]}} \\
0.70 & 11.50 & 5.000 & Densities & Independent & \shortstack[c]{12.86$^{*}$\\{\scriptsize [12.84, 12.89]}} & \shortstack[c]{0.7812$^{*}$\\{\scriptsize [0.7799, 0.7826]}} & \shortstack[c]{-0.2187$^{*}$\\{\scriptsize [-0.2212, -0.2161]}} \\
0.70 & 11.50 & 5.000 & Densities & Positive & \shortstack[c]{13.34$^{*}$\\{\scriptsize [13.33, 13.36]}} & \shortstack[c]{0.7257$^{*}$\\{\scriptsize [0.7239, 0.7273]}} & \shortstack[c]{-0.1584$^{*}$\\{\scriptsize [-0.1598, -0.1569]}} \\
0.70 & 11.50 & 5.000 & Densities & Negative & \shortstack[c]{-21292.$^{*}$\\{\scriptsize [-21894., -20653.]}} & \shortstack[c]{0.7409$^{*}$\\{\scriptsize [0.7391, 0.7426]}} & \shortstack[c]{-2151.$^{*}$\\{\scriptsize [-2212., -2087.]}} \\
0.70 & 11.50 & 10.00 & Totals & Independent & \shortstack[c]{13.41$^{*}$\\{\scriptsize [13.37, 13.45]}} & \shortstack[c]{0.8180$^{*}$\\{\scriptsize [0.8158, 0.8200]}} & \shortstack[c]{-0.00004643\\{\scriptsize [-0.003142, 0.002893]}} \\
0.70 & 11.50 & 10.00 & Totals & Positive & \shortstack[c]{13.74$^{*}$\\{\scriptsize [13.72, 13.77]}} & \shortstack[c]{0.7246$^{*}$\\{\scriptsize [0.7215, 0.7277]}} & \shortstack[c]{0.1535$^{*}$\\{\scriptsize [0.1499, 0.1567]}} \\
0.70 & 11.50 & 10.00 & Totals & Negative & \shortstack[c]{-12380.$^{*}$\\{\scriptsize [-12795., -11975.]}} & \shortstack[c]{0.7573$^{*}$\\{\scriptsize [0.7547, 0.7601]}} & \shortstack[c]{-1346.$^{*}$\\{\scriptsize [-1391., -1301.]}} \\
0.70 & 11.50 & 10.00 & Densities & Independent & \shortstack[c]{13.41$^{*}$\\{\scriptsize [13.37, 13.45]}} & \shortstack[c]{0.8179$^{*}$\\{\scriptsize [0.8159, 0.8201]}} & \shortstack[c]{-0.1821$^{*}$\\{\scriptsize [-0.1858, -0.1781]}} \\
0.70 & 11.50 & 10.00 & Densities & Positive & \shortstack[c]{13.74$^{*}$\\{\scriptsize [13.72, 13.77]}} & \shortstack[c]{0.7246$^{*}$\\{\scriptsize [0.7217, 0.7276]}} & \shortstack[c]{-0.1221$^{*}$\\{\scriptsize [-0.1242, -0.1199]}} \\
0.70 & 11.50 & 10.00 & Densities & Negative & \shortstack[c]{-12378.$^{*}$\\{\scriptsize [-12783., -11979.]}} & \shortstack[c]{0.7573$^{*}$\\{\scriptsize [0.7545, 0.7599]}} & \shortstack[c]{-1346.$^{*}$\\{\scriptsize [-1390., -1302.]}} \\
0.70 & 11.50 & 100.0 & Totals & Independent & \shortstack[c]{14.77$^{*}$\\{\scriptsize [14.68, 14.86]}} & \shortstack[c]{0.9251$^{*}$\\{\scriptsize [0.9174, 0.9323]}} & \shortstack[c]{0.00009547\\{\scriptsize [-0.006945, 0.007536]}} \\
0.70 & 11.50 & 100.0 & Totals & Positive & \shortstack[c]{14.44$^{*}$\\{\scriptsize [14.38, 14.50]}} & \shortstack[c]{0.6954\\{\scriptsize [0.6759, 0.7165]}} & \shortstack[c]{0.2614$^{*}$\\{\scriptsize [0.2392, 0.2841]}} \\
0.70 & 11.50 & 100.0 & Totals & Negative & \shortstack[c]{-674.8$^{*}$\\{\scriptsize [-775.9, -586.0]}} & \shortstack[c]{0.8684$^{*}$\\{\scriptsize [0.8577, 0.8783]}} & \shortstack[c]{-99.73$^{*}$\\{\scriptsize [-114.4, -86.89]}} \\
0.70 & 11.50 & 100.0 & Densities & Independent & \shortstack[c]{14.77$^{*}$\\{\scriptsize [14.67, 14.86]}} & \shortstack[c]{0.9248$^{*}$\\{\scriptsize [0.9171, 0.9326]}} & \shortstack[c]{-0.07531$^{*}$\\{\scriptsize [-0.08673, -0.06474]}} \\
0.70 & 11.50 & 100.0 & Densities & Positive & \shortstack[c]{14.44$^{*}$\\{\scriptsize [14.38, 14.50]}} & \shortstack[c]{0.6953\\{\scriptsize [0.6761, 0.7168]}} & \shortstack[c]{-0.04244$^{*}$\\{\scriptsize [-0.04907, -0.03653]}} \\
0.70 & 11.50 & 100.0 & Densities & Negative & \shortstack[c]{-678.9$^{*}$\\{\scriptsize [-773.4, -577.5]}} & \shortstack[c]{0.8680$^{*}$\\{\scriptsize [0.8569, 0.8794]}} & \shortstack[c]{-100.5$^{*}$\\{\scriptsize [-114.2, -85.79]}} \\
0.70 & 11.50 & 1000. & Totals & Independent & \shortstack[c]{15.33$^{*}$\\{\scriptsize [15.16, 15.48]}} & \shortstack[c]{0.9804$^{*}$\\{\scriptsize [0.9602, 0.9991]}} & \shortstack[c]{-0.0001391\\{\scriptsize [-0.01284, 0.01313]}} \\
0.70 & 11.50 & 1000. & Totals & Positive & \shortstack[c]{14.56$^{*}$\\{\scriptsize [14.30, 14.79]}} & \shortstack[c]{0.6473\\{\scriptsize [0.5445, 0.7524]}} & \shortstack[c]{0.3431$^{*}$\\{\scriptsize [0.2337, 0.4457]}} \\
0.70 & 11.50 & 1000. & Totals & Negative & \shortstack[c]{1.277\\{\scriptsize [-14.49, 13.41]}} & \shortstack[c]{0.9611$^{*}$\\{\scriptsize [0.9263, 0.9954]}} & \shortstack[c]{-3.013$^{*}$\\{\scriptsize [-6.391, -0.4432]}} \\
0.70 & 11.50 & 1000. & Densities & Independent & \shortstack[c]{15.34$^{*}$\\{\scriptsize [15.17, 15.48]}} & \shortstack[c]{0.9811$^{*}$\\{\scriptsize [0.9606, 0.9987]}} & \shortstack[c]{-0.01874\\{\scriptsize [-0.04307, 0.003105]}} \\
0.70 & 11.50 & 1000. & Densities & Positive & \shortstack[c]{14.54$^{*}$\\{\scriptsize [14.31, 14.78]}} & \shortstack[c]{0.6388\\{\scriptsize [0.5452, 0.7427]}} & \shortstack[c]{-0.009058\\{\scriptsize [-0.02191, 0.001504]}} \\
0.70 & 11.50 & 1000. & Densities & Negative & \shortstack[c]{1.311\\{\scriptsize [-13.49, 12.54]}} & \shortstack[c]{0.9615$^{*}$\\{\scriptsize [0.9282, 0.9929]}} & \shortstack[c]{-3.053$^{*}$\\{\scriptsize [-6.250, -0.6335]}} \\
1.0 & 11.50 & 1.000 & Totals & Independent & \shortstack[c]{11.40$^{*}$\\{\scriptsize [11.39, 11.40]}} & \shortstack[c]{0.9901$^{*}$\\{\scriptsize [0.9896, 0.9906]}} & \shortstack[c]{0.00001014\\{\scriptsize [-0.0005342, 0.0005023]}} \\
1.0 & 11.50 & 1.000 & Totals & Positive & \shortstack[c]{11.40$^{*}$\\{\scriptsize [11.39, 11.40]}} & \shortstack[c]{0.9901$^{*}$\\{\scriptsize [0.9896, 0.9906]}} & \shortstack[c]{0.000007809\\{\scriptsize [-0.0005353, 0.0005260]}} \\
1.0 & 11.50 & 1.000 & Totals & Negative & \shortstack[c]{26.84\\{\scriptsize [-445.5, 467.8]}} & \shortstack[c]{0.9901$^{*}$\\{\scriptsize [0.9896, 0.9906]}} & \shortstack[c]{1.342\\{\scriptsize [-39.68, 39.64]}} \\
1.0 & 11.50 & 1.000 & Densities & Independent & \shortstack[c]{11.40$^{*}$\\{\scriptsize [11.39, 11.40]}} & \shortstack[c]{0.9901$^{*}$\\{\scriptsize [0.9896, 0.9906]}} & \shortstack[c]{-0.009909$^{*}$\\{\scriptsize [-0.01062, -0.009156]}} \\
1.0 & 11.50 & 1.000 & Densities & Positive & \shortstack[c]{11.40$^{*}$\\{\scriptsize [11.39, 11.40]}} & \shortstack[c]{0.9901$^{*}$\\{\scriptsize [0.9896, 0.9906]}} & \shortstack[c]{-0.009909$^{*}$\\{\scriptsize [-0.01069, -0.009173]}} \\
1.0 & 11.50 & 1.000 & Densities & Negative & \shortstack[c]{8.888\\{\scriptsize [-439.5, 470.2]}} & \shortstack[c]{0.9901$^{*}$\\{\scriptsize [0.9896, 0.9906]}} & \shortstack[c]{-0.2278\\{\scriptsize [-39.17, 39.84]}} \\
1.0 & 11.50 & 2.000 & Totals & Independent & \shortstack[c]{11.43$^{*}$\\{\scriptsize [11.42, 11.44]}} & \shortstack[c]{0.9931$^{*}$\\{\scriptsize [0.9925, 0.9936]}} & \shortstack[c]{-0.000006201\\{\scriptsize [-0.0006385, 0.0006128]}} \\
1.0 & 11.50 & 2.000 & Totals & Positive & \shortstack[c]{11.45$^{*}$\\{\scriptsize [11.44, 11.46]}} & \shortstack[c]{0.9921$^{*}$\\{\scriptsize [0.9915, 0.9928]}} & \shortstack[c]{0.002465$^{*}$\\{\scriptsize [0.001798, 0.003119]}} \\
1.0 & 11.50 & 2.000 & Totals & Negative & \shortstack[c]{-1043.$^{*}$\\{\scriptsize [-1339., -760.2]}} & \shortstack[c]{0.9924$^{*}$\\{\scriptsize [0.9918, 0.9930]}} & \shortstack[c]{-97.48$^{*}$\\{\scriptsize [-124.8, -71.32]}} \\
1.0 & 11.50 & 2.000 & Densities & Independent & \shortstack[c]{11.43$^{*}$\\{\scriptsize [11.42, 11.44]}} & \shortstack[c]{0.9931$^{*}$\\{\scriptsize [0.9925, 0.9937]}} & \shortstack[c]{-0.006914$^{*}$\\{\scriptsize [-0.007741, -0.006031]}} \\
1.0 & 11.50 & 2.000 & Densities & Positive & \shortstack[c]{11.45$^{*}$\\{\scriptsize [11.44, 11.46]}} & \shortstack[c]{0.9921$^{*}$\\{\scriptsize [0.9914, 0.9927]}} & \shortstack[c]{-0.005453$^{*}$\\{\scriptsize [-0.006095, -0.004725]}} \\
1.0 & 11.50 & 2.000 & Densities & Negative & \shortstack[c]{-1033.$^{*}$\\{\scriptsize [-1328., -740.8]}} & \shortstack[c]{0.9924$^{*}$\\{\scriptsize [0.9918, 0.9930]}} & \shortstack[c]{-96.53$^{*}$\\{\scriptsize [-123.8, -69.53]}} \\
1.0 & 11.50 & 5.000 & Totals & Independent & \shortstack[c]{11.46$^{*}$\\{\scriptsize [11.45, 11.47]}} & \shortstack[c]{0.9954$^{*}$\\{\scriptsize [0.9947, 0.9962]}} & \shortstack[c]{0.00002299\\{\scriptsize [-0.0007915, 0.0008096]}} \\
1.0 & 11.50 & 5.000 & Totals & Positive & \shortstack[c]{11.47$^{*}$\\{\scriptsize [11.47, 11.48]}} & \shortstack[c]{0.9928$^{*}$\\{\scriptsize [0.9917, 0.9940]}} & \shortstack[c]{0.003923$^{*}$\\{\scriptsize [0.002747, 0.005075]}} \\
1.0 & 11.50 & 5.000 & Totals & Negative & \shortstack[c]{-643.6$^{*}$\\{\scriptsize [-820.3, -487.5]}} & \shortstack[c]{0.9937$^{*}$\\{\scriptsize [0.9927, 0.9947]}} & \shortstack[c]{-66.15$^{*}$\\{\scriptsize [-83.98, -50.38]}} \\
1.0 & 11.50 & 5.000 & Densities & Independent & \shortstack[c]{11.46$^{*}$\\{\scriptsize [11.45, 11.47]}} & \shortstack[c]{0.9955$^{*}$\\{\scriptsize [0.9947, 0.9962]}} & \shortstack[c]{-0.004536$^{*}$\\{\scriptsize [-0.005722, -0.003350]}} \\
1.0 & 11.50 & 5.000 & Densities & Positive & \shortstack[c]{11.47$^{*}$\\{\scriptsize [11.47, 11.48]}} & \shortstack[c]{0.9929$^{*}$\\{\scriptsize [0.9917, 0.9940]}} & \shortstack[c]{-0.003279$^{*}$\\{\scriptsize [-0.004121, -0.002403]}} \\
1.0 & 11.50 & 5.000 & Densities & Negative & \shortstack[c]{-647.2$^{*}$\\{\scriptsize [-813.9, -488.4]}} & \shortstack[c]{0.9937$^{*}$\\{\scriptsize [0.9927, 0.9946]}} & \shortstack[c]{-66.51$^{*}$\\{\scriptsize [-83.35, -50.47]}} \\
1.0 & 11.50 & 10.00 & Totals & Independent & \shortstack[c]{11.47$^{*}$\\{\scriptsize [11.46, 11.48]}} & \shortstack[c]{0.9967$^{*}$\\{\scriptsize [0.9958, 0.9978]}} & \shortstack[c]{0.000004499\\{\scriptsize [-0.0009374, 0.0009454]}} \\
1.0 & 11.50 & 10.00 & Totals & Positive & \shortstack[c]{11.48$^{*}$\\{\scriptsize [11.48, 11.49]}} & \shortstack[c]{0.9931$^{*}$\\{\scriptsize [0.9913, 0.9949]}} & \shortstack[c]{0.004623$^{*}$\\{\scriptsize [0.002901, 0.006368]}} \\
1.0 & 11.50 & 10.00 & Totals & Negative & \shortstack[c]{-315.3$^{*}$\\{\scriptsize [-417.1, -217.7]}} & \shortstack[c]{0.9947$^{*}$\\{\scriptsize [0.9932, 0.9961]}} & \shortstack[c]{-35.48$^{*}$\\{\scriptsize [-46.52, -24.88]}} \\
1.0 & 11.50 & 10.00 & Densities & Independent & \shortstack[c]{11.47$^{*}$\\{\scriptsize [11.46, 11.48]}} & \shortstack[c]{0.9968$^{*}$\\{\scriptsize [0.9957, 0.9978]}} & \shortstack[c]{-0.003258$^{*}$\\{\scriptsize [-0.004646, -0.001885]}} \\
1.0 & 11.50 & 10.00 & Densities & Positive & \shortstack[c]{11.49$^{*}$\\{\scriptsize [11.48, 11.49]}} & \shortstack[c]{0.9932$^{*}$\\{\scriptsize [0.9915, 0.9949]}} & \shortstack[c]{-0.002259$^{*}$\\{\scriptsize [-0.003363, -0.001228]}} \\
1.0 & 11.50 & 10.00 & Densities & Negative & \shortstack[c]{-315.6$^{*}$\\{\scriptsize [-421.9, -201.1]}} & \shortstack[c]{0.9947$^{*}$\\{\scriptsize [0.9933, 0.9961]}} & \shortstack[c]{-35.52$^{*}$\\{\scriptsize [-47.06, -23.08]}} \\
1.0 & 11.50 & 100.0 & Totals & Independent & \shortstack[c]{11.49\\{\scriptsize [11.48, 11.51]}} & \shortstack[c]{0.9991\\{\scriptsize [0.9970, 1.001]}} & \shortstack[c]{0.00001485\\{\scriptsize [-0.001329, 0.001488]}} \\
1.0 & 11.50 & 100.0 & Totals & Positive & \shortstack[c]{11.50\\{\scriptsize [11.48, 11.52]}} & \shortstack[c]{0.9932\\{\scriptsize [0.9842, 1.001]}} & \shortstack[c]{0.006255\\{\scriptsize [-0.002253, 0.01529]}} \\
1.0 & 11.50 & 100.0 & Totals & Negative & \shortstack[c]{1.193\\{\scriptsize [-17.59, 17.82]}} & \shortstack[c]{0.9982\\{\scriptsize [0.9949, 1.001]}} & \shortstack[c]{-1.490\\{\scriptsize [-4.207, 0.9141]}} \\
1.0 & 11.50 & 100.0 & Densities & Independent & \shortstack[c]{11.49\\{\scriptsize [11.47, 11.51]}} & \shortstack[c]{0.9992\\{\scriptsize [0.9971, 1.001]}} & \shortstack[c]{-0.0007756\\{\scriptsize [-0.003350, 0.001620]}} \\
1.0 & 11.50 & 100.0 & Densities & Positive & \shortstack[c]{11.50\\{\scriptsize [11.48, 11.52]}} & \shortstack[c]{0.9930\\{\scriptsize [0.9855, 1.001]}} & \shortstack[c]{-0.0005243\\{\scriptsize [-0.002494, 0.001310]}} \\
1.0 & 11.50 & 100.0 & Densities & Negative & \shortstack[c]{1.199\\{\scriptsize [-17.88, 20.19]}} & \shortstack[c]{0.9982\\{\scriptsize [0.9947, 1.002]}} & \shortstack[c]{-1.491\\{\scriptsize [-4.253, 1.260]}} \\
1.0 & 11.50 & 1000. & Totals & Independent & \shortstack[c]{11.49\\{\scriptsize [11.48, 11.51]}} & \shortstack[c]{0.9998\\{\scriptsize [0.9963, 1.003]}} & \shortstack[c]{0.00003492\\{\scriptsize [-0.001935, 0.001837]}} \\
1.0 & 11.50 & 1000. & Totals & Positive & \shortstack[c]{11.51\\{\scriptsize [11.45, 11.56]}} & \shortstack[c]{0.9920\\{\scriptsize [0.9587, 1.027]}} & \shortstack[c]{0.007592\\{\scriptsize [-0.02689, 0.04235]}} \\
1.0 & 11.50 & 1000. & Totals & Negative & \shortstack[c]{11.36\\{\scriptsize [8.986, 13.85]}} & \shortstack[c]{0.9995\\{\scriptsize [0.9934, 1.006]}} & \shortstack[c]{-0.02949\\{\scriptsize [-0.5400, 0.5055]}} \\
1.0 & 11.50 & 1000. & Densities & Independent & \shortstack[c]{11.49\\{\scriptsize [11.48, 11.51]}} & \shortstack[c]{0.9998\\{\scriptsize [0.9963, 1.003]}} & \shortstack[c]{-0.0003404\\{\scriptsize [-0.004143, 0.003604]}} \\
1.0 & 11.50 & 1000. & Densities & Positive & \shortstack[c]{11.50\\{\scriptsize [11.45, 11.57]}} & \shortstack[c]{0.9947\\{\scriptsize [0.9562, 1.029]}} & \shortstack[c]{-0.0001700\\{\scriptsize [-0.003175, 0.002904]}} \\
1.0 & 11.50 & 1000. & Densities & Negative & \shortstack[c]{11.36\\{\scriptsize [8.932, 13.80]}} & \shortstack[c]{0.9997\\{\scriptsize [0.9935, 1.006]}} & \shortstack[c]{-0.02905\\{\scriptsize [-0.5583, 0.5030]}} \\
1.3 & 11.50 & 1.000 & Totals & Independent & \shortstack[c]{11.40$^{*}$\\{\scriptsize [11.39, 11.40]}} & \shortstack[c]{1.287$^{*}$\\{\scriptsize [1.286, 1.288]}} & \shortstack[c]{0.000006210\\{\scriptsize [-0.0005112, 0.0004921]}} \\
1.3 & 11.50 & 1.000 & Totals & Positive & \shortstack[c]{11.40$^{*}$\\{\scriptsize [11.39, 11.40]}} & \shortstack[c]{1.287$^{*}$\\{\scriptsize [1.286, 1.288]}} & \shortstack[c]{0.000002563\\{\scriptsize [-0.0005351, 0.0004818]}} \\
1.3 & 11.50 & 1.000 & Totals & Negative & \shortstack[c]{-9.551\\{\scriptsize [-451.6, 432.3]}} & \shortstack[c]{1.287$^{*}$\\{\scriptsize [1.286, 1.288]}} & \shortstack[c]{-1.819\\{\scriptsize [-40.21, 36.56]}} \\
1.3 & 11.50 & 1.000 & Densities & Independent & \shortstack[c]{11.40$^{*}$\\{\scriptsize [11.39, 11.40]}} & \shortstack[c]{1.287$^{*}$\\{\scriptsize [1.286, 1.288]}} & \shortstack[c]{0.2871$^{*}$\\{\scriptsize [0.2863, 0.2880]}} \\
1.3 & 11.50 & 1.000 & Densities & Positive & \shortstack[c]{11.40$^{*}$\\{\scriptsize [11.39, 11.40]}} & \shortstack[c]{1.287$^{*}$\\{\scriptsize [1.286, 1.288]}} & \shortstack[c]{0.2872$^{*}$\\{\scriptsize [0.2862, 0.2880]}} \\
1.3 & 11.50 & 1.000 & Densities & Negative & \shortstack[c]{8.538\\{\scriptsize [-465.1, 464.1]}} & \shortstack[c]{1.287$^{*}$\\{\scriptsize [1.286, 1.288]}} & \shortstack[c]{0.03885\\{\scriptsize [-41.10, 39.61]}} \\
1.3 & 11.50 & 2.000 & Totals & Independent & \shortstack[c]{10.54$^{*}$\\{\scriptsize [10.53, 10.55]}} & \shortstack[c]{1.189$^{*}$\\{\scriptsize [1.188, 1.190]}} & \shortstack[c]{-0.00004087\\{\scriptsize [-0.001006, 0.0009554]}} \\
1.3 & 11.50 & 2.000 & Totals & Positive & \shortstack[c]{10.17$^{*}$\\{\scriptsize [10.16, 10.18]}} & \shortstack[c]{1.216$^{*}$\\{\scriptsize [1.215, 1.217]}} & \shortstack[c]{-0.05052$^{*}$\\{\scriptsize [-0.05149, -0.04961]}} \\
1.3 & 11.50 & 2.000 & Totals & Negative & \shortstack[c]{21768.$^{*}$\\{\scriptsize [20830., 22614.]}} & \shortstack[c]{1.208$^{*}$\\{\scriptsize [1.207, 1.209]}} & \shortstack[c]{2011.$^{*}$\\{\scriptsize [1924., 2089.]}} \\
1.3 & 11.50 & 2.000 & Densities & Independent & \shortstack[c]{10.54$^{*}$\\{\scriptsize [10.53, 10.55]}} & \shortstack[c]{1.189$^{*}$\\{\scriptsize [1.188, 1.190]}} & \shortstack[c]{0.1892$^{*}$\\{\scriptsize [0.1878, 0.1906]}} \\
1.3 & 11.50 & 2.000 & Densities & Positive & \shortstack[c]{10.17$^{*}$\\{\scriptsize [10.16, 10.18]}} & \shortstack[c]{1.216$^{*}$\\{\scriptsize [1.215, 1.217]}} & \shortstack[c]{0.1658$^{*}$\\{\scriptsize [0.1647, 0.1668]}} \\
1.3 & 11.50 & 2.000 & Densities & Negative & \shortstack[c]{21757.$^{*}$\\{\scriptsize [20707., 22614.]}} & \shortstack[c]{1.208$^{*}$\\{\scriptsize [1.207, 1.209]}} & \shortstack[c]{2010.$^{*}$\\{\scriptsize [1913., 2089.]}} \\
1.3 & 11.50 & 5.000 & Totals & Independent & \shortstack[c]{9.976$^{*}$\\{\scriptsize [9.958, 9.995]}} & \shortstack[c]{1.120$^{*}$\\{\scriptsize [1.118, 1.121]}} & \shortstack[c]{-0.00001340\\{\scriptsize [-0.001608, 0.001581]}} \\
1.3 & 11.50 & 5.000 & Totals & Positive & \shortstack[c]{9.582$^{*}$\\{\scriptsize [9.570, 9.594]}} & \shortstack[c]{1.182$^{*}$\\{\scriptsize [1.180, 1.185]}} & \shortstack[c]{-0.07737$^{*}$\\{\scriptsize [-0.07930, -0.07549]}} \\
1.3 & 11.50 & 5.000 & Totals & Negative & \shortstack[c]{13680.$^{*}$\\{\scriptsize [13189., 14168.]}} & \shortstack[c]{1.164$^{*}$\\{\scriptsize [1.162, 1.166]}} & \shortstack[c]{1380.$^{*}$\\{\scriptsize [1331., 1430.]}} \\
1.3 & 11.50 & 5.000 & Densities & Independent & \shortstack[c]{9.976$^{*}$\\{\scriptsize [9.956, 9.995]}} & \shortstack[c]{1.120$^{*}$\\{\scriptsize [1.118, 1.121]}} & \shortstack[c]{0.1198$^{*}$\\{\scriptsize [0.1175, 0.1219]}} \\
1.3 & 11.50 & 5.000 & Densities & Positive & \shortstack[c]{9.582$^{*}$\\{\scriptsize [9.569, 9.595]}} & \shortstack[c]{1.182$^{*}$\\{\scriptsize [1.180, 1.185]}} & \shortstack[c]{0.1049$^{*}$\\{\scriptsize [0.1036, 0.1063]}} \\
1.3 & 11.50 & 5.000 & Densities & Negative & \shortstack[c]{13695.$^{*}$\\{\scriptsize [13200., 14149.]}} & \shortstack[c]{1.164$^{*}$\\{\scriptsize [1.162, 1.166]}} & \shortstack[c]{1382.$^{*}$\\{\scriptsize [1332., 1428.]}} \\
1.3 & 11.50 & 10.00 & Totals & Independent & \shortstack[c]{9.711$^{*}$\\{\scriptsize [9.688, 9.733]}} & \shortstack[c]{1.083$^{*}$\\{\scriptsize [1.081, 1.085]}} & \shortstack[c]{-0.00008394\\{\scriptsize [-0.002035, 0.001984]}} \\
1.3 & 11.50 & 10.00 & Totals & Positive & \shortstack[c]{9.314$^{*}$\\{\scriptsize [9.297, 9.332]}} & \shortstack[c]{1.165$^{*}$\\{\scriptsize [1.162, 1.169]}} & \shortstack[c]{-0.09083$^{*}$\\{\scriptsize [-0.09438, -0.08726]}} \\
1.3 & 11.50 & 10.00 & Totals & Negative & \shortstack[c]{6843.$^{*}$\\{\scriptsize [6532., 7139.]}} & \shortstack[c]{1.134$^{*}$\\{\scriptsize [1.131, 1.137]}} & \shortstack[c]{741.9$^{*}$\\{\scriptsize [708.1, 774.0]}} \\
1.3 & 11.50 & 10.00 & Densities & Independent & \shortstack[c]{9.711$^{*}$\\{\scriptsize [9.689, 9.731]}} & \shortstack[c]{1.083$^{*}$\\{\scriptsize [1.081, 1.085]}} & \shortstack[c]{0.08340$^{*}$\\{\scriptsize [0.08048, 0.08600]}} \\
1.3 & 11.50 & 10.00 & Densities & Positive & \shortstack[c]{9.315$^{*}$\\{\scriptsize [9.296, 9.332]}} & \shortstack[c]{1.165$^{*}$\\{\scriptsize [1.162, 1.169]}} & \shortstack[c]{0.07453$^{*}$\\{\scriptsize [0.07280, 0.07645]}} \\
1.3 & 11.50 & 10.00 & Densities & Negative & \shortstack[c]{6836.$^{*}$\\{\scriptsize [6536., 7132.]}} & \shortstack[c]{1.134$^{*}$\\{\scriptsize [1.131, 1.137]}} & \shortstack[c]{741.2$^{*}$\\{\scriptsize [708.7, 773.3]}} \\
1.3 & 11.50 & 100.0 & Totals & Independent & \shortstack[c]{9.347$^{*}$\\{\scriptsize [9.316, 9.378]}} & \shortstack[c]{1.022$^{*}$\\{\scriptsize [1.017, 1.027]}} & \shortstack[c]{0.00007377\\{\scriptsize [-0.003500, 0.003707]}} \\
1.3 & 11.50 & 100.0 & Totals & Positive & \shortstack[c]{8.916$^{*}$\\{\scriptsize [8.828, 8.998]}} & \shortstack[c]{1.135$^{*}$\\{\scriptsize [1.114, 1.157]}} & \shortstack[c]{-0.1140$^{*}$\\{\scriptsize [-0.1369, -0.09295]}} \\
1.3 & 11.50 & 100.0 & Totals & Negative & \shortstack[c]{221.7$^{*}$\\{\scriptsize [173.1, 279.5]}} & \shortstack[c]{1.042$^{*}$\\{\scriptsize [1.033, 1.051]}} & \shortstack[c]{30.74$^{*}$\\{\scriptsize [23.69, 39.09]}} \\
1.3 & 11.50 & 100.0 & Densities & Independent & \shortstack[c]{9.347$^{*}$\\{\scriptsize [9.316, 9.380]}} & \shortstack[c]{1.022$^{*}$\\{\scriptsize [1.017, 1.027]}} & \shortstack[c]{0.02212$^{*}$\\{\scriptsize [0.01600, 0.02850]}} \\
1.3 & 11.50 & 100.0 & Densities & Positive & \shortstack[c]{8.912$^{*}$\\{\scriptsize [8.829, 8.993]}} & \shortstack[c]{1.136$^{*}$\\{\scriptsize [1.117, 1.157]}} & \shortstack[c]{0.02140$^{*}$\\{\scriptsize [0.01668, 0.02577]}} \\
1.3 & 11.50 & 100.0 & Densities & Negative & \shortstack[c]{224.3$^{*}$\\{\scriptsize [176.3, 278.4]}} & \shortstack[c]{1.042$^{*}$\\{\scriptsize [1.034, 1.051]}} & \shortstack[c]{31.14$^{*}$\\{\scriptsize [24.19, 38.99]}} \\
1.3 & 11.50 & 1000. & Totals & Independent & \shortstack[c]{9.283$^{*}$\\{\scriptsize [9.256, 9.310]}} & \shortstack[c]{1.004$^{*}$\\{\scriptsize [0.9954, 1.014]}} & \shortstack[c]{0.0002767\\{\scriptsize [-0.004700, 0.004941]}} \\
1.3 & 11.50 & 1000. & Totals & Positive & \shortstack[c]{8.809$^{*}$\\{\scriptsize [8.430, 9.193]}} & \shortstack[c]{1.131$^{*}$\\{\scriptsize [1.028, 1.231]}} & \shortstack[c]{-0.1268$^{*}$\\{\scriptsize [-0.2268, -0.02353]}} \\
1.3 & 11.50 & 1000. & Totals & Negative & \shortstack[c]{12.23\\{\scriptsize [6.403, 19.22]}} & \shortstack[c]{1.009$^{*}$\\{\scriptsize [0.9924, 1.026]}} & \shortstack[c]{0.6383\\{\scriptsize [-0.6224, 2.153]}} \\
1.3 & 11.50 & 1000. & Densities & Independent & \shortstack[c]{9.281$^{*}$\\{\scriptsize [9.254, 9.309]}} & \shortstack[c]{1.004$^{*}$\\{\scriptsize [0.9947, 1.014]}} & \shortstack[c]{0.004083\\{\scriptsize [-0.005837, 0.01479]}} \\
1.3 & 11.50 & 1000. & Densities & Positive & \shortstack[c]{8.800$^{*}$\\{\scriptsize [8.411, 9.198]}} & \shortstack[c]{1.132$^{*}$\\{\scriptsize [1.030, 1.237]}} & \shortstack[c]{0.004137\\{\scriptsize [-0.002992, 0.01178]}} \\
1.3 & 11.50 & 1000. & Densities & Negative & \shortstack[c]{12.12\\{\scriptsize [6.161, 19.60]}} & \shortstack[c]{1.008$^{*}$\\{\scriptsize [0.9917, 1.026]}} & \shortstack[c]{0.6232\\{\scriptsize [-0.6841, 2.260]}} \\
\end{longtable}
\endgroup

%% file: tables/table_montecarlo_noise.tex
\begin{table}[!htbp]\centering
\small
\caption{Monte Carlo: reverse-regression attenuation and spatial aggregation as the idiosyncratic-shock standard deviation $\sigma_\eps$ varies. The \emph{naive} specification regresses log totals on log totals with no area control; the \emph{preferred} specification uses densities with an area control. At the finest scale ($K/P=1$) both recover $\hat\beta\approx\lambda\mu$, where $\lambda=\sigma_\nu^2/(\sigma_\nu^2+\sigma_\eps^2)$ is the reliability ratio (attenuation). Under aggregation the naive slope is contracted toward one (firm-count channel), whereas the preferred slope stays near $\lambda\mu$ because the density transformation removes that channel. Averages over Monte Carlo replications; $\sigma_\nu^2=\mathrm{Var}(\log y)=1$.}
\label{tab:mcNoise}
\resizebox{\textwidth}{!}{%
\begin{tabular}{cccc cc cc}
\hline\hline
 & & & & \multicolumn{2}{c}{Naive (totals)} & \multicolumn{2}{c}{Preferred (density+area)}\\
\cline{5-6}\cline{7-8}
$\mu$ & $\sigma_\eps$ & $\lambda$ & $\lambda\mu$ & $\hat\beta$ (fine) & $\hat\beta$ (coarse) & $\hat\beta$ (fine) & $\hat\beta$ (coarse) \\
\hline
0.7 & 0.1 & 0.990 & 0.693 & 0.693 & 0.980 & 0.693 & 0.633 \\
0.7 & 0.3 & 0.917 & 0.642 & 0.642 & 0.976 & 0.642 & 0.560 \\
0.7 & 0.5 & 0.800 & 0.560 & 0.560 & 0.965 & 0.560 & 0.485 \\
0.7 & 1.0 & 0.500 & 0.350 & 0.350 & 0.909 & 0.350 & 0.191 \\
1.0 & 0.1 & 0.990 & 0.990 & 0.990 & 0.999 & 0.990 & 0.993 \\
1.0 & 0.3 & 0.917 & 0.917 & 0.918 & 0.998 & 0.918 & 0.937 \\
1.0 & 0.5 & 0.800 & 0.800 & 0.800 & 0.996 & 0.800 & 0.819 \\
1.0 & 1.0 & 0.500 & 0.500 & 0.500 & 0.980 & 0.500 & 0.502 \\
1.3 & 0.1 & 0.990 & 1.287 & 1.287 & 1.004 & 1.287 & 1.139 \\
1.3 & 0.3 & 0.917 & 1.193 & 1.193 & 1.003 & 1.193 & 1.095 \\
1.3 & 0.5 & 0.800 & 1.040 & 1.040 & 1.000 & 1.040 & 1.020 \\
1.3 & 1.0 & 0.500 & 0.650 & 0.650 & 0.995 & 0.650 & 0.754 \\
\hline\hline
\multicolumn{8}{l}{\textit{Note}: ``fine'' is $K/P=1$; ``coarse'' is the largest $K/P$ in the grid. Area is positively correlated with the firm count.}\\
\end{tabular}%
}
\end{table}

%% file: tables/table_montecarlo_biasApprox.tex
\begin{table}[!htbp]\centering
\small
\caption{Monte Carlo verification of the bias approximation in Equation~\eqref{eq:mcBiasApprox}. For each true elasticity $\mu$, shock standard deviation $\sigma_\eps$, and aggregation ratio $K/P$, the table reports the simulated bias $\hat\beta_P-\mu$ from the regression of $\log Y_p$ on $\log NTL_p$ (average over 1000 replications, $K=100{,}000$), the approximation $-\mu\kappa_P+(1-\mu)\left[\delta_P-\tfrac{\mu}{2}\rho_P\right]$, with $\kappa_P$ computed from Equation~\eqref{eq:mcReliability} using the true shocks $\eps_i$ and the true $\mu$, and $\delta_P$ and $\rho_P$ from Equation~\eqref{eq:mcDeltaRhoDef}, and their difference (the remainder $R$ of Theorem~\ref{teo:aggregation}).}
\label{tab:mcBiasApproxVerification}
\resizebox{\textwidth}{!}{%
\begin{tabular}{cc rrr rrr rrr}
\hline\hline
 & & \multicolumn{3}{c}{$\mu=0.7$} & \multicolumn{3}{c}{$\mu=1$} & \multicolumn{3}{c}{$\mu=1.3$}\\
\cline{3-5}\cline{6-8}\cline{9-11}
$\sigma_\eps$ & $K/P$ & $\hat\beta_P-\mu$ & Approx. & $R$ & $\hat\beta_P-\mu$ & Approx. & $R$ & $\hat\beta_P-\mu$ & Approx. & $R$ \\
\hline
0.1 & 1 & -0.007 & -0.007 & 0.000 & -0.010 & -0.010 & -0.000 & -0.013 & -0.013 & -0.000 \\
 & 2 & 0.035 & 0.038 & -0.003 & -0.007 & -0.006 & -0.001 & -0.111 & -0.108 & -0.003 \\
 & 5 & 0.081 & 0.072 & 0.010 & -0.005 & -0.004 & -0.001 & -0.180 & -0.171 & -0.009 \\
 & 10 & 0.118 & 0.088 & 0.030 & -0.003 & -0.002 & -0.001 & -0.217 & -0.205 & -0.012 \\
 & 100 & 0.225 & 0.126 & 0.099 & -0.001 & -0.000 & -0.000 & -0.278 & -0.270 & -0.008 \\
 & 1000 & 0.281 & 0.199 & 0.082 & -0.000 & -0.000 & -0.000 & -0.296 & -0.293 & -0.002 \\
\addlinespace
0.3 & 1 & -0.058 & -0.058 & 0.000 & -0.083 & -0.083 & -0.000 & -0.107 & -0.107 & 0.000 \\
 & 2 & -0.009 & -0.002 & -0.006 & -0.059 & -0.053 & -0.005 & -0.164 & -0.154 & -0.010 \\
 & 5 & 0.044 & 0.039 & 0.006 & -0.039 & -0.032 & -0.007 & -0.209 & -0.191 & -0.018 \\
 & 10 & 0.086 & 0.059 & 0.028 & -0.028 & -0.021 & -0.007 & -0.235 & -0.214 & -0.021 \\
 & 100 & 0.210 & 0.094 & 0.115 & -0.008 & -0.004 & -0.003 & -0.282 & -0.268 & -0.014 \\
 & 1000 & 0.276 & 0.166 & 0.110 & -0.001 & -0.001 & -0.001 & -0.296 & -0.292 & -0.004 \\
\addlinespace
0.5 & 1 & -0.140 & -0.140 & 0.000 & -0.200 & -0.200 & 0.000 & -0.260 & -0.260 & -0.000 \\
 & 2 & -0.081 & -0.070 & -0.011 & -0.147 & -0.134 & -0.013 & -0.256 & -0.235 & -0.021 \\
 & 5 & -0.019 & -0.017 & -0.002 & -0.102 & -0.082 & -0.020 & -0.262 & -0.226 & -0.036 \\
 & 10 & 0.030 & 0.008 & 0.023 & -0.075 & -0.055 & -0.020 & -0.269 & -0.230 & -0.039 \\
 & 100 & 0.180 & 0.029 & 0.151 & -0.022 & -0.012 & -0.010 & -0.289 & -0.265 & -0.024 \\
 & 1000 & 0.267 & 0.089 & 0.178 & -0.005 & -0.002 & -0.003 & -0.297 & -0.291 & -0.007 \\
\addlinespace
1 & 1 & -0.350 & -0.350 & 0.000 & -0.500 & -0.500 & 0.000 & -0.650 & -0.650 & 0.000 \\
 & 2 & -0.278 & -0.254 & -0.025 & -0.400 & -0.370 & -0.030 & -0.538 & -0.490 & -0.048 \\
 & 5 & -0.204 & -0.176 & -0.029 & -0.307 & -0.246 & -0.061 & -0.452 & -0.351 & -0.101 \\
 & 10 & -0.146 & -0.146 & 0.000 & -0.244 & -0.173 & -0.071 & -0.406 & -0.283 & -0.123 \\
 & 100 & 0.053 & -0.284 & 0.337 & -0.089 & -0.040 & -0.049 & -0.326 & -0.231 & -0.094 \\
 & 1000 & 0.208 & -0.683 & 0.891 & -0.021 & -0.007 & -0.014 & -0.305 & -0.271 & -0.034 \\
\hline\hline
\multicolumn{11}{l}{\parbox{1.35\textwidth}{\textit{Note}: aggregation into $P$ units of heterogeneous size as in the baseline Monte Carlo design; $\sigma_\nu^2=\mathrm{Var}(\log y_i)=1$, $\phi=11.5$. At $K/P=1$ (no aggregation) $\delta_P=\rho_P=0$ and the bias is pure attenuation, $-\mu\kappa_P$, so $R\approx0$ by construction up to $O(\sigma_\eps^2)$ terms.}}\\
\end{tabular}
}
\end{table}

%% file: tables/table_kappa.tex
\begin{scriptsize}
\begin{longtable}{lllcc}
\caption{Selected NLS point estimates for the Nonlinear Model (NLM) in Eq.~(\ref{eq:lightOnIncomeQuadratic}), by country and administrative level, corresponding to Figure~\ref{fig:EstimatedQuadraticEffect_countries}. The model is estimated on the available 2012--2019 panel, except Kenya which starts in 2013, and includes time-varying intercepts $\alpha_t$; the intercept and year effects are omitted from the table to focus on the cross-country and cross-scale comparison of $\hat{\beta}$, $\hat{\ell}_0$, and $\hat{\gamma}$. Significance: $^{*}\,10\%$, $^{**}\,5\%$, $^{***}\,1\%$.}\label{tab:kappaEstimates}\\
\toprule
Country & Level & Parameter & Estimate & Std. Error \\
\midrule
\endfirsthead
\multicolumn{5}{l}{\tablename\ \thetable{} -- continued from previous page}\\
\toprule
Country & Level & Parameter & Estimate & Std. Error \\
\midrule
\endhead
\midrule
\multicolumn{5}{r}{Continued on next page}\\
\endfoot
\bottomrule
\endlastfoot
Brazil & MUN & $\beta$ & 0.940$^{***}$ & 0.003 \\
 &  & $\ell_0$ & 0.027$^{***}$ & 0.001 \\
 &  & $\gamma$ & -0.065$^{***}$ & 0.003 \\
 & MICRO & $\beta$ & 1.089$^{***}$ & 0.006 \\
 &  & $\ell_0$ & 0.009$^{***}$ & 0.001 \\
 &  & $\gamma$ & 0.077$^{***}$ & 0.008 \\
 & MESO & $\beta$ & 1.138$^{***}$ & 0.013 \\
 &  & $\ell_0$ & 0.006$^{***}$ & 0.001 \\
 &  & $\gamma$ & 0.126$^{***}$ & 0.017 \\
 & FED & $\beta$ & 1.299$^{***}$ & 0.022 \\
 &  & $\ell_0$ & 0.089$^{***}$ & 0.011 \\
 &  & $\gamma$ & 0.188$^{***}$ & 0.016 \\
Indonesia & KAB/KOTA & $\beta$ & 0.664$^{***}$ & 0.008 \\
 &  & $\ell_0$ & 0.008$^{***}$ & 0.001 \\
 &  & $\gamma$ & -0.367$^{***}$ & 0.010 \\
 & PROV & $\beta$ & 0.988$^{***}$ & 0.027 \\
 &  & $\ell_0$ & 0.101$^{***}$ & 0.017 \\
 &  & $\gamma$ & -0.193$^{***}$ & 0.023 \\
Italy & MUN & $\beta$ & 1.189$^{***}$ & 0.004 \\
 &  & $\ell_0$ & 1.182$^{***}$ & 0.026 \\
 &  & $\gamma$ & -0.033$^{***}$ & 0.002 \\
 & LLA & $\beta$ & 1.312$^{***}$ & 0.019 \\
 &  & $\ell_0$ & 2.867$^{***}$ & 0.188 \\
 &  & $\gamma$ & 0.061$^{***}$ & 0.007 \\
 & NUTS3 & $\beta$ & 1.538$^{***}$ & 0.053 \\
 &  & $\ell_0$ & 7.202$^{***}$ & 0.881 \\
 &  & $\gamma$ & -0.031 & 0.019 \\
 & NUTS2 & $\beta$ & 3.789$^{**}$ & 1.593 \\
 &  & $\ell_0$ & 36.889 & 22.308 \\
 &  & $\gamma$ & -0.052 & 0.047 \\
Kenya & COUNTY & $\beta$ & 0.466$^{***}$ & 0.017 \\
 &  & $\ell_0$ & 0.001 & 0.001 \\
 &  & $\gamma$ & -0.780$^{***}$ & 0.026 \\
USA & COUNTY & $\beta$ & 1.104$^{***}$ & 0.004 \\
 &  & $\ell_0$ & 0.137$^{***}$ & 0.004 \\
 &  & $\gamma$ & -0.113$^{***}$ & 0.005 \\
 & CZ & $\beta$ & 1.066$^{***}$ & 0.008 \\
 &  & $\ell_0$ & 0.049$^{***}$ & 0.005 \\
 &  & $\gamma$ & -0.181$^{***}$ & 0.011 \\
 & STATE & $\beta$ & 0.896$^{***}$ & 0.044 \\
 &  & $\ell_0$ & -0.101 & 0.081 \\
 &  & $\gamma$ & -0.313$^{***}$ & 0.028 \\
\end{longtable}
\end{scriptsize}

%% file: tables/table_cvCrossScale.tex
\begin{tabular}{lllcccccccc}
  \hline
Country & Jump & \shortstack{Estimate\\$\rightarrow$ predict} & \shortstack{$R^2_{\text{oos}}$\\naive (BM)} & \shortstack{$R^2_{\text{oos}}$\\naive (NLM)} & \shortstack{Gain\\naive} & \shortstack{$R^2_{\text{oos}}$\\demean (BM)} & \shortstack{$R^2_{\text{oos}}$\\demean (NLM)} & \shortstack{Gain\\demean} & \shortstack{Level bias\\(BM)} & \shortstack{Level bias\\(NLM)} \\ 
  \hline
Brazil & small & MICRO $\rightarrow$ MUN & 0.819 & 0.831 & 0.012 & 0.842 & 0.849 & 0.007 & 0.255 & 0.223 \\ 
  Brazil & large & FED $\rightarrow$ MUN & 0.580 & 0.469 & -0.112 & 0.819 & 0.817 & -0.002 & 0.812 & 0.982 \\ 
  Indonesia & small & PROV $\rightarrow$ KAB/KOTA & 0.820 & 0.820 & -0.001 & 0.880 & 0.896 & 0.017 & -0.525 & -0.597 \\ 
  Italy & small & LLA $\rightarrow$ MUN & 0.848 & 0.854 & 0.006 & 0.851 & 0.855 & 0.005 & 0.081 & 0.060 \\ 
  Italy & large & NUTS2 $\rightarrow$ MUN & 0.751 & 0.248 & -0.502 & 0.846 & 0.697 & -0.149 & -0.457 & -0.994 \\ 
  USA & small & CZ $\rightarrow$ COUNTY & 0.854 & 0.857 & 0.004 & 0.883 & 0.892 & 0.010 & -0.297 & -0.326 \\ 
  USA & large & STATE $\rightarrow$ COUNTY & 0.076 & 0.076 & -0.000 & 0.877 & 0.877 & 0.000 & -1.559 & -1.559 \\ 
   \hline
\end{tabular}

%% file: tables/table_cvTemporal.tex
\begin{tabular}{llccccccc}
  \hline
Country & Level & \shortstack{$R^2_{\text{oos}}$\\(BM)} & \shortstack{$R^2_{\text{oos}}$\\(NLM)} & \shortstack{Gain\\NLM--BM} & \shortstack{RMSE\\(BM)} & \shortstack{RMSE\\(NLM)} & \shortstack{Max.\ yearly\\error (BM)} & \shortstack{Max.\ yearly\\error (NLM)} \\ 
  \hline
Brazil & MUN & 0.856 & 0.862 & 0.006 & 0.631 & 0.617 & 0.697 & 0.681 \\ 
  Indonesia & KAB/KOTA & 0.910 & 0.912 & 0.002 & 0.647 & 0.640 & 0.726 & 0.708 \\ 
  Italy & MUN & 0.853 & 0.865 & 0.012 & 0.569 & 0.545 & 0.573 & 0.548 \\ 
  Kenya & COUNTY & 0.961 & 0.961 & 0.000 & 0.384 & 0.383 & 0.416 & 0.414 \\ 
  USA & COUNTY & 0.885 & 0.897 & 0.012 & 0.592 & 0.560 & 0.619 & 0.590 \\ 
   \hline
\end{tabular}

%% file: tables/table_pooledFinestPredictionBalanced.tex
\begin{tabular}{llcccccc}
  \hline
Calibration & Model & Mean error & MAE & RMSE & P10 & Median & P90 \\
  \hline
Country-year FE & BM & 0.000 & 0.445 & 0.595 & -0.694 & -0.009 & 0.705 \\
  Country-year FE & NLM & 0.000 & 0.428 & 0.574 & -0.670 & 0.000 & 0.671 \\
  Aggregate country-year intercept & BM & -0.563 & 0.671 & 0.820 & -1.261 & -0.573 & 0.146 \\
  Aggregate country-year intercept & NLM & -1.011 & 1.037 & 1.168 & -1.701 & -1.005 & -0.332 \\
  Aggregate country intercept & BM & -0.621 & 0.717 & 0.866 & -1.328 & -0.629 & 0.097 \\
  Aggregate country intercept & NLM & -1.166 & 1.184 & 1.308 & -1.865 & -1.159 & -0.480 \\
  Common intercept & BM & -0.000 & 0.476 & 0.638 & -0.754 & 0.003 & 0.741 \\
  Common intercept & NLM & -0.000 & 0.469 & 0.628 & -0.742 & 0.010 & 0.728 \\
   \hline
\end{tabular}

%% file: tables/table_pooledFinestAggregateAlpha.tex
\begin{tabular}{ccc}
  \hline
Country & $\hat{\alpha}$ (BM) & $\hat{\alpha}$ (NLM) \\ 
  \hline
Brazil & 5.902 & 6.562 \\ 
  Italy & 5.727 & 6.088 \\ 
  USA & 6.529 & 7.117 \\ 
   \hline
\end{tabular}

%% file: tables/table_muCorrected.tex
\begin{table}[!htbp]\centering
\small
\caption{Aggregation-corrected elementary elasticity under negligible measurement error ($\kappa_\infty=0$). For each country and administrative level, $\hat\beta^{\mathrm{BM}}$ is the Baseline-Model slope of Equation~\eqref{eq:lightOnIncome} (log activity density on log light density with an area control and year intercepts), i.e.\ the specification that corrects for aggregation \emph{by design}; $\hat\beta^{\mathrm{biased}}$ is the slope of the naive totals regression, $\log Y_{pt}$ on $\log NTL_{pt}$ with year intercepts only, i.e.\ the specification to which Theorem~\ref{teo:aggregation} applies. Standard errors are clustered by spatial unit. The corrected elasticity $\hat\mu$ solves $\tfrac{\rho_P}{2}\mu^2+(1-\delta_P-\tfrac{\rho_P}{2})\mu+\delta_P-\hat\beta^{\mathrm{biased}}=0$ (root closest to $\hat\beta^{\mathrm{biased}}$), i.e.\ Equation~\eqref{eq:mcBiasApprox} with $\kappa_\infty=0$, using the cell-based $\delta_P$ and $\rho_P$ of Table~\ref{tab:deltaRhoEmpirics}. Brackets report the 95\% interval obtained by mapping the confidence bounds of $\hat\beta^{\mathrm{biased}}$ through the same inversion, holding $\delta_P$ and $\rho_P$ fixed.}
\label{tab:muCorrected}
\begin{tabular}{ll cc cc}
\hline\hline
Country & Level & $\hat\beta^{\mathrm{BM}}$ (s.e.) & $\hat\beta^{\mathrm{biased}}$ (s.e.) & $\hat\mu$ & 95\% interval \\
\hline
Brazil        & MUN      & 0.87 (0.01) & 0.89 (0.01) & 0.76 & [0.73, 0.80] \\
              & MICRO    & 1.06 (0.02) & 1.05 (0.02) & 1.17 & [1.07, 1.28] \\
              & MESO     & 1.10 (0.04) & 1.10 (0.04) & 0.95 & [0.89, 0.99] \\
              & STATE    & 1.13 (0.05) & 1.13 (0.05) & --   & --           \\
\addlinespace
Italy         & MUN      & 1.01 (0.01) & 1.00 (0.01) & 1.00 & [0.98, 1.01] \\
              & LLA      & 0.99 (0.02) & 1.01 (0.02) & 1.02 & [0.93, 1.10] \\
              & NUTS3    & 1.07 (0.06) & 1.03 (0.05) & 1.14 & --           \\
              & NUTS2    & 1.06 (0.16) & 1.03 (0.06) & --   & --           \\
\addlinespace
United States & COUNTY   & 0.96 (0.01) & 0.96 (0.01) & 0.91 & [0.87, 0.96] \\
              & CZ       & 0.98 (0.02) & 0.96 (0.02) & 0.90 & [0.79, 1.00] \\
              & STATE    & 0.92 (0.09) & 0.76 (0.11) & --   & --           \\
\addlinespace
Indonesia     & KAB/KOTA & 0.62 (0.02) & 0.62 (0.02) & 0.70 & [0.65, 0.74] \\
              & PROV     & 0.80 (0.05) & 0.80 (0.05) & 0.91 & [0.85, 0.96] \\
\addlinespace
Kenya         & COUNTY   & 0.46 (0.05) & 0.52 (0.05) & --   & --           \\
\hline\hline
\multicolumn{6}{l}{\parbox{0.92\textwidth}{\textit{Note}: ``--'' denotes an inversion outside the domain of the local approximation, where no real root of the quadratic exists (Brazilian and U.S.\ states, Italian NUTS2, Kenyan counties). For Italian NUTS3 the point inversion exists but the mapped confidence interval is partly undefined. Failures concentrate exactly where the Monte Carlo verification (Figure~\ref{fig:mcBiasApproxVerification}) shows that the remainder of Theorem~\ref{teo:aggregation} becomes large: coarse aggregation and, for Kenya, a naive slope below $\delta_P$, which no $\mu\geq 0$ can generate under $\kappa_\infty=0$ --- evidence against negligible measurement error in that sample.}}\\
\end{tabular}
\end{table}

%% file: tables/table_bm_full_robust.tex
\begin{scriptsize}
\begin{longtable}{lllccc}
\caption{Complete Baseline Model estimates with $95\%$ robust confidence intervals.}\label{tab:bm_full_robust}\\
\hline\hline
Country & Level & Parameter & Estimate & $95\%$ CI (unit) & $95\%$ CI (higher geo.) \\
\hline
\endfirsthead
\hline\hline
Country & Level & Parameter & Estimate & $95\%$ CI (unit) & $95\%$ CI (higher geo.) \\
\hline
\endhead
Brazil & MUN & $\alpha_{2012}$ & 5.569 & [5.475, 5.662] & [5.057, 6.080] \\
Brazil & MUN & $d_{2013}$ & $-$0.127 & [$-$0.134, $-$0.120] & [$-$0.170, $-$0.084] \\
Brazil & MUN & $d_{2014}$ & $-$0.225 & [$-$0.233, $-$0.217] & [$-$0.274, $-$0.175] \\
Brazil & MUN & $d_{2015}$ & $-$0.319 & [$-$0.328, $-$0.311] & [$-$0.369, $-$0.270] \\
Brazil & MUN & $d_{2016}$ & $-$0.399 & [$-$0.409, $-$0.389] & [$-$0.456, $-$0.342] \\
Brazil & MUN & $d_{2017}$ & $-$0.445 & [$-$0.456, $-$0.435] & [$-$0.512, $-$0.379] \\
Brazil & MUN & $d_{2018}$ & $-$0.450 & [$-$0.461, $-$0.439] & [$-$0.521, $-$0.380] \\
Brazil & MUN & $d_{2019}$ & $-$0.403 & [$-$0.414, $-$0.391] & [$-$0.478, $-$0.328] \\
Brazil & MUN & $\log$(NTL\_km2) & 0.872 & [0.857, 0.887] & [0.795, 0.949] \\
Brazil & MUN & $\log$(Km2) & $-$0.025 & [$-$0.040, $-$0.009] & [$-$0.107, 0.057] \\
Brazil & MICRO & $\alpha_{2012}$ & 4.695 & [4.254, 5.136] & -- \\
Brazil & MICRO & $d_{2013}$ & $-$0.152 & [$-$0.168, $-$0.135] & -- \\
Brazil & MICRO & $d_{2014}$ & $-$0.250 & [$-$0.270, $-$0.230] & -- \\
Brazil & MICRO & $d_{2015}$ & $-$0.355 & [$-$0.375, $-$0.334] & -- \\
Brazil & MICRO & $d_{2016}$ & $-$0.466 & [$-$0.493, $-$0.438] & -- \\
Brazil & MICRO & $d_{2017}$ & $-$0.496 & [$-$0.523, $-$0.470] & -- \\
Brazil & MICRO & $d_{2018}$ & $-$0.490 & [$-$0.515, $-$0.465] & -- \\
Brazil & MICRO & $d_{2019}$ & $-$0.453 & [$-$0.480, $-$0.426] & -- \\
Brazil & MICRO & $\log$(NTL\_km2) & 1.056 & [1.022, 1.090] & -- \\
Brazil & MICRO & $\log$(Km2) & 0.089 & [0.039, 0.140] & -- \\
Brazil & MESO & $\alpha_{2012}$ & 4.311 & [3.090, 5.531] & -- \\
Brazil & MESO & $d_{2013}$ & $-$0.150 & [$-$0.174, $-$0.125] & -- \\
Brazil & MESO & $d_{2014}$ & $-$0.243 & [$-$0.272, $-$0.215] & -- \\
Brazil & MESO & $d_{2015}$ & $-$0.361 & [$-$0.396, $-$0.327] & -- \\
Brazil & MESO & $d_{2016}$ & $-$0.479 & [$-$0.523, $-$0.434] & -- \\
Brazil & MESO & $d_{2017}$ & $-$0.513 & [$-$0.557, $-$0.469] & -- \\
Brazil & MESO & $d_{2018}$ & $-$0.488 & [$-$0.530, $-$0.446] & -- \\
Brazil & MESO & $d_{2019}$ & $-$0.462 & [$-$0.502, $-$0.422] & -- \\
Brazil & MESO & $\log$(NTL\_km2) & 1.103 & [1.016, 1.190] & -- \\
Brazil & MESO & $\log$(Km2) & 0.117 & [0.000, 0.233] & -- \\
Brazil & FED & $\alpha_{2012}$ & 3.794 & [2.129, 5.460] & -- \\
Brazil & FED & $d_{2013}$ & $-$0.148 & [$-$0.171, $-$0.124] & -- \\
Brazil & FED & $d_{2014}$ & $-$0.240 & [$-$0.281, $-$0.199] & -- \\
Brazil & FED & $d_{2015}$ & $-$0.370 & [$-$0.424, $-$0.315] & -- \\
Brazil & FED & $d_{2016}$ & $-$0.482 & [$-$0.550, $-$0.415] & -- \\
Brazil & FED & $d_{2017}$ & $-$0.516 & [$-$0.580, $-$0.452] & -- \\
Brazil & FED & $d_{2018}$ & $-$0.494 & [$-$0.562, $-$0.426] & -- \\
Brazil & FED & $d_{2019}$ & $-$0.458 & [$-$0.526, $-$0.389] & -- \\
Brazil & FED & $\log$(NTL\_km2) & 1.129 & [1.025, 1.234] & -- \\
Brazil & FED & $\log$(Km2) & 0.149 & [0.013, 0.285] & -- \\
Italy & MUN & $\alpha_{2012}$ & 5.006 & [4.941, 5.070] & [4.748, 5.263] \\
Italy & MUN & $d_{2013}$ & 0.003 & [0.001, 0.006] & [$-$0.004, 0.010] \\
Italy & MUN & $d_{2014}$ & $-$0.032 & [$-$0.035, $-$0.029] & [$-$0.045, $-$0.019] \\
Italy & MUN & $d_{2015}$ & $-$0.011 & [$-$0.015, $-$0.008] & [$-$0.021, $-$0.002] \\
Italy & MUN & $d_{2016}$ & 0.060 & [0.056, 0.064] & [0.049, 0.071] \\
Italy & MUN & $d_{2017}$ & $-$0.008 & [$-$0.013, $-$0.003] & [$-$0.022, 0.006] \\
Italy & MUN & $d_{2018}$ & 0.077 & [0.073, 0.082] & [0.065, 0.090] \\
Italy & MUN & $d_{2019}$ & 0.133 & [0.127, 0.138] & [0.116, 0.149] \\
Italy & MUN & $\log$(NTL\_km2) & 1.010 & [0.998, 1.023] & [0.944, 1.077] \\
Italy & MUN & $\log$(Km2) & $-$0.035 & [$-$0.048, $-$0.021] & [$-$0.073, 0.003] \\
Italy & LLA & $\alpha_{2012}$ & 4.724 & [4.337, 5.110] & -- \\
Italy & LLA & $d_{2013}$ & $-$0.003 & [$-$0.007, 0.002] & -- \\
Italy & LLA & $d_{2014}$ & $-$0.015 & [$-$0.021, $-$0.009] & -- \\
Italy & LLA & $d_{2015}$ & 0.001 & [$-$0.005, 0.008] & -- \\
Italy & LLA & $d_{2016}$ & 0.068 & [0.060, 0.077] & -- \\
Italy & LLA & $d_{2017}$ & 0.002 & [$-$0.007, 0.011] & -- \\
Italy & LLA & $d_{2018}$ & 0.096 & [0.087, 0.105] & -- \\
Italy & LLA & $d_{2019}$ & 0.155 & [0.145, 0.166] & -- \\
Italy & LLA & $\log$(NTL\_km2) & 0.990 & [0.948, 1.032] & -- \\
Italy & LLA & $\log$(Km2) & 0.043 & [$-$0.017, 0.102] & -- \\
Italy & NUTS3 & $\alpha_{2012}$ & 5.214 & [4.161, 6.266] & -- \\
Italy & NUTS3 & $d_{2013}$ & $-$0.001 & [$-$0.007, 0.004] & -- \\
Italy & NUTS3 & $d_{2014}$ & 0.001 & [$-$0.008, 0.010] & -- \\
Italy & NUTS3 & $d_{2015}$ & 0.014 & [0.005, 0.023] & -- \\
Italy & NUTS3 & $d_{2016}$ & 0.085 & [0.072, 0.097] & -- \\
Italy & NUTS3 & $d_{2017}$ & 0.036 & [0.022, 0.049] & -- \\
Italy & NUTS3 & $d_{2018}$ & 0.113 & [0.099, 0.127] & -- \\
Italy & NUTS3 & $d_{2019}$ & 0.170 & [0.153, 0.186] & -- \\
Italy & NUTS3 & $\log$(NTL\_km2) & 1.062 & [0.940, 1.184] & -- \\
Italy & NUTS3 & $\log$(Km2) & $-$0.044 & [$-$0.154, 0.066] & -- \\
Italy & NUTS2 & $\alpha_{2012}$ & 5.694 & [3.081, 8.307] & -- \\
Italy & NUTS2 & $d_{2013}$ & $-$0.001 & [$-$0.012, 0.010] & -- \\
Italy & NUTS2 & $d_{2014}$ & $-$0.004 & [$-$0.024, 0.015] & -- \\
Italy & NUTS2 & $d_{2015}$ & 0.014 & [$-$0.002, 0.029] & -- \\
Italy & NUTS2 & $d_{2016}$ & 0.082 & [0.062, 0.102] & -- \\
Italy & NUTS2 & $d_{2017}$ & 0.046 & [0.019, 0.073] & -- \\
Italy & NUTS2 & $d_{2018}$ & 0.119 & [0.095, 0.142] & -- \\
Italy & NUTS2 & $d_{2019}$ & 0.172 & [0.144, 0.200] & -- \\
Italy & NUTS2 & $\log$(NTL\_km2) & 0.906 & [0.370, 1.442] & -- \\
Italy & NUTS2 & $\log$(Km2) & $-$0.034 & [$-$0.366, 0.298] & -- \\
USA & COUNTY & $\alpha_{2012}$ & 6.205 & [5.682, 6.727] & [4.489, 7.921] \\
USA & COUNTY & $d_{2013}$ & 0.031 & [0.025, 0.036] & [0.010, 0.051] \\
USA & COUNTY & $d_{2014}$ & 0.007 & [0.001, 0.014] & [$-$0.013, 0.028] \\
USA & COUNTY & $d_{2015}$ & 0.072 & [0.064, 0.080] & [0.029, 0.116] \\
USA & COUNTY & $d_{2016}$ & 0.091 & [0.082, 0.100] & [0.042, 0.140] \\
USA & COUNTY & $d_{2017}$ & 0.117 & [0.107, 0.126] & [0.069, 0.164] \\
USA & COUNTY & $d_{2018}$ & 0.131 & [0.120, 0.141] & [0.077, 0.184] \\
USA & COUNTY & $d_{2019}$ & 0.174 & [0.162, 0.185] & [0.101, 0.246] \\
USA & COUNTY & $\log$(NTL\_km2) & 0.956 & [0.934, 0.977] & [0.895, 1.016] \\
USA & COUNTY & $\log$(Km2) & $-$0.074 & [$-$0.143, $-$0.005] & [$-$0.298, 0.151] \\
USA & CZ & $\alpha_{2012}$ & 7.209 & [6.049, 8.369] & -- \\
USA & CZ & $d_{2013}$ & 0.026 & [0.013, 0.040] & -- \\
USA & CZ & $d_{2014}$ & 0.019 & [0.006, 0.033] & -- \\
USA & CZ & $d_{2015}$ & 0.062 & [0.046, 0.079] & -- \\
USA & CZ & $d_{2016}$ & 0.085 & [0.064, 0.105] & -- \\
USA & CZ & $d_{2017}$ & 0.108 & [0.086, 0.131] & -- \\
USA & CZ & $d_{2018}$ & 0.121 & [0.095, 0.146] & -- \\
USA & CZ & $d_{2019}$ & 0.138 & [0.115, 0.162] & -- \\
USA & CZ & $\log$(NTL\_km2) & 0.983 & [0.952, 1.015] & -- \\
USA & CZ & $\log$(Km2) & $-$0.170 & [$-$0.299, $-$0.040] & -- \\
USA & STATE & $\alpha_{2012}$ & 9.490 & [6.647, 12.333] & -- \\
USA & STATE & $d_{2013}$ & $-$0.007 & [$-$0.025, 0.012] & -- \\
USA & STATE & $d_{2014}$ & 0.019 & [$-$0.001, 0.039] & -- \\
USA & STATE & $d_{2015}$ & 0.079 & [0.063, 0.095] & -- \\
USA & STATE & $d_{2016}$ & 0.103 & [0.081, 0.125] & -- \\
USA & STATE & $d_{2017}$ & 0.119 & [0.097, 0.140] & -- \\
USA & STATE & $d_{2018}$ & 0.134 & [0.107, 0.161] & -- \\
USA & STATE & $d_{2019}$ & 0.155 & [0.118, 0.191] & -- \\
USA & STATE & $\log$(NTL\_km2) & 0.935 & [0.775, 1.094] & -- \\
USA & STATE & $\log$(Km2) & $-$0.301 & [$-$0.537, $-$0.065] & -- \\
Kenya & COUNTY & $\alpha_{2013}$ & 13.778 & [12.801, 14.754] & -- \\
Kenya & COUNTY & $d_{2014}$ & 0.019 & [$-$0.024, 0.063] & -- \\
Kenya & COUNTY & $d_{2015}$ & 0.019 & [$-$0.045, 0.082] & -- \\
Kenya & COUNTY & $d_{2016}$ & 0.026 & [$-$0.063, 0.116] & -- \\
Kenya & COUNTY & $d_{2017}$ & $-$0.053 & [$-$0.173, 0.068] & -- \\
Kenya & COUNTY & $d_{2018}$ & $-$0.059 & [$-$0.189, 0.071] & -- \\
Kenya & COUNTY & $d_{2019}$ & $-$0.011 & [$-$0.152, 0.131] & -- \\
Kenya & COUNTY & $\log$(NTL\_km2) & 0.459 & [0.366, 0.552] & -- \\
Kenya & COUNTY & $\log$(Km2) & $-$0.772 & [$-$0.906, $-$0.638] & -- \\
Indonesia & KAB/KOTA & $\alpha_{2012}$ & 10.520 & [10.089, 10.951] & [9.862, 11.178] \\
Indonesia & KAB/KOTA & $d_{2013}$ & 0.007 & [$-$0.022, 0.036] & [$-$0.071, 0.085] \\
Indonesia & KAB/KOTA & $d_{2014}$ & 0.007 & [$-$0.017, 0.031] & [$-$0.018, 0.032] \\
Indonesia & KAB/KOTA & $d_{2015}$ & 0.019 & [$-$0.018, 0.056] & [$-$0.096, 0.134] \\
Indonesia & KAB/KOTA & $d_{2016}$ & 0.061 & [0.031, 0.091] & [0.007, 0.116] \\
Indonesia & KAB/KOTA & $d_{2017}$ & 0.068 & [0.037, 0.100] & [0.014, 0.123] \\
Indonesia & KAB/KOTA & $d_{2018}$ & 0.038 & [0.005, 0.072] & [$-$0.026, 0.103] \\
Indonesia & KAB/KOTA & $d_{2019}$ & 0.047 & [0.009, 0.086] & [$-$0.015, 0.109] \\
Indonesia & KAB/KOTA & $\log$(NTL\_km2) & 0.611 & [0.575, 0.647] & [0.553, 0.669] \\
Indonesia & KAB/KOTA & $\log$(Km2) & $-$0.402 & [$-$0.462, $-$0.343] & [$-$0.497, $-$0.308] \\
Indonesia & PROV & $\alpha_{2012}$ & 10.002 & [7.918, 12.087] & -- \\
Indonesia & PROV & $d_{2013}$ & 0.000 & [$-$0.063, 0.064] & -- \\
Indonesia & PROV & $d_{2014}$ & $-$0.017 & [$-$0.057, 0.023] & -- \\
Indonesia & PROV & $d_{2015}$ & $-$0.063 & [$-$0.119, $-$0.006] & -- \\
Indonesia & PROV & $d_{2016}$ & 0.013 & [$-$0.055, 0.081] & -- \\
Indonesia & PROV & $d_{2017}$ & $-$0.003 & [$-$0.083, 0.077] & -- \\
Indonesia & PROV & $d_{2018}$ & $-$0.039 & [$-$0.112, 0.035] & -- \\
Indonesia & PROV & $d_{2019}$ & $-$0.085 & [$-$0.183, 0.013] & -- \\
Indonesia & PROV & $\log$(NTL\_km2) & 0.817 & [0.712, 0.923] & -- \\
Indonesia & PROV & $\log$(Km2) & $-$0.240 & [$-$0.434, $-$0.045] & -- \\
\hline\hline
\multicolumn{6}{l}{\textit{Note}: Confidence intervals are reported throughout (constructed from cluster-robust standard errors).}\\
\multicolumn{6}{l}{``Unit'' clusters by spatial unit; ``higher geo.'' clusters at the coarser spatial level where available.}\\
\multicolumn{6}{l}{The same robust inference is applied to all BM parameters: intercepts $\alpha_t$, year effects, $\beta$, and $\gamma$.}\\
\end{longtable}
\end{scriptsize}

%% file: tables/table_nlm_robust.tex
\begin{table}[!htbp]\centering
\scriptsize
\caption{Nonlinear Model estimates with $95\%$ robust confidence intervals.}
\label{tab:nlm_robust}
\begin{tabular}{lllccc}
\hline\hline
Country & Level & Parameter & Estimate & $95\%$ CI (unit) & $95\%$ CI (higher geo.) \\
\hline
Brazil & MUN & $\hat\beta$ & 0.940 & [0.922, 0.957] & [0.844, 1.035] \\
Brazil & MUN & $\hat\ell_0$ & 0.027 & [0.022, 0.033] & [0.003, 0.052] \\
Brazil & MUN & $\hat\gamma$ & $-$0.065 & [$-$0.081, $-$0.049] & [$-$0.142, 0.011] \\
Brazil & MICRO & $\hat\beta$ & 1.089 & [1.054, 1.124] & -- \\
Brazil & MICRO & $\hat\ell_0$ & 0.009 & [0.002, 0.016] & -- \\
Brazil & MICRO & $\hat\gamma$ & 0.077 & [0.029, 0.126] & -- \\
Brazil & MESO & $\hat\beta$ & 1.138 & [1.067, 1.208] & -- \\
Brazil & MESO & $\hat\ell_0$ & 0.006 & [$-$0.007, 0.019] & -- \\
Brazil & MESO & $\hat\gamma$ & 0.126 & [0.022, 0.230] & -- \\
Brazil & FED & $\hat\beta$ & 1.299 & [1.203, 1.396] & -- \\
Brazil & FED & $\hat\ell_0$ & 0.089 & [0.036, 0.143] & -- \\
Brazil & FED & $\hat\gamma$ & 0.188 & [0.085, 0.291] & -- \\
Italy & MUN & $\hat\beta$ & 1.189 & [1.167, 1.212] & [1.126, 1.253] \\
Italy & MUN & $\hat\ell_0$ & 1.182 & [0.983, 1.381] & [0.331, 2.034] \\
Italy & MUN & $\hat\gamma$ & $-$0.033 & [$-$0.045, $-$0.020] & [$-$0.068, 0.002] \\
Italy & LLA & $\hat\beta$ & 1.312 & [1.204, 1.419] & -- \\
Italy & LLA & $\hat\ell_0$ & 2.867 & [1.611, 4.124] & -- \\
Italy & LLA & $\hat\gamma$ & 0.061 & [0.004, 0.118] & -- \\
Italy & NUTS3 & $\hat\beta$ & 1.538 & [1.225, 1.851] & -- \\
Italy & NUTS3 & $\hat\ell_0$ & 7.202 & [1.131, 13.273] & -- \\
Italy & NUTS3 & $\hat\gamma$ & $-$0.031 & [$-$0.133, 0.072] & -- \\
Italy & NUTS2 & $\hat\beta$ & 3.789 & [$-$6.466, 14.044] & -- \\
Italy & NUTS2 & $\hat\ell_0$ & 36.889 & [$-$118.032, 191.809] & -- \\
Italy & NUTS2 & $\hat\gamma$ & $-$0.052 & [$-$0.355, 0.251] & -- \\
USA & COUNTY & $\hat\beta$ & 1.104 & [1.079, 1.129] & [1.048, 1.160] \\
USA & COUNTY & $\hat\ell_0$ & 0.137 & [0.106, 0.167] & [0.068, 0.205] \\
USA & COUNTY & $\hat\gamma$ & $-$0.113 & [$-$0.180, $-$0.045] & [$-$0.332, 0.107] \\
USA & CZ & $\hat\beta$ & 1.066 & [1.019, 1.113] & -- \\
USA & CZ & $\hat\ell_0$ & 0.049 & [0.019, 0.078] & -- \\
USA & CZ & $\hat\gamma$ & $-$0.181 & [$-$0.310, $-$0.052] & -- \\
USA & STATE & $\hat\beta$ & 0.896 & [0.675, 1.118] & -- \\
USA & STATE & $\hat\ell_0$ & $-$0.101 & [$-$0.338, 0.136] & -- \\
USA & STATE & $\hat\gamma$ & $-$0.313 & [$-$0.551, $-$0.076] & -- \\
Kenya & COUNTY & $\hat\beta$ & 0.466 & [0.366, 0.566] & -- \\
Kenya & COUNTY & $\hat\ell_0$ & 0.001 & [$-$0.002, 0.003] & -- \\
Kenya & COUNTY & $\hat\gamma$ & $-$0.780 & [$-$0.911, $-$0.648] & -- \\
Indonesia & KAB/KOTA & $\hat\beta$ & 0.664 & [0.612, 0.717] & [0.577, 0.752] \\
Indonesia & KAB/KOTA & $\hat\ell_0$ & 0.008 & [$-$0.000, 0.016] & [$-$0.006, 0.022] \\
Indonesia & KAB/KOTA & $\hat\gamma$ & $-$0.367 & [$-$0.431, $-$0.304] & [$-$0.467, $-$0.267] \\
Indonesia & PROV & $\hat\beta$ & 0.988 & [0.837, 1.139] & -- \\
Indonesia & PROV & $\hat\ell_0$ & 0.101 & [0.007, 0.195] & -- \\
Indonesia & PROV & $\hat\gamma$ & $-$0.193 & [$-$0.323, $-$0.063] & -- \\
\hline\hline
\multicolumn{6}{l}{\textit{Note}: NLM parameters of Eq.~(\ref{eq:lightOnIncomeQuadratic}): $\beta$ shifted-log slope, $\ell_0$ shift, $\gamma$ area coefficient.}\\
\multicolumn{6}{l}{$95\%$ CIs from cluster-robust standard errors; ``unit'' clusters by spatial unit, ``higher geo.'' at the coarser level.}\\
\end{tabular}
\end{table}

%% file: tables/table_beta_tost.tex
\begin{table}[!htbp]\centering
\small
\caption{Equivalence of the NTL elasticity to unity. Two one-sided tests (TOST) of $H_0:\beta=1$ against the equivalence range $(1-\delta,1+\delta)$, at tolerance $\delta=10\%$ and $\delta=20\%$. ``Yes'' means the $90\%$ confidence interval lies entirely inside the equivalence range, so $\beta$ is statistically indistinguishable from one up to the stated tolerance at the $5\%$ level; ``No'' means equivalence cannot be concluded. The test is reported under both unit clustering and the more conservative higher-geographic clustering (defined only at the finest level of each country).}
\label{tab:beta_tost}
\begin{tabular}{llc cc cc}
\hline\hline
 & & & \multicolumn{2}{c}{Unit cluster} & \multicolumn{2}{c}{Higher-geo.\ cluster}\\
\cline{4-5}\cline{6-7}
Country & Level & $\hat\beta$ & $\delta=10\%$ & $\delta=20\%$ & $\delta=10\%$ & $\delta=20\%$ \\
\hline
Brazil & MUN & 0.872 & No & Yes & No & Yes \\
Brazil & MICRO & 1.056 & Yes & Yes & -- & -- \\
Brazil & MESO & 1.103 & No & Yes & -- & -- \\
Brazil & FED & 1.129 & No & No & -- & -- \\
Italy & MUN & 1.010 & Yes & Yes & Yes & Yes \\
Italy & LLA & 0.990 & Yes & Yes & -- & -- \\
Italy & NUTS3 & 1.062 & No & Yes & -- & -- \\
Italy & NUTS2 & 0.906 & No & No & -- & -- \\
USA & COUNTY & 0.956 & Yes & Yes & Yes & Yes \\
USA & CZ & 0.983 & Yes & Yes & -- & -- \\
USA & STATE & 0.935 & No & Yes & -- & -- \\
Kenya & COUNTY & 0.459 & No & No & -- & -- \\
Indonesia & KAB/KOTA & 0.611 & No & No & No & No \\
Indonesia & PROV & 0.817 & No & No & -- & -- \\
\hline\hline
\multicolumn{7}{l}{\textit{Note}: TOST at the $5\%$ level. ``--'' = higher-geographic cluster not defined at that level.}\\
\end{tabular}
\end{table}